\documentclass[10pt,oneside]{amsart}
\usepackage[paperwidth=193mm,paperheight=260mm,tmargin=21mm,lmargin=19mm,rmargin=19mm,bmargin=21mm,asymmetric]{geometry}
\usepackage{amsaddr}

\usepackage[utf8]{inputenc} 
\usepackage[T1]{fontenc} 
\usepackage{lmodern} 
\usepackage[ngerman,english]{babel}

\usepackage{amsmath} 
\usepackage{amsfonts}
\usepackage{amsthm}
\usepackage{amssymb}

\usepackage{braket}     
\usepackage{bm}         

\usepackage{graphicx}   
\usepackage[table]{xcolor}
\usepackage{wrapfig}

\usepackage{tikz}
\usepackage{pgfplots}

\usetikzlibrary{calc}
\usetikzlibrary{3d}
\usetikzlibrary{patterns}

\pgfplotsset{axis lines=left, grid=major} 
\pgfplotsset{legend style={font=\small}}
\pgfplotsset{every axis/.append style={font=\small}}

\usetikzlibrary{spy}
\tikzset{new spy style/.style={spy scope={%
 magnification=5,
 size=1.25cm, 
 connect spies,
 every spy on node/.style={
   rectangle,
   draw,
   },
 every spy in node/.style={
   draw,
   rectangle,
   fill=gray!20,
   }
  }
 }
}

\usepackage{units} 
\usepackage{gensymb} 

\usepackage{url} 
\usepackage[nocompress]{cite}

\usepackage{rotating}
\usepackage{array}	
\usepackage{booktabs}	
\usepackage{tabularx}
\newcolumntype{L}{>{\raggedright\arraybackslash}X} 
\newcolumntype{R}[1]{>{\raggedright}p{#1}} 

\newtheorem{theorem}{Theorem}[section]
\newtheorem{lemma}[theorem]{Lemma}
\newtheorem{proposition}[theorem]{Proposition}
\newtheorem{corollary}[theorem]{Corollary}

\theoremstyle{definition}
\newtheorem{definition}[theorem]{Definition}
\newtheorem{algo}[theorem]{Algorithm}

\theoremstyle{remark}
\newtheorem{example}[theorem]{Example}
\newtheorem{remark}[theorem]{Remark}

\numberwithin{equation}{section}

\makeatletter
\newcommand{\rmnum}[1]{\romannumeral #1}
\newcommand{\Rmnum}[1]{\expandafter\@slowromancap\romannumeral #1@}
\makeatother

\usepackage{enumitem}
\setlist[description]{style=unboxed} 
\setlist[enumerate]{label=(\roman*)}

\usepackage{caption} 
\newlength{\hfloatsep}
\definecolor{blue}{rgb}{0.1922,0.4196,0.5961}
\definecolor{green}{rgb}{0,0.7,0}
\definecolor{lightgray}{rgb}{0.8,0.8,0.8}

\newlength{\includewidth} 
\newcommand{\markdashed}{\scalebox{1.5}{-}\,\scalebox{1.5}{-}}   
\newcommand{\markline}{\scalebox{3.1}[1.5]{-}}           

\usepackage{pifont}
\newcommand{\cmark}{\ding{51}}
\newcommand{\xmark}{\ding{55}}

\DeclareTextFontCommand{\defemph}{\bf\em}

\renewcommand*{\phi}{\varphi}
\renewcommand*{\epsilon}{\varepsilon}
\renewcommand*{\rho}{\varrho}

\providecommand*{\vv}[1]{\ensuremath{\boldsymbol{#1}}}	
\providecommand*{\mv}[1]{\ensuremath{\boldsymbol{#1}}}	

\providecommand*{\dd}{\ensuremath{\operatorname{d}}}	
\providecommand*{\transp}{\ensuremath{{^{\intercal}}}}	

\providecommand*{\sign}{\ensuremath{\operatorname{sign}}}

\providecommand*{\ddfrac}[2]{\ensuremath{\frac{\dd #1}{\dd #2}}} 
\providecommand*{\pfrac}[2]{\ensuremath{\frac{\partial #1}{\partial #2}}} 
\providecommand*{\pslash}[2]{\ensuremath{{\partial #1}/{\partial #2}}}  	
\providecommand*{\diverg}{\ensuremath{\operatorname{div}}} 

\providecommand*{\abs}[1]{\ensuremath{\left \lvert #1 \right \rvert}}	
\RequirePackage[only,llbracket,rrbracket]{stmaryrd}
\providecommand*{\jump}[1]{\ensuremath{\left\llbracket #1 \right\rrbracket}}	
\providecommand*{\mean}[1]{\ensuremath{\left\{#1\right\}}}	%

\newcommand{\setstyle}[1]{{\mathbb #1}}
\newcommand{\setR}{\setstyle{R}}
\newcommand{\setN}{\setstyle{N}}

\newcommand{\setS}{\setstyle{S}}

\newcommand*{  \U  }{\ensuremath{ \vv{U}		}}

\newcommand*{  \W  }{\ensuremath{ \vv{W}		}}

\newcommand*{  \s  }{\ensuremath{ \mathfrak{s}		}}

\newcommand*{  \pr }{\ensuremath{ \tilde{p}		}}
							  
\newcommand*{\admisr}{\ensuremath{ \tilde{\mathcal{A}}  }}
\newcommand*{\admis }{\ensuremath{ \mathcal{A}          }}

\newcommand*{\liq}{   \ensuremath{ \text{\tiny liq}        }} 
\newcommand*{\vap}{   \ensuremath{ \text{\tiny vap}        }} 
\renewcommand*{\r}{   \ensuremath{ {\text{\tiny r}}       }} 
\renewcommand*{\l}{   \ensuremath{ {\text{\tiny l}}       }} 

\renewcommand*{\L}{   \ensuremath{ \text{\tiny L}          }} 
\newcommand*{\R}{     \ensuremath{ \text{\tiny R}          }} 

\newcommand*{\sat}{   \ensuremath{ \text{\tiny sat}        }}
\renewcommand*{\sc}{  \ensuremath{ \text{\tiny sc}         }}
\newcommand*{\se}{    \ensuremath{ \text{\tiny se}         }}
\newcommand*{\tmin}{  \ensuremath{ \text{\tiny min}        }} 
\newcommand*{\tmax}{  \ensuremath{ \text{\tiny max}        }} 

\renewcommand*{\c}{   \ensuremath{ \text{\tiny c}        }}
\newcommand*{\e}{   \ensuremath{ \text{\tiny e}        }}

\newcommand*{\surfcoeff}{  \ensuremath{\zeta^\ast}         }
\newcommand*{\surf}{       \ensuremath{\zeta}              }

\newcommand*{\surfset}{\ensuremath{\mathcal{Z}}} 

\newcommand*{\Rvel}{  \ensuremath{R}         }
\newcommand*{\Svel}{  \ensuremath{S}         }
\newcommand*{\Pvel}{  \ensuremath{P}         }
\newcommand*{\Evel}{    \ensuremath{E}       }

\newcommand*{\Kr}{    \ensuremath{\mathbb{K}}           }

\renewcommand*{\R}{   \ensuremath{ \text{\tiny R}        }}
\newcommand*{\cfl}{   \ensuremath{ \operatorname{CFL}}}

\newcommand{\intdomain}[1]{\int_{\Omega_\liq\cup\Omega_\vap} #1 \;\dd v}

\newcommand{\lrwave}[1]{${#1}^\surf_\R$}
  
\newcommand{\klwave}[1]{${#1}_\L$}
\newcommand{\krwave}[1]{${#1}_\R$}

\usetikzlibrary{external} 
\begin{document}

\title[On Riemann Solvers and Kinetic Relations]{On Riemann Solvers and Kinetic Relations for Isothermal Two-Phase Flows with Surface Tension}

\author{Christian Rohde and Christoph Zeiler}
\address{Institut f\"ur Angewandte Analysis und Numerische Simulation, \\Universit\"at Stuttgart, Pfaffenwaldring 57, 70569 Stuttgart, Germany}

\email{Christian.Rohde@mathematik.uni-stuttgart.de}

\email{Christoph.Zeiler@mathematik.uni-stuttgart.de}
\thanks{This research work is supported by the German Research Foundation (DFG) through the grant RO 2222/4-1.}




\begin{abstract}
We consider a sharp-interface approach for the inviscid isothermal dynamics of compressible two-phase flow, that accounts for phase transition and surface tension effects. To fix the mass exchange and entropy dissipation rate across the interface kinetic relations are frequently used. The complete uni-directional dynamics can then be understood by solving  generalized two-phase Riemann problems. We present new well-posedness theorems for the Riemann problem and corresponding computable Riemann solvers, that cover quite general equations of state, metastable input data and curvature effects.   \\
The new Riemann solver is used to validate different kinetic relations on physically relevant problems including a comparison with experimental data. Riemann solvers are building blocks for many numerical schemes that are used to track interfaces in two-phase flow. It is shown that the new Riemann solver enables reliable and efficient computations for physical situations that could not be treated before. 
\end{abstract}

\maketitle


\section{Introduction}

The dynamics of an isothermal homogeneous fluid that can appear in either a liquid or a vapor phase 
is governed by the compressible Euler equations for density and velocity provided that viscosity and heat conduction 
effects are neglected.  In this framework it is natural to  consider a sharp interface approach for the 
phase boundary  which results in a mathematical model in the form of a free boundary value problem. 
Let $\Omega\subset\setR^d$ with 
$d\in\setN$ be an open,  bounded set. 
For any $t\in[0,\theta]$, $\theta>0$, we assume that $\Omega$ is portioned into 
the union of two open sets $\Omega_\vap(t)$, $\Omega_\liq(t)$, which contain the two bulk phases, and 
a hypersurface $\Gamma(t)$ -- the sharp interface --, that separates the two  spatial bulk sets. 
In the  spatial-temporal bulk sets  
$\Set{ (\vv{x},t)  \in  \Omega \times (0,\theta)| \vv{x} \in \Omega_\vap(t) \cup \Omega_\liq(t)}$
we have then  the hydromechanical system  
\begin{align}\label{eq:euler}
\begin{array}{rcccl}
  \rho_t &+& \diverg (\rho\, \vv v) &=&0,\\[1.5ex]
{(\rho\,\vv v)}_t
         & +
         & \diverg{(\rho\,\vv v \otimes \vv v + \pr(\rho)\,\mv{I} )}
         & =
         & \vv{0}.
    \end{array}
\end{align}
Here, $\rho=\rho(\vv{x},t)>0$ denotes the unknown density field and 
$\vv v=\vv v(\vv{x},t) = ( v_1(\vv{x},t),\cdots, v_d(\vv{x},t))^\transp \in \setR^d$ 
the unknown velocity field. The pressure $\pr=\pr(\rho)$ is a given scalar function 
and $\mv{I}\in \setR^{d\times d}$ the $d$-dimensional unit matrix.

Besides appropriate initial and boundary conditions it remains to provide coupling conditions at the free boundary $\Gamma(t)$.  
Let some $\vv{\xi}\in \Gamma(t)$ be given. 
We denote the speed of $\Gamma(t)$ in the normal direction  $\vv{n}=\vv{n}(\vv{\xi},t)\in \setS^{d-1}$ by $\sigma=\sigma(\vv{\xi},t)\in\setR$. 
Throughout the paper the  direction of the normal vector is always chosen, such that $\vv{n}$ points into the vapor domain $\Omega_\vap$. 
Across the interface the following $d+1$ trace conditions  are posed which represent the conservation of mass and the balance of momentum in presence of capillary surface forces
(see e.g.~\cite{Batchelor99}).
\begin{align}
  \jump{\rho\, (\vv v\cdot \vv{n} - \sigma) } &= 0, \label{eq:rh:rho}\\
  \jump{\rho\, (\vv v\cdot \vv{n} - \sigma)\,\vv{v}\cdot\vv{n} + \pr(\rho) } &= (d-1)\surfcoeff\,\kappa, \label{eq:rh:m}\\   
 \jump{ \vv{v} \cdot \vv{t}^l}         &=0   \qquad (l=1,\ldots,{d-1}).   \label{eq:rh:mt}     
\end{align}
Thereby, we use $\jump{a}:=a_\vap-a_\liq$ and $a_{\vap/\liq} :=\lim_{\epsilon\to 0, \epsilon>0} a(\vv{\xi} \pm \epsilon\,\vv{n})$ 
for some quantity $a$ defined in $\Omega_\vap(t) \cup \Omega_\liq(t)$. 
In \eqref{eq:rh:m} by $\kappa = \kappa(\vv{\xi},t)\in\setR$ we denote the mean curvature
of $\Gamma(t)$ associated with orientation given through 
the choice of the normal $\vv{n}$. 
The  surface tension coefficient $\surfcoeff\ge0$ is assumed to be constant, and $\vv{t}^1,\ldots,\vv{t}^{d-1} \in \setS^{d-1}$ denote a  complete set of vectors tangential to $\vv{n}$.\\
We apply the concept of entropy solutions and seek for functions $(\rho, \vv v)$ that satisfy the entropy condition
$ 
{E(\rho,\vv v)}_t +\diverg \left((E(\rho,\vv v) + \pr(\rho))\, \vv{v} \right)      
   \leq 0 
$ 
in the distributional sense in the bulk regions and 
\begin{align}\label{eq:Euler:entropyjump}
 -\sigma\, (\jump{E(\rho,\vv v)} + (d-1)\surfcoeff\,\kappa )+ \jump{(E(\rho,\vv v)+\pr(\rho))\,\vv{v}\cdot\vv{n}}  \leq 0
\end{align}
at the interface. Here, we used 
$E(\rho,\vv v) = \rho\, \psi(1/\rho) + 1/2\,\rho\,\abs{\vv v}^2$ 
and the Helmholtz free energy $\psi$ defined below in Definition~\ref{def:thermo}.
Note that \eqref{eq:Euler:entropyjump} accounts for surface tension.

Additionally to  the coupling conditions  \eqref{eq:rh:rho}, \eqref{eq:rh:m}, \eqref{eq:rh:mt}, \eqref{eq:Euler:entropyjump}  the mass transfer across the phase boundary has to be determined. In this paper 
we rely on so-called \defemph{kinetic relations} \cite{ABEKNO1990,TRU1993}. In the most  simple case this results in an additional algebraic jump condition across $\Gamma(t)$,  which may be summarized in 
\begin{align}\label{eq:Euler:kinrel}
K(\rho_\liq, \vv v_\liq,\rho_\vap,\vv v_\vap) =0.
\end{align}
A local well-posedness result for  the  free boundary value problem   \eqref{eq:euler}-\eqref{eq:Euler:kinrel}    with a  special kinetic relation (denoted in this paper  as $K_2$, see Section~\ref{sec:kinrelex})  has been recently proposed in \cite{Kabil2016}. 
Much more analytical knowledge can be derived  if we restrict ourselves to describe the  local one-dimensional evolution in the normal direction through some  $ \vv{\xi}\in \Gamma(t)$. 
Mathematically this leads to consider a 
\defemph{generalized Riemann problem}
for a mixed-type ensemble  of conservation  laws. Note that the local curvature $\kappa(\vv{\xi},t)$ enters as a source term in the jump relation for momentum. 
We will present the precise setting  and  the corresponding thermodynamical framework in Section~\ref{sec:model}. \\

Riemann problems for two-phase flows have been intensively studied in the last two decades (see \cite{LEF2002} for a  general theory,  \cite{HANTKE2013,GODSEG2006,HAT1986,COLPRI2003,LEFTHA2002,MUEVOS2006,COLCOR1999}  for specific examples and 
\cite{ROHZER2014,CHACOQENGROH2012,SchleperHLLP,HELSEG2006} for approximate Riemann solvers). 
However, even in the isothermal case  the theory is not yet complete. It is the first major purpose of this paper to present a solution theory for generalized Riemann problems and \defemph{computable Riemann
solvers},
such that physically more  realistic scenarios can be analyzed. In particular we will follow the   concept of  \defemph{monotone decreasing kinetic functions} from  \cite{COLPRI2003} 
and  generalize it   accordingly (see Theorem~\ref{theo:kinrel} for a well-posedness theorem and Algorithm~\ref{alg:kinrelsolver} for a the Riemann solver).  Let us point out already here that  not for any relevant 
kinetic relation the concept of monotone decreasing kinetic functions applies, such that Theorem~\ref{theo:kinrel} fails. Nevertheless a solution of the Riemann problem might exist,  and possibly can still 
be computed by Algorithm~\ref{alg:kinrelsolver}.  In contrast to  previous results from the  literature the new approach governs surface tension effects, allows for  so-called metastable input states and 
can be applied to a much larger class of fluids via a general form for the equation of state. Even tabularized equations can be used. Finally we note that 
the smoothness assumptions on $K$ in \eqref{eq:Euler:kinrel} are relaxed.  This allows to consider kinetic relations which  exhibit a typical threshold behavior for entropy release. \\

In the second and third part of the paper we present then several analytic and numerical results that can be achieved by the new Riemann solver. \\
First in Section~\ref{sec:kinrel} and Section~\ref{sec:kinrelex} we review 
physically relevant kinetic relations and analyze to what extent they can be treated by the theory of monotone decreasing functions. In particular
we can classify all of them according to their entropy dissipation rate. As a by-product   it turns out that  the classical   Liu entropy criterion can be understood as a limiting case 
for the kinetic relations \cite{LIU1974}.   
A central part of our work is the comparison of exact solutions of Riemann problems for  the selected  kinetic relations. Already this theoretical approach 
shows the limitations  of several suggestions from literature.\\   
To conclude Section~\ref{sec:kinrelex}, we validate the kinetic relations  
against  data from shock tube experiments in \cite{Simdodecane1999}. 
It turns out that the use of a  kinetic relation that has been derived by  density functional theory in \cite{klink2015analysis} gives excellent agreement 
with the measured data while other choices fail.\\ 
Besides the obvious interest to understand   Riemann problems  from the analytic point of view,
the Riemann problem is essential for all numerical methods that rely on some kind  of  interface tracking (see \cite{MERROH2007,DREROH2008,FJS2013,FechDGsubcell,FZ2015pre,JAEROHZER2012,Wiebe2016}).
The tracking approach uses any finite volume or discontinuous Galerkin method as  powerful tool to solve \eqref{eq:euler} numerically in the bulk sets. Across 
the interface it requires special numerical fluxes, that can be  computed from solving the generalized Riemann problem. We show in the final third part of this contribution  that  it is possible 
to perform reliable and efficient computations for a wide variety of scenarios with the new Riemann solver. In previous works  the range of applicability was limited to very special situations. 
Furthermore, the new exact solver enables us to validate a previously developed approximate Riemann solver \cite{ROHZER2014}, which is based on relaxation techniques.\\
The results of this paper rely mainly on the PhD thesis of Christoph Zeiler \cite{Z2015}.

\section{The two-phase Riemann problem}\label{sec:model}

\subsection{Preliminaries and two-phase thermodynamics}\label{sub:thermo}

We denote the specific volume by $\tau=1/\rho$ and we fix the thermodynamic framework in terms of $\tau$. We assume that the thermodynamic framework holds for the rest of the paper.
\begin{definition}\label{def:thermo}
  Let the numbers $\tau_\liq^\tmin$, $\tau_\liq^\tmax$, $\tau_\vap^\tmin$, $\surf^\tmin$, $\surf^\tmax \in\setR$ with 
  $0 < \tau_\liq^\tmin < \tau_\liq^\tmax < \tau_\vap^\tmin$, $\surf^\tmin<0<\surf^\tmax$ and functions
  $p\in\mathcal{C}^2\left(\admis_\liq\cup\admis_\vap,\setR\right)$,
  $\psi,\mu\in\mathcal{C}^3\left(\admis_\liq\cup\admis_\vap,\setR\right)$
  be given.
  The intervals $\admis_\liq=(\tau_\liq^\tmin, \tau_\liq^\tmax)$ and $\admis_\vap=(\tau_\vap^\tmin, \infty)$ define the \defemph{liquid phase} and the \defemph{vapor phase} and $\admis_\liq\cup\admis_\vap$ is called \defemph{admissible set of specific volume values}. 
  
  A triple  of functions $p$, $\psi$, $\mu$ with    
  \begin{align}
     &p(\tau) = -\psi'(\tau) \text{ and }    
  \mu(\tau)=\psi(\tau)+p(\tau)\,\tau   \label{eq:iso:freeenergy}  
  \end{align}  
    are called 
  \defemph{pressure},
  \defemph{specific Helmholtz free energy} and 
  \defemph{specific Gibbs free energy}, respectively.\\ 
  It is supposed that  they satisfy 
  for any $\surf\in \surfset:=(\surf^\tmin,\surf^\tmax)$ the conditions
  \begin{align}
    &p'<0    \text{ in } \admis_\liq\cup\admis_\vap, 		 \label{eq:H1} \\
    &p''>0   \text{ in } \admis_\liq\cup\admis_\vap,		 \label{eq:H2} \\
    &\exists\, \tau_\liq^\sat(\surf)\in \admis_\liq, \tau_\vap^\sat(\surf)\in \admis_\vap : 
    \left\{
    \begin{array}{ccc}
        p(\tau_\vap^\sat) -  p(\tau_\liq^\sat)  &=&  \surf, \\
      \mu(\tau_\vap^\sat) -\mu(\tau_\liq^\sat)  &=& 0,
    \end{array}
    \right.
    \label{eq:H3} \\ 
    & p(\tau) \to \infty  \text{ for } \tau\to\tau_\liq^\tmin,       \label{eq:H4} \\
    & \forall \tau_\liq\in\admis_\liq, \tau_\vap\in\admis_\vap: p'(\tau_\liq) < p'(\tau_\vap),   \label{eq:H5}\\
  %
  %
    & \lim_{R\to\infty}\int_{\tau_\vap^\tmin}^R  c(\tau) \dd \tau = \infty  \text{ with } c(\tau):= \sqrt{-p'(\tau)}.    \label{eq:H6} 
  \end{align} 
\end{definition}

\begin{figure}\centering
    \input{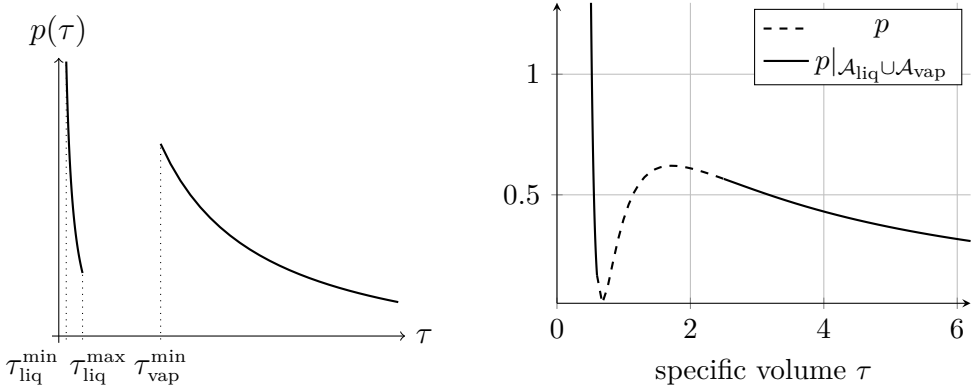}    
    \begin{tikzpicture}[scale=0.45]    
	\draw[->] (-0.2,0) -- (10.2,0) node[right] {$\tau$};
	\draw[->] (0,-0.2) -- (0,8.2) node[above] {$p(\tau)$};

	\draw[thick] plot[domain=0.22:0.7] ({\x}, {pliq(\x)-2} ); 
	\draw[thick] plot[domain=3:10] (\x, {pvap(\x)} );

	\coordinate (A) at (0.22,{pliq(0.22)-2}); 
        \coordinate (C) at (0.7 ,{pliq(0.7 )-2});      
        \coordinate (D) at (3   ,{pvap( 3  )});      
        \draw[dotted] (A|- 0,-0.1) node [below] {$\tau^\tmin_\liq$ $\tau^\tmax_\liq$} -- (A);
        \draw[dotted] (C|- 0,-0.1)  -- (C);
        \draw[dotted] (D|- 0,-0.1) node [below] {$\tau^\tmin_\vap$} -- (D);
    \end{tikzpicture}
    \hspace{\hfloatsep}
    \begin{tikzpicture}
      \begin{axis}[x post scale=0.8,y post scale=0.7,xmin=0, xlabel={specific volume $\tau$}]
 	\addplot[thick,dashed] plot[domain=0.6:2.5]  ({\x}, { pvdw(\x) } ) ;
 	
 	\addplot[thick] plot[domain=0.51:0.6]({\x}, { pvdw(\x) } ) ;
 	\addplot[thick] plot[domain=2.5:6.2]    ({\x}, { pvdw(\x) } ) ;

      \legend{ $p$,$p\vert_{\admis_\liq\cup\admis_\vap}$};
      \end{axis}
    \end{tikzpicture}    
  \caption{Left: prototypical example of a pressure function.
  Right: van der Waals pressure of Example~\ref{exp:vdw}.} \label{fig:iso:p} 
\end{figure}

Note that $p$ is monotone decreasing and convex in both phases, see Figure~\ref{fig:iso:p} (left) for some illustration. 
The interval $[\tau_\liq^\tmax, \tau_\vap^\tmin]$ is excluded from our studies as a set of unphysical states. In fact \eqref{eq:euler} becomes ill-posed for specific volumes in $[\tau_\liq^\tmax, \tau_\vap^\tmin]$ .    
The number $\surf \in \surfset$    is  arbitrary in \eqref{eq:H3}  but prescribed   through    $\surf= (d-1)\surfcoeff\,\kappa   $    in any application and thus linked to the given 
curvature $\kappa$. The limitation of $\surf$ to the interval $\surfset=(\surf^\tmin,\surf^\tmax)$ is due to the fact, that we can not expect that the saturation values in 
\eqref{eq:H3} exists  for any value $\surf \in \setR$.   With other words, our theory is restricted to interfaces with moderate curvature $\kappa$. 
The pair $\left(\tau_\liq^\sat(\surf), \tau_\vap^\sat(\surf)\right)\in \admis_\liq \times \admis_\vap$ in hypothesis \eqref{eq:H3} is called \defemph{pair of saturation states} and depends on $\surf\in\surfset$.
These states are attained in the \defemph{thermodynamic equilibrium}, i.e.,
\begin{align}\label{eq:equilibrium}
  p_\vap -  p_\liq  &=  (d-1)\surfcoeff\,\kappa, &
  \mu_\vap &= \mu_\liq &
  (\text{and } T_\liq=T_\vap).
\end{align}
%
%
%
%
The sets $(\tau_\liq^\sat, \tau_\liq^\tmax)$ and $(\tau_\vap^\tmin,\tau_\vap^\sat)$ are called \defemph{metastable liquid} and \defemph{metastable vapor phases}, while the sets $(\tau_\liq^\tmin, \tau_\liq^\sat]$, $[\tau_\vap^\sat, \infty)$ are called \defemph{stable} (liquid/vapor) phases.
Specific volume values, which belong to these sets, are called (liquid/vapor) metastable or stable states.
 
Hypotheses \eqref{eq:H2}, \eqref{eq:H4} and \eqref{eq:H5} limit the amount of  possible  wave configurations of the solution to Riemann problems.
In \eqref{eq:H4} it is assumed that there is a minimal molecular distance, where the liquid cannot be compressed further, and \eqref{eq:H5} is natural, since the sound speed in the liquid phase
of a fluid is usually much higher than in the vapor phase. Hypothesis \eqref{eq:H6} excludes the case of vacuum which is out of our interests.

Equations of state have to be determined, e.g.\ by experimental measurements.
However, for a simple model fluid, that may occur in a liquid and a vapor phase, we may consider the following explicit form, such that all conditions of Definition~\ref{def:thermo} are satisfied.
\begin{example}[Van der Waals equation of state] \label{exp:vdw}
 The \defemph{van der Waals equations of state} are given by the pressure function
 \begin{align} \label{eq:vdw:p}
   p(\tau) = \frac{R\,T}{\tau-\tau_\liq^\tmin}- \frac{a}{\tau^2}, 
 \end{align}
 with positive constants $a,\tau_\liq^\tmin,R$ for $\tau>\tau_\liq^\tmin$ and corresponding specific Helmholtz and Gibbs free energy functions according to \eqref{eq:iso:freeenergy}.
 The function is monotone decreasing for $T\geq T_\c$, where $T_\c=8\,a/(27\,R\,\tau_\liq^\tmin)$ is the \defemph{critical temperature}.
 Below the critical temperature there are two decreasing parts which determine the phases, see Figure~\ref{fig:iso:p} (right).
 The increasing part in between is called \defemph{elliptic} or \defemph{spinodal phase}.

 The van der Waals equations of state are defined for all $\tau \in(\tau_\liq^\tmin, \infty)$. They fulfill Definition~\ref{def:thermo} for temperature values below $T_\c$. 
 One has basically to constrain the admissible set $\admis_\liq\cup\admis_\vap$ to the convex parts of $p$, see Figure~\ref{fig:iso:p} (right) for an illustration. 
 Thus, the spinodal phase is a subset of the interval $[\tau_\liq^\tmax, \tau_\vap^\tmin]$.

 The parameters for the graphs in Figure~\ref{fig:iso:p} and e.g.\ Example~\ref{exp:kin:1} are
 \begin{align}\label{eq:vdw:param}
  a&=3, & \tau_\liq^\tmin&=\frac{1}{3}, & R&=\frac{8}{3} & &\text{and}& T=0.85.
 \end{align} 
 For these numbers, the critical temperature is $T_\c=1$.
 In order to fulfill the conditions above, we consider \eqref{eq:vdw:p} only for $\tau\in\admis_\liq\cup\admis_\vap$ with $\admis_\liq=(1/3,0.6)$ and $\admis_\vap=(2.5,\infty)$.
\end{example}

\subsection{Formulation of the two-phase Riemann problem}\label{sub:psystem}

The jump condition \eqref{eq:rh:mt} shows that  the tangential part of the velocity field is independent of the field in normal direction. 
Therefore it is reasonable  to consider a formally one-dimensional problem in normal direction to the interface $\Gamma(t)$. We pose the  Riemann  initial states 
\begin{align}\label{eq:ic}
  \begin{pmatrix} \rho\\ \vv v \cdot \vv n \end{pmatrix} (x,0)= 
  \begin{cases}
    \left(1/\tau_\L, v_\L\right)^\transp   &  \text{for }  x \leq0,\\
    \left(1/\tau_\R, v_\R\right)^\transp   &  \text{for }  x >0,
  \end{cases}
\end{align}
and $\tau_\L\in\admis_\liq$, $\tau_\R\in\admis_\vap$, $v_\L,v_\R\in\setR$, $x=(\vv x - \vv \xi)\cdot \vv{n}$. 

We keep in mind, that the original problem remains multidimensional in the sense that the local momentum balance \eqref{eq:rh:m} depends on surface tension. 
However, we solve the Riemann problem for given constant curvature $\kappa$,  such that the results can only be meaningful locally in time, but see Section~\ref{sec:numerics} on the use of Riemann solvers within numerical tracking schemes.

It is more convenient to switch to Lagrangian coordinates from now on. 
Using Lagrangian coordinates $(\xi,t)$ the task is to find specific volume and velocity fields $\tau=\tau(\xi,t)>0$ and $v=v(\xi,t)\in\setR$, such that 
\begin{align}\label{eq:psystem}
 \begin{pmatrix}
  \tau \\ v 
 \end{pmatrix}_t
 +
 \begin{pmatrix}
  -v \\ p(\tau)  
 \end{pmatrix}_\xi
 =
 \begin{pmatrix}
  0\\0
 \end{pmatrix}
\end{align}
holds in the bulk domain and 
\begin{align}\label{eq:psystemjump}
\s \jump{\tau} + \jump{v} &= 0, &
-\s\jump{v} + \jump{ p(\tau)}  &= \surf
\end{align}  
at the interface. 
Here, $p=p(\tau)$ is the pressure as in  Definition~\ref{def:thermo},
$\s$ the speed of the phase boundary in Lagrangian coordinates and 
$\surf:=(d-1)\,\surfcoeff\,\kappa$ the constant surface tension term.
The Lagrangian speed $\s$ is linked to the mass flux in Eulerian coordinates $j:=\rho_\liq(\vv v_\liq\cdot \vv n - \sigma)=\rho_\vap(\vv v_\vap\cdot \vv n - \sigma)$ via the formula
  \begin{align} \label{eq:j-s}
   \s = -j.
  \end{align}
We are in particular interested in  weak solutions $\vv U = (\tau, v)^\transp$ of \eqref{eq:psystem}  that satisfy besides \eqref{eq:psystemjump}  the entropy condition
$(\psi(\tau)+\frac{1}{2}v^2)_t + (p(\tau)\,v)_\xi \leq 0$ in the distributional sense in the bulk set and 
\begin{align}\label{eq:entropyjump}
  -\s \left( \jump{\psi(\tau) }+\jump{\tau} \mean{p(\tau)} + \surf\mean{\tau} \right) \leq 0,
\end{align}
at the interface. Note that \eqref{eq:entropyjump} is the interfacial entropy condition \eqref{eq:Euler:entropyjump} in Lagrangian coordinates.

System \eqref{eq:psystem} can be written for $\vv U = (\tau, v)^\transp$ in conservation form $\vv U_t + \vv f(\vv U)_\xi=\vv 0$ with $\vv f = (-v,  p(\tau))^\transp$.
The eigenvalues of $\vv f$ are
\begin{align}\label{eq:psystem:ew}
 \lambda_1(\tau) &= -c(\tau), &  \lambda_2(\tau) &= c(\tau),
\end{align}
where $c=c(\tau)$ is the sound speed in Lagrangian coordinates (see \eqref{eq:H6}).

\section{Two-phase Riemann solvers for monotone decreasing kinetic functions}\label{sec:solver}

Colombo \& Priuli introduced in \cite{COLPRI2003} exact solutions of the Riemann problem  for the two-phase p-system with homogeneous Rankine-Hugoniot conditions ($\surf \equiv 0$).
The  solutions are only given for  initial states in stable phases. 
However, the limitation to initial states in stable phases is inappropriate, e.g.\ for the interfacial flux computation (see Section~\ref{sec:numerics}).
Note also that static solutions correspond to  saturation states and appear at least locally in most scenarios. Thus, two-phase Riemann solvers have to handle initials states
in the vicinity of saturation states, which are stable and metastable states.

In this section, we extend the theory  in \cite{COLPRI2003} for the case with surface tension and for initial data in metastable phases. Theorem~\ref{theo:kinrel} presents the well-posedness results. 
We stress that this approach relies on kinetic relations that take the form of monotone decreasing kinetic functions (Definition~\ref{def:kinfun} below). \\
Subsection~\ref{sub:solveralgo} introduces the algorithm of the corresponding Riemann solver for given kinetic functions. 
Note that our implementation allows equations of state, that are provided by external thermodynamic libraries like \cite{CoolProp}. 
In this section  the  (monotone decreasing) kinetic functions are not specified. Physically relevant  
examples for such functions and a detailed study on their properties follow in 
Section~\ref{sec:kinrel}. 

\subsection{Solving the two-phase Riemann problem exactly}\label{sub:kinrelsolution}

Let now initial states $(\tau_\L, v_\L)^\transp \in\admis_\liq\times\setR$, $(\tau_\R, v_\R)^\transp\in\admis_\vap\times\setR$ and a constant surface tension term $\surf\in\surfset$ be given.
The required additional condition to attain unique solutions are kinetic functions.
Later on subsonic phase boundaries are constrained to those which are related to a kinetic function.

A discontinuous wave 
  \begin{align} \label{eq:discontinuous}
    \U(\xi,t) = \begin{cases}
              \U_\liq   & \text{for } \xi-\s\,t \leq0,\\
              \U_\vap  & \text{for } \xi-\s\,t >0
    \end{cases}
  \end{align}
of speed $\s\in\setR$, that connects a left state $\U_\liq=(\tau_\liq,v_\liq)^\transp\in\admis_\liq\times\setR$ and a
right state $\U_\vap=(\tau_\vap,v_\vap)^\transp\in\admis_\vap\times\setR$ is called \defemph{phase boundary} if it  satisfies the entropy condition \eqref{eq:entropyjump}.  
It follows from \eqref{eq:psystemjump} that phase boundaries propagate with speed
\begin{align}\label{eq:transitionspeed}
\s_\e(\tau_\liq  , \tau_\vap )=-\sqrt{\frac{ \surf -p(\tau_\vap )+ p(\tau_\liq  ) }{\tau_\vap -\tau_\liq  }} && \text{or} &&
\s_\c(\tau_\liq  , \tau_\vap )=+\sqrt{\frac{ \surf -p(\tau_\vap )+ p(\tau_\liq  ) }{\tau_\vap -\tau_\liq  }}.
\end{align}
The subscript $\e$ stands for evaporation and $\c$ for condensation. 
For $\tau_\liq\in\admis_\liq$ and $\tau_\vap\in\admis_\vap$, a phase boundary with negative speed is called an \defemph{evaporation wave} and a phase boundary with positive speed is called a \defemph{condensation wave}.

Furthermore, we have for evaporation waves $v_\vap = v_\liq  +\Pvel(\tau_\liq  ,\tau_\vap )$ and for condensation waves $v_\vap = v_\liq  -\Pvel(\tau_\liq  ,\tau_\vap )$, where
\begin{align}\label{eq:def:Pvel}
  \Pvel(\tau_\liq  ,\tau_\vap )&= \sign(\tau_\vap-\tau_\liq)\, \sqrt{(\tau_\vap -\tau_\liq  )\,\left(\surf- p(\tau_\vap )+ p(\tau_\liq  )\right)}.
\end{align} 
An evaporation wave (condensation wave) is called \defemph{subsonic} if there holds
\begin{align}
  \abs{\s_\e(\tau_\liq  , \tau_\vap )} < c(\tau_\vap ) && (\,\abs{\s_\c(\tau_\liq  , \tau_\vap )} < c(\tau_\vap )\,)   \label{eq:subsonic}
\end{align}
and \defemph{sonic} if \eqref{eq:subsonic} holds with equal sign.
Phase boundaries, that satisfy \eqref{eq:subsonic}, are undercompressive shock waves, cf.\ \cite{LEF2002}.
Note that these waves violate the \defemph{Lax entropy condition}
\begin{align}\label{eq:laxcondition}
    \lambda_1(\tau_\liq  ) > \s_\c(\tau_\liq  , \tau_\vap ) > \lambda_1(\tau_\vap ), && \lambda_2(\tau_\liq  ) > \s_\e(\tau_\liq  , \tau_\vap ) > \lambda_2(\tau_\vap ).
\end{align}  
  
It is well known, that self-similar solutions of two-phase Riemann problem are composed of rarefaction waves, bulk shock waves and phase boundaries. 
For brevity, let us introduce elementary waves. An \defemph{elementary wave} is either a rarefaction wave or a bulk shock wave of Lax type and satisfies
  \begin{align}\label{eq:def:Evel}
    v_\r &= \begin{cases}
               v_\l  +\Evel(\tau_\l  ,\tau_\r ) & \text{if } i=1,\\
               v_\l  -\Evel(\tau_\l  ,\tau_\r ) & \text{if } i=2
            \end{cases}
    &\text{for} &&
    \Evel(\tau_\l  ,\tau_\r )&= \begin{cases}
                           \Rvel(\tau_\l ,\tau_\r) & \text{if } i=1 \text{ and } \tau_\l   < \tau_\r \text{ or } i=2 \text{ and } \tau_\l   > \tau_\r,\\
                           \Svel(\tau_\l ,\tau_\r) & \text{else,}
                         \end{cases}
  \end{align}   
with
\begin{align}\label{eq:def:RvelSvel}
  \Rvel(\tau_\l ,\tau_\r)&:=\int\limits_{\tau_\l  }^{\tau_\r } \sqrt{-p'(\tau)} \dd \tau, &
  \Svel(\tau_\l ,\tau_\r)&:= \sign(\tau_\r-\tau_\l)\,\sqrt{-(\tau_\r -\tau_\l  )\,\left( p(\tau_\r )- p(\tau_\l  )\right)}.
\end{align}

\begin{definition}[Pair of monotone decreasing kinetic functions]\label{def:kinfun}
Let the fixed surface tension term $\surf\in\surfset$, corresponding equations of state from Definition~\ref{def:thermo}, 
numbers 
$\tau_\liq^\sc\in (\tau_\liq^\tmin, \tau_\liq^\sat)$,
$\tau_\vap^\se\in (\tau_\vap^\tmin, \infty)$ 
and two differentiable functions 
\begin{align*} 
  k_\c &: [\tau_\liq^\sc , \tau_\liq^\sat] \to \admis_\vap & &\text{and} & 
  k_\e &: [\tau_\vap^\sat, \tau_\vap^\se ] \to \admis_\liq  
\end{align*}
be given.

We call $(k_\c$, $k_\e)$ a \defemph{pair of monotone decreasing kinetic functions} 
if $k_\c'\leq0$, $k_\e'\leq0$ and the following conditions are satisfied
\begin{align}\label{eg:kinfun:entropy}
  \jump{\psi(\tau) }+\jump{\tau} \mean{p(\tau)} + \surf\mean{\tau} \,
  \begin{cases}
   \,\geq 0 &\text{for all } \tau_\liq\in[\tau_\liq^\sc , \tau_\liq^\sat], \tau_\vap=k_\c(\tau_\liq), \\
   \,\leq 0 &\text{for all } \tau_\vap\in[\tau_\vap^\sat, \tau_\vap^\se ], \tau_\liq=k_\e(\tau_\vap)
  \end{cases}
\end{align}
with  $\jump{\tau} = \tau_\vap-\tau_\liq$ and $\mean{\tau} = \frac{1}{2}(\tau_\liq+\tau_\vap)$ and
\begin{align*}
 k_\c(\tau_\liq^\sat) &= \tau_\vap^\sat, & k_\c(\tau_\liq^\sc) &= \tau_\vap^\sc, & \abs{\s_\c(\tau_\liq^\sc, \tau_\vap^\sc)} &= c(\tau_\vap^\sc), & & \\
 k_\e(\tau_\vap^\sat) &= \tau_\liq^\sat, & k_\e(\tau_\vap^\se) &= \tau_\liq^\se, & \abs{\s_\e(\tau_\liq^\se, \tau_\vap^\se)} &= c(\tau_\vap^\se), & k_\e'(\tau_\vap^\se) &=0.
\end{align*}
\end{definition}
Note that sonic phase boundaries are determined by the end states $\tau_\liq^\sc$, $\tau_\vap^\sc$ resp.  $\tau_\vap^\se= \tau_\liq^\se$. The superscripts $^\sc$, $^\se$ stand for \defemph{sonic condensation} and \defemph{sonic evaporation}, respectively. 

We will consider pairs of monotone decreasing kinetic functions in order to single out a unique two-phase Riemann solution.
Examples of such functions will be given in Section~\ref{sec:kinrel}.

\begin{definition}[Admissible phase boundary] \label{def:admissible}
 A phase boundary, that connects a left state $\U_\liq  =(\tau_\liq  ,  v_\liq  )^\transp\in\admis_\liq \times \setR$ and a right state $\U_\vap =(\tau_\vap ,  v_\vap )^\transp \in \admis_\vap \times \setR$ is called \defemph{admissible phase boundary} if and only if
 \begin{itemize}
  \item it is a sonic or a supersonic wave of Lax type \eqref{eq:laxcondition}, or
  \item it is a subsonic condensation wave that satisfies $k_\c(\tau_\liq) = \tau_\vap$, or
  \item it is a subsonic evaporation wave that satisfies $k_\e(\tau_\vap) = \tau_\liq$,
 \end{itemize}
 where $k_\c$ and $k_\e$ are a pair of monotone decreasing kinetic functions as in Definition~\ref{def:kinfun}.
\end{definition}
Note that with \eqref{eg:kinfun:entropy}, it follows that all admissible phase boundaries satisfy the entropy inequality \eqref{eq:entropyjump}. Furthermore, \defemph{thermodynamic equilibrium solutions} (discontinuous waves \eqref{eq:discontinuous} with $\U_\l=(\tau_\liq^\sat(\surf),0)^\transp$, $\U_\r=(\tau_\vap^\sat(\surf),0)^\transp$) are admissible subsonic phase boundaries.

\medskip
We seek for a self-similar entropy solutions of the two-phase Riemann problem, that contains exactly one admissible phase boundary. Furthermore, we prefer solutions with subsonic phase boundaries, whenever this is possible. 
We call such a solution \defemph{(admissible) two-phase Riemann solution}.

It is possible to define generalized Lax curves for these requirements.
The Lax curve $v^\ast = v_\L   + \mathcal{L}_1(\tau_\L  ,\tau^\ast)$ of the first family is given by Table~\ref{tab:KL1}. 
The structure changes depending on the arguments $\tau_\L$ and $\tau^\ast$. 
We enumerate the different wave patterns with the symbols of the first column in the table.
The subscript $_\L$ indicates that the wave connects the left initial state $(\tau_\L,  v_\L)^\transp$ to an intermediate state $(\tau^\ast,  v^\ast)^\transp$.

\begin{table}
\newcolumntype{s}{>{\columncolor{lightgray}} c}
\centering
\begin{tabularx}{\textwidth}{*{8}l}
  \toprule
  type& $\tau_\L  $		& $\tau^\ast$				& composition	& $\mathcal{L}_1(\tau_\L  ,\tau^\ast)$									& $\tau_\liq$  		& $\tau_\vap$	\\
  \cmidrule(r){1-1}\cmidrule(lr){2-2}\cmidrule(lr){3-3}			\cmidrule(lr){4-4}\cmidrule(lr){5-5} 											\cmidrule(lr){6-6}	\cmidrule(l){7-7} 
  \klwave{1}& $\admis_\liq$	& $\admis_\liq$				& 1E		& $\Evel(\tau_\L  ,\tau^\ast)$ 										& --			& --		\\ 
  \klwave{2}& $\admis_\liq$	& $[\tau_\vap^\sat, \tau_\vap^\se ]$	& 1E-KE		& $\Evel(\tau_\L  ,k_\e(\tau^\ast))+\Pvel(k_\e(\tau^\ast), \tau^\ast)$					& $k_\e(\tau^\ast)$	& $\tau^\ast$	\\ 
  \klwave{3}&$\admis_\liq$	& $(\tau_\vap^\se,\infty)$		& 1E-SE-1R	& $\Evel(\tau_\L  ,\tau_\liq^\se)+\Pvel(\tau_\liq^\se, \tau_\vap^\se)+ \Rvel(\tau_\vap^\se,\tau^\ast)$	& $\tau_\liq^\se$	& $\tau_\vap^\se$ \\ 
  \bottomrule  
\end{tabularx}
\caption{Definition of the map $\mathcal{L}_1 : \admis_\liq \times \admis_\liq\cup [\tau_\vap^\sat, \infty) \to \setR$, that determines the 
Lax curve $v^\ast = v_\L   + \mathcal{L}_1(\tau_\L  ,\tau^\ast)$ of the first family. 
The resulting (multiple) waves for left and right trace specific volume values $\tau_\L$ and $\tau^\ast$ are composed of the waves given in the fourth column (from left to right): 1E stands for 1-elementary wave, 1R for 1-rarefaction wave, KE stands for subsonic evaporation wave that is related to a kinetic function, SE for sonic evaporation.
The functions $\Evel$, $\Pvel$ and $\Rvel$ are given in \eqref{eq:def:Evel}, \eqref{eq:def:Pvel} and \eqref{eq:def:RvelSvel}, respectively. The interface states are given by the last two columns.
}
\label{tab:KL1}
\end{table}

Figure~\ref{fig:KL1} shows a wave of type~\klwave{2} and a wave of type~\klwave{3}, where we used 
\begin{align}\label{eq:p:fig}
 p(\tau)&= \begin{cases}
         2/\tau +1 &: \tau \in (0 , 2/3),\\
        20/\tau -1 &: \tau \in (3 , \infty)
       \end{cases}&
 &\text{with}&
 \tau_\liq^\sat &= 1/2, &
 \tau_\vap^\sat &= 10/3.  
\end{align}
The equation of state \eqref{eq:p:fig} was chosen in order to visualize wave patterns more clearly.

\begin{figure}\centering

	\input{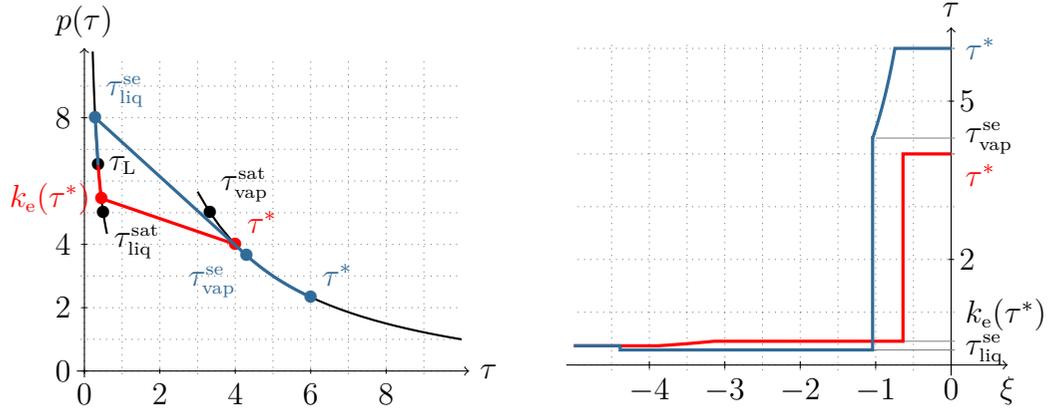}
	\pgfmathsetmacro{\taul}{4/11}

	\newcommand{\kred}{0.45}
	\pgfmathsetmacro{\kblue}{2/7}
	\newcommand{\sevap}{4.3}
	\newcommand{\taured}{4}
	\newcommand{\taublue}{6}
	\begin{tikzpicture}[xscale=0.5,yscale=0.42]

	\draw[thin,color=gray,dotted] (-0.1,-0.1) grid (9.9,9.9);

	\draw[->] (-0.2,0) -- (10.2,0) node[right] {$\tau$};
	\draw[->] (0,-0.2) -- (0,10.2) node[above] {$p(\tau)$};

	\foreach \x/\xtext in {0/0, 2/2, 4/4, 6/6, 8/8}
	    \draw[shift={(\x,0)}] (0pt,2pt) -- (0pt,-2pt) node[below] {$\xtext$};
	\foreach \y/\ytext in {0/0, 2/2, 4/4, 8/8}
	    \draw[shift={(0,\y)}] (2pt,0pt) -- (-2pt,0pt) node[left] {$\ytext$};

	\coordinate (Sl) at (\sl, 5);
	\coordinate (Sv) at (\sv, 5);
	\draw (Sl) node {\textbullet} node [below right] { $\tau_\liq^\sat$ };
	\draw (Sv) node {\textbullet} node [above right] { $\tau_\vap^\sat$ };

	\draw[thick] plot[domain=0.22:0.6] ({\x}, {pliq(\x)} ); 
	\draw[thick] plot[domain=3:10] (\x, {pvap(\x)} );

	\draw (\taul,  {pliq(\taul)} )  node {\textbullet}  node [right] { $\tau_\L$ };
	\draw [red]  (\kred,  {pliq(\kred)} )  node {\textbullet}  node [left] { $k_\e(\tau^\ast)$ };	
	\draw [blue] (\sevap, {pvap(\sevap)})  node {\textbullet}  node [below left] { $\tau_\vap^\se$ };
	\draw [blue] (\kblue, {pliq(\kblue)})  node {\textbullet}  node [above right] {  $\tau_\liq^\se$ };	

	\draw[very thick, color = red ] 
	  plot[domain=\taul:\kred] (\x, {pliq(\x)} ) -- 
	  (\taured, {pvap(\taured)} ) 
	  node {\textbullet} 
	  node [above right] {$\tau^\ast$};
	\draw[very thick, color = blue] 
	  (\taul, {pliq(\taul)} ) -- (\kblue, {pliq(\kblue)} ) --
	  plot[domain=\sevap:\taublue] (\x, {pvap(\x)} ) 
	  node {\textbullet} 
	  node [above right] {$\tau^\ast$};

	\end{tikzpicture}
	\newcommand{\stime}{2}
	\hspace{\hfloatsep}
	\begin{tikzpicture}[xscale=0.5,yscale=0.7]

	\draw[thin,color=gray,dotted] (-9.9,-0.1) grid (0.1,6.1);

	\draw[->] (-10.2,0) -- (1.5,0) node[below ] {$\xi$};
	\draw[->] (0,-0.2)  -- (0,6.4) node[above] {$\tau$};

	\foreach \x/\xtext in {-8/-4, -6/-3, -4/-2, -2/-1, 0/0 }
	    \draw[shift={(\x,0)}] (0pt,2pt) -- (0pt,-2pt) node[below] {$\xtext$};
	\foreach \y/\ytext in {2/2, 5/5}
	    \draw[shift={(0,\y)}] (2pt,0pt) -- (-2pt,0pt) node[right] {$\ytext$};

	\draw [very thin, gray] ( -6.2,   \kred) -- (0.1,   \kred)  node [black, above right] { $k_\e(\tau^\ast)$ };
	\draw [very thin, gray] ( -6.2,  \kblue) -- (0.1,  \kblue)  node [black,       right] { $\tau_\liq^\se$ };
	\draw [very thin, gray] ( -1.0, \taured) -- (0.1, \taured)  node [below right,red] { $\tau^\ast$ };
	\draw [very thin, gray] ( -2.0,  \sevap) -- (0.1,  \sevap)  node [black, right] { $\tau_\vap^\se$};
	\draw [very thin, gray] ( -0.3,\taublue) -- (0.1,\taublue)  node [right, blue] { $\tau^\ast$};

	\draw [very thick, color = red ]
	  ({-5*\stime},\taul) -- 
	  plot[domain=\taul:\kred] ( {-sliq(\x)*\stime}, \x ) -- 
	  ( {-ptspeed(\kred,\taured)*\stime} , \kred ) -- ( {-ptspeed(\kred,\taured)*\stime} , \taured ) -- (0, \taured);

	\draw [very thick, color = blue ]
	  ({-5*\stime},\taul) -- 
	  ( {-liqspeed(\taul,\kblue)*\stime} , \taul ) -- ( {-liqspeed(\taul,\kblue)*\stime} , \kblue ) --  
	  ( {-ptspeed(\kblue,\sevap)*\stime} , \kblue ) -- ( {-ptspeed(\kblue,\sevap)*\stime} , \sevap ) --  
	  plot[domain=\sevap:\taublue] ( {-svap(\x)*\stime}, \x ) --
	  (0,\taublue);      
	\end{tikzpicture}
	\caption{The sketch on the left hand side shows the graph of the pressure function \eqref{eq:p:fig}. The $\tau$-values, where the Lax curve of the first family alters its wave structure, are marked with a dot. The red curve corresponds to wave type~\klwave{2} and the blue curve to wave type~\klwave{3}. 
	The figure on the right hand side shows these waves at time $t=1$ in the $(\tau,\xi)$-plane. }      \label{fig:KL1}
\end{figure}

We summarize the main properties to a proposition. 
\begin{proposition}[Properties of the generalized Lax curve $\mathcal{L}_1$]  \label{prop:KL1}
  Let a left state $(\tau_\L, v_\L)^\transp \in \admis_\liq\times \setR$ and 
  the map $\mathcal{L}_1 : \admis_\liq \times \admis_\liq\cup [\tau_\vap^\sat, \infty) \to \setR$ of Table~\ref{tab:KL1} be given. 
  Then the following properties hold.
  \begin{enumerate}
    \item The map $\mathcal{L}_1$ is continuous.
    \item The map 
	\begin{align*}
	&\admis_\liq\cup [\tau_\vap^\sat, \infty) \to \setR , &
	&\tau^\ast \mapsto  v^\ast = v_\L   + \mathcal{L}_1(\tau_\L  ,\tau^\ast)
	\end{align*}
	is differentiable and strictly monotone increasing in $\admis_\liq$ and in $[\tau_\vap^\sat, \infty)$. 
  
     \item It holds that $\mathcal{L}_1(\tau_\L  ,\tau_\liq^\sat) = \mathcal{L}_1(\tau_\L  ,\tau_\vap^\sat) $.  
     
     \item All propagation speeds ($\s_e$, $\lambda_1$) are negative. For waves of type~\klwave{2} and type~\klwave{3} the phase boundary propagates faster than the elementary wave in the liquid phase and slower than the rarefaction wave connecting to $\tau^\ast$ in wave type~\klwave{3}.
     
     \item Evaporation waves are either subsonic or sonic.
     
     \item The speed of an evaporation wave is limited by the sound speed $-c(\tau_\vap^\se)$. 
  
  \end{enumerate}  
\end{proposition}
\begin{proof}
 (\rmnum{1}) By definition, the map $\mathcal{L}_1$ is piecewise continuous. It is readily checked with Table~\ref{tab:KL1}, that also the transition from one domain of definition to another is continuous.
 
 (\rmnum{2}) Note that $\mathcal{L}_1$ is piecewise smooth. The critical point is $\tau^\ast=\tau_\vap^\se$. A short calculation gives
 \begin{align*}
  &\lim_{\tau^\ast\to\tau_\vap^\se} \ddfrac{\Svel}{\tau^\ast} (\tau_\L,k_\e(\tau^\ast)) = 0, &
  &\lim_{\tau^\ast\to\tau_\vap^\se} \ddfrac{\Rvel}{\tau^\ast} (\tau_\L,k_\e(\tau^\ast)) = 0  \text{ with } k_\e'(\tau_\vap^\se)=0, \\
  &\lim_{\tau^\ast\to\tau_\vap^\se} \ddfrac{\Svel}{\tau^\ast} (\tau_\vap^\se,\tau^\ast) = c(\tau_\vap^\se), &
  &\lim_{\tau^\ast\to\tau_\vap^\se} \ddfrac{\Rvel}{\tau^\ast} (\tau_\vap^\se,\tau^\ast) = c(\tau_\vap^\se) \text{ and } \\
  &\lim_{\tau^\ast\to\tau_\vap^\se} \ddfrac{\Pvel}{\tau} (k_\e(\tau^\ast),\tau^\ast) = c(\tau_\vap^\se)  & 
  &\text{ with } k_\e'(\tau_\vap^\se)=0 
  \text{ and } \abs{\s_\e(\tau_\liq^\sc, \tau_\vap^\sc)} = c(\tau_\vap^\sc)
 \end{align*}
 for the functions $\Svel$, $\Rvel$ and $\Pvel$, from \eqref{eq:def:RvelSvel}, \eqref{eq:def:Pvel}. Thus, the derivatives of a wave of type~\klwave{2} and type~\klwave{3} coincide in $\tau_\vap^\se$. 
 The functions $\Svel$ and $\Rvel$ are strictly monotone increasing with respect to the second argument.
 A short calculation shows that $\mathcal{L}_1$ is strictly monotone increasing also for a wave of type~\klwave{2}, since $k_\e'<0$.

 (\rmnum{3}) The condition holds, since $\Pvel(\tau_\liq^\sat,\tau_\vap^\sat)=0$.

 (\rmnum{4})-(\rmnum{6}) By definition, all waves of the first family have non-positive propagation speeds. The speed of the evaporation wave is in 
 $[-c(\tau_\vap^\se),\, 0]$. Due to the pressure assumptions (Definition~\ref{def:thermo}), waves in the liquid phase propagate faster (in absolute values) than the vapor sound speed. The phase boundary in wave type~\klwave{3} is sonic and the vapor rarefaction wave is attached.
\end{proof}


The generalized Lax curve of the second family may contain a condensation wave.
Condensation waves change from subsonic to supersonic or vice versa in the point $\tau_\liq^\sc$.
The next lemmas introduce further points in $\admis_\liq\cup\admis_\vap$ where the solution changes its structure. The lemmas are a direct consequence of the pressure assumptions in Definition~\ref{def:thermo}.

The first lemma states that phase boundaries move slower than sound in the liquid phase.
That means in terms of the pressure function, that the slope of $p$ in an arbitrary $\tau_\liq\in\admis_\liq$ is steeper as the slope of the chord from $(\tau_\liq,p(\tau_\liq)+\surf)$ to  $(\tau_\vap,p(\tau_\vap))$, for any $\tau_\vap\in\admis_\vap$.
\begin{lemma}[Sound in the liquid travels faster than phase boundaries]\label{lm:H5:2}
  Let the pressure function $p : \admis_\liq\cup\admis_\vap \to \setR$ of Definition~\ref{def:thermo} be given.
  
  For all $\tau_\liq\in\admis_\liq$ and $\tau_\vap\in\admis_\vap$ it holds, that
  \begin{align*}
          p'(\tau_\liq) < \frac{p(\tau_\vap) -p(\tau_\liq) - \surf}{\tau_\vap - \tau_\liq},    
  \end{align*}
  or equivalently $c(\tau_\liq)> \abs{\s_{\e/\c}(\tau_\liq,\tau_\vap)}$.
\end{lemma}
\begin{proof}
  Consider first the case $\tau_\liq=\tau_\liq^\sat$. Define 
  $f(\tau):=    p'(\tau_\liq^\tmax) \, (\tau - \tau_\liq^\tmax) + p(\tau_\liq^\tmax) + \surf -p(\tau)$. 
  Due to \eqref{eq:H1} and \eqref{eq:H3} $f(\tau_\vap^\tmin)<0$ holds, and due to \eqref{eq:H5} we have $f'(\tau_\vap^\tmin)<0$.
  With \eqref{eq:H2} it follows 
  \begin{align*}
    \begin{matrix}
     p(\tau_\vap) &>&  p(\tau_\liq^\tmax) + \surf &+& p'(\tau_\liq^\tmax) (\tau_\vap - \tau_\liq^\tmax) & & \\
		  &=&  p(\tau_\liq^\tmax) + \surf &+& p'(\tau_\liq^\tmax) (\tau_\liq - \tau_\liq^\tmax) &+& p'(\tau_\liq^\tmax) (\tau_\vap - \tau_\liq) \\
		  &>&  p(\tau_\liq      ) + \surf & &                                                   &+& p'(\tau_\liq      ) (\tau_\vap - \tau_\liq). 
    \end{matrix}
  \end{align*}
\end{proof}

The subsequent lemmas introduce values $\hat{\tau}$, $\check{\tau}$ and a function $g_s$. The value  $\hat{\tau}$ is such that the pressure function has the same slope in $\tau_\R$ as the chord from $(\hat{\tau},p(\hat{\tau})+\surf)$ to  $(\tau_\R,p(\tau_R))$, see Figure~\ref{fig:KL2} (left) for illustration.
The value $\check{\tau}$ is such that the points $(\check{\tau},p(\check{\tau})+\surf)$, $(\tau_\vap^\sat,p(\tau_\vap^\sat))$, $(\tau_\R,p(\tau_\R))$ lie on one straight line. 
The function $g_s$ is determined such that the pressure function has the same slope in $g_s(\tau)$ as the chord from $(\tau,p(\tau)+\surf)$ to  $(g_s(\tau),p(g_s(\tau)))$, see Figure~\ref{fig:KL2} (left).

\begin{figure}\centering

	\input{tikzconstants}
	\newcommand{\taur}{4}
	\newcommand{\tauh}{0.26}
	
	\newcommand{\taured}{0.4}
	\newcommand{\kred}{5.5}
	
	\newcommand{\taublue}{0.48}
	\newcommand{\kblue}{3.5}
	
	\newcommand{\scvap}{6}	
	\newcommand{\scliq}{0.45}
	
	\begin{tikzpicture}[xscale=0.7,yscale=0.5]

	\draw[thin,color=gray,dotted] (-0.1,0.9) grid (7.9,9.9);

	\draw[->] (-0.2,1) -- (8.2,1) node[right] {$\tau$};
	\draw[->] (0,0.8) -- (0,10.2) node[left] {$p(\tau)$};

	\foreach \x/\xtext in {0/0, 2/2, 4/4, 6/6, 8/8}
	    \draw[shift={(\x,1)}] (0pt,2pt) -- (0pt,-2pt) node[below] {$\xtext$};
	\foreach \y/\ytext in {1/1, 3/3, 8/8}
	    \draw[shift={(0,\y)}] (2pt,0pt) -- (-2pt,0pt) node[left] {$\ytext$};

	\draw[thick] plot[domain=0.22:0.6] ({\x}, {pliq(\x)} ); 
	\draw[thick] plot[domain=3:8] (\x, {pvap(\x)} );
	
	\draw (\sl, {pliq(\sl)} ) node {\textbullet} node [below left,xshift=-6pt] { $\tau_\liq^\sat$ } --
	      (\sv, {pvap(\sv)} ) node {\textbullet} node [above right] { $\tau_\vap^\sat$ };
	\draw (\scliq, {pliq(\scliq)} ) node {\textbullet} node [above left,xshift=-6pt] { $\tau_\liq^\sc$ } --
	      (\scvap, {pvap(\scvap)} ) node {\textbullet} node [below] { $\tau_\vap^\sc$ };
	\coordinate (R) at (\taur,{pvap(\taur)});
	\draw (R) node {\textbullet} node [right] { $\tau_\R$ } --
	     (\tauh,{pliq(\tauh)}) node {\textbullet} node [left] { $\hat{\tau}\,$ };;

	\draw[very thick, color = red ] 
	  plot[domain=\taur:\kred] (\x, {pvap(\x)} ) node {\textbullet} node [above right] {$g_s(\tau^\ast)$} -- 
	  (\taured,{pliq(\taured)} )
	  node {\textbullet} 
	  node [above right] {$\tau^\ast \,$};
	  
	\draw[very thick, color = blue] 
	  (R) --  (\kblue,{pvap(\kblue)} ) node {\textbullet} node [right] {$k_\c(\tau^\ast)$} -- 
	  (\taublue,{pliq(\taublue)} )
	  node {\textbullet} 
	  node [left,xshift=-6pt] {$\tau^\ast$};

	\end{tikzpicture}
	\newcommand{\stime}{6}
	\begin{tikzpicture}[xscale=0.5, yscale=0.7]

	\draw[thin,color=gray,dotted] (-0.1,-0.1) grid (9.9,6.5);

	\draw[->] (-0.2,0) -- (10.2,0) node[right] {$\xi$};
	\draw[->] (0,-0.2) -- (0,6.8) node[left] {$\tau$};

	\foreach \x/\xtext in {0/0, 2/\frac{1}{3}, 4/\frac{2}{3}, 6/1, 8/\frac{4}{3}}
	    \draw[shift={(\x,0)}] (0pt,2pt) -- (0pt,-2pt) node[below] {$\xtext$};
	\foreach \y/\ytext in { 2/2, 6/6}
	    \draw[shift={(0,\y)}] (2pt,0pt) -- (-2pt,0pt) node[left] {$\ytext$};
	
	\draw [very thin, gray] (10.2,    \taur) -- (-0.1,    \taur)  node [black, left] { $\tau_\R$ };
	\draw [very thin, gray] ( 5.5,   \kblue) -- (-0.1,   \kblue)  node [blue, below left] { $k_\c(\tau^\ast)$};
	\draw [very thin, gray] ( 5.0,    \kred) -- (-0.1,    \kred)  node [red, left] { $g_s(\tau^\ast)$};
	\draw [very thin, gray] ( 5.0,  \taured) -- (-0.1,  \taured)  node [black, below left, red ] { $\tau^\ast$ };
	\draw [very thin, gray] ( 2.0, \taublue) -- (-0.1, \taublue)  node [black, above left, blue] { $\tau^\ast$};

	\draw [very thick, color = blue  ]
	  (0,\taublue) -- 
	  ({ptspeed(\taublue,\kblue)*\stime},\taublue) -- 
	  ({ptspeed(\taublue,\kblue)*\stime} , \kblue ) -- 
	  ({vapspeed(\kblue,\taur)*\stime} , \kblue ) --
	  ({vapspeed(\kblue,\taur)*\stime} , \taur ) --
	  (10, \taur) ;

	\draw [very thick, color = red  ]
	  (0,\taured) -- 
	  ({ptspeed(\taured,\kred)*\stime},\taured) -- 
	  ({ptspeed(\taured,\kred)*\stime},\kred) -- 
	  plot[domain=\kred:\taur] ( {svap(\x)*\stime}, \x )  --
	  (10, \taur) ;
	
	\end{tikzpicture}
	\caption{Pressure function (left) and specific volume distribution (right), like Figure~\ref{fig:KL1}. 
	The red curve corresponds to wave type~\krwave{4} and the green curve to wave type~\krwave{3}. }   \label{fig:KL2}

\end{figure}

For ease of notation, we skip the dependencies on numbers that are constant for two-phase Riemann problems, i.e.\ $\tau_\L\in\admis_\liq$, $\tau_\R\in\admis_\vap$ and $\surf\in\surfset$. 
Recall that the pressure function and saturation states depend on the constant $\surf$, see Definition~\ref{def:thermo}

\begin{table}
\centering
\begin{tabularx}{\textwidth}{*{7}l}
  \toprule
  type& $\tau^\ast$				& $\tau_\R$				& composition			& $\mathcal{L}_2(\tau^\ast,\tau_\R)$					& $\tau_\liq$  		& $\tau_\vap$	\\
  \cmidrule(r){1-1}\cmidrule(lr){2-2}		\cmidrule(lr){3-3}			\cmidrule(lr){4-4}		\cmidrule(lr){5-5}							\cmidrule(lr){6-6}	\cmidrule(l){7-7}
  \krwave{1}& $\admis_\vap$			& $\admis_\vap$				& 2E				& $\Evel(\tau^\ast, \tau_\R)$ 						& --			& --		\\ 
  \krwave{2}& $(\tau_\liq^\tmin, \hat{\tau}]$	& $(\tau_\vap^\tmin,\tau_\vap^\sc]$	& LC				& $\Pvel(\tau^\ast, \tau_\R)$						& $\tau^\ast$		& $\tau_\R$	\\ 
  \krwave{3}& $(\hat{\tau}, \tau_\liq^\sc)$	& $(\tau_\vap^\tmin,\tau_\vap^\sc]$	& SC-2R				& $\Pvel(\tau^\ast, g_s(\tau^\ast))+\Rvel(g_s(\tau^\ast), \tau_\R)$ 	& $\tau^\ast$		& $g_s(\tau^\ast)$\\
  \krwave{4}& $[\tau_\liq^\sc, \tau_\liq^\sat]$	& $(\tau_\vap^\tmin,\tau_\vap^\sc]$	& KC-2E				& $\Pvel(\tau^\ast, k_\c(\tau^\ast))+\Evel(k_\c(\tau^\ast), \tau_\R)$ 	& $\tau^\ast$		& $k_\c(\tau^\ast)$\\

  \krwave{5}& $(\tau_\liq^\tmin, \check{\tau}]$	& $(\tau_\vap^\sc, \infty)$		& LC				& $\Pvel(\tau^\ast, \tau_\R)$						& $\tau^\ast$		& $\tau_\R$	\\ 
  \krwave{6}& $(\check{\tau}, \tau_\liq^\sat]$	& $(\tau_\vap^\sc, \infty)$		& KC-2S				& $\Pvel(\tau^\ast, k_\c(\tau^\ast))+\Svel(k_\c(\tau^\ast), \tau_\R)$ 	& $\tau^\ast$		& $k_\c(\tau^\ast)$\\  
  \bottomrule
\end{tabularx}

  \caption{Definition of the map $\mathcal{L}_2 : (\tau_\liq^\tmin,\tau_\liq^\sat]\cup \admis_\vap\times \admis_\vap \to \setR$, that determines the
  Lax curve $v^\ast = v_\R + \mathcal{L}_2(\tau^\ast,\tau_\R)$ of the second family. 
  The resulting (multiple) waves for left and right trace specific volume values $\tau^\ast$ and $\tau_\R$ are composed of the waves given in the fourth column (from left to right): 2E stands for 2-elementary wave, SC for sonic condensation, LC for supersonic (Lax-type) condensation and KC for stands for a condensation wave that is related to a kinetic function. 
  The functions $\Evel$, $\Pvel$, $\Rvel$ and $\Svel$ are given in \eqref{eq:def:Evel}, \eqref{eq:def:Pvel} and \eqref{eq:def:RvelSvel}. The interface states are given by the last two columns.
  }  
  \label{tab:KL2}
\end{table}

\begin{lemma}[The values $\hat{\tau}$ and $\check{\tau}$] \label{lm:hat}
  For a fixed $\tau_\R \in(\tau_\vap^\tmin, \tau_\vap^\sc]$ there exists a unique $\hat{\tau} \in \admis_\liq$, such that
  \begin{align}   \label{eq:tauhat}
    p'( \tau_\R) = \frac{p(\tau_\R) - p(\hat{\tau}) -\surf }{ \tau_\R - \hat{\tau} },
  \end{align}
  or equivalently $\lambda_2( \tau_\R) = \s_\c(\hat{\tau}, \tau_\R)$ holds. Moreover, $\hat{\tau}\in (\tau_\liq^\tmin, \tau_\liq^\sc]$.
  
  On the other hand, for fixed $\tau_\R > \tau_\vap^\sc$, there exists a unique $\check{\tau} \in \admis_\liq$, such that
  \begin{align*}
    \frac{p(k_\c(\check{\tau})) - p(\check{\tau}) - \surf}{ k_\c(\check{\tau}) - \check{\tau} } = 
    \frac{p(\tau_\R) - p(\check{\tau}) - \surf}{ \tau_\R - \check{\tau} },
  \end{align*}
  or equivalently $\s_\c(\check{\tau}, k_\c(\check{\tau})) = \s_\c(\check{\tau}, \tau_\R)$ holds.
  Moreover, $\check{\tau} \in (\tau_\liq^\sc , \tau_\liq^\sat)$.
\end{lemma}
At the value $\hat{\tau}$, a supersonic condensation wave (see wave of type~\krwave{2} in Table~\ref{tab:KL2}) splits up into a sonic condensation wave and a 2-rarefaction wave.
At the value $\check{\tau}$, a supersonic condensation wave (see wave of type~\krwave{5}) breaks into a subsonic condensation wave and a 2-shock wave.
\begin{proof}[Proof of Lemma~\ref{lm:hat}]
 \newcommand{\ff}{\ensuremath{\hat{f}}}
 \newcommand{\hh}{\ensuremath{\check{f}}}
 Define the function
 \begin{align*}
  \ff(\tau; \tau_\R) &:= p'( \tau_\R) - \frac{p(\tau_\R) - p({\tau}) -\surf }{ \tau_\R - {\tau} }, &
  &\text{whereby}&
  \lim_{\tau\to\tau_\liq^\tmin} \ff(\tau; \tau_\R) &= \infty
 \end{align*}
 holds due to \eqref{eq:H4}. By definition of the points $\tau_\liq^\sc,\tau_\vap^\sc$ we find $\ff(\tau_\liq^\sc; \tau_\R) < \ff(\tau_\liq^\sc; \tau_\vap^\sc)=0$. The function $\ff$ is continuous, thus $\hat{\tau}\in (\tau_\liq^\tmin, \tau_\liq^\sc]$ exists where $ \ff(\hat{\tau}; \tau_\R)=0$.
 The derivation $\ff'(\tau; \tau_\R) =  \left( p'( \tau) - \frac{p(\tau_\R) - p({\tau}) -\surf }{ \tau_\R - {\tau} } \right) / (\tau_\R - {\tau})$ is positive due to $\tau \in \admis_\liq$ and Lemma~\ref{lm:H5:2}. Thus, there exists a unique $\hat\tau$.
 
 For the second part define
 \begin{align*}
  \hh(\tau) &:=   \frac{p(k_\c({\tau})) - p({\tau}) - \surf}{ k_\c({\tau}) - {\tau} } - \frac{p({\tau_\R}) - p(k_\c({\tau}))}{ {\tau_\R} - k_\c({\tau}) }
  && \text{for} && \tau\in[\tau_\liq^\sc,\tau_\liq^\sat].  
 \end{align*}
 Note that $k_\c(\tau_\liq^\sc)=\tau_\vap^\sc$ and $k_\c(\tau_\liq^\sat)=\tau_\vap^\sat$. With \eqref{eq:H2} there holds $\hh(\tau_\liq^\sc) <0$ and with \eqref{eq:H1} $\hh(\tau_\liq^\sat) >0$. The function $\hh$ is continuous such that there exists $\check{\tau} \in (\tau_\liq^\sc , \tau_\liq^\sat)$ with $\hh(\check{\tau})=0$. 
 For uniqueness, we show that $\hh$ is strictly monotone increasing. From \eqref{eq:H5}, it follows that $p'(\tau)<p'(k_\c(\tau))$ for $\tau\in[\tau_\liq^\sc , \tau_\liq^\sat]$. This is applied to $\hh'$ and yields
 \begin{align*}
   \hh'(\tau) > \frac{k_\c'({\tau})\!-\!1}{k_\c({\tau})\!-\!\tau}  \left(p'(k_\c({\tau})) -\frac{p(k_\c({\tau}))\!-\!p({\tau})\!-\!\surf}{ k_\c({\tau}) - {\tau} } \right)
              + \frac{k_\c'({\tau})  }{\tau_\R\!-\!k_\c({\tau})}  \left(p'(k_\c({\tau})) -\frac{p(  {\tau_\R} )\!-\!p(k_\c({\tau})))}{ {\tau_\R} - k_\c({\tau}) } \right)>0.
 \end{align*}
 The first bracket is zero for $\tau=\tau_\liq^\sc$ and negative otherwise. The second bracket is negative due to \eqref{eq:H2}. Thus, $\check{\tau}$ is uniquely determined.
\end{proof}

Waves of type~\lrwave{3} are composed of a sonic condensation wave and an attached 2-rarefaction wave, cf.\ Table~\ref{tab:KL2}.  
The following lemma is helpful to find the sonic vapor end state of the wave in terms of the liquid end state. 
\begin{lemma}[The function $g_s$]\label{lm:gs}
  For any given $\tau_\R \in(\tau_\vap^\tmin, \tau_\vap^\sc]$, let $\hat{\tau}\in (\tau_\liq^\tmin, \tau_\liq^\sc]$ as in Lemma~\ref{lm:hat} be given.
  
  There exists a continuous monotone increasing function $g_s: [\hat{\tau}, \tau_\liq^\sc] \to [\tau_\R, \tau_\vap^\sc]$, $\tau\mapsto g_s(\tau)$ such that
  \begin{align*}  
   p'( g_s(\tau)) = \frac{ p(g_s(\tau)) - p({\tau}) - \surf}{ g_s(\tau) - {\tau} },
  \end{align*}
  or equivalently $\lambda_2(g_s(\tau)) = \s_\c(\tau, g_s(\tau))$ holds.
\end{lemma}
Note that the domain of definition depends on $\hat{\tau}$ and thus on $\tau_\R$. The function $g_s$ does not depend on $\tau_\R$, however the restriction to $[\hat{\tau}, \tau_\liq^\sc]$ guarantees the existence of $g_s$. 
\begin{proof}[Proof of Lemma~\ref{lm:gs}]
 We apply the implicit function theorem to the function
 \begin{align*}
  F(\tau_\liq,\tau_\vap) &:=  p'( \tau_\vap) (\tau_\vap-\tau_\liq) - p(\tau_\vap) + p(\tau_\liq) + \surf.
 \end{align*}
 With \eqref{eq:tauhat}, it follows that $F(\hat{\tau},\tau_\R)=0$. 
 The local existence of the function $g_s$ follows from $\pslash{F}{\tau_\liq} =  - p'(\tau_\vap) + p'(\tau_\liq)<0$ with \eqref{eq:H5}. We can proceed with the latter argument until $\tau_\liq^\sc$ is reached, where $F(\tau_\liq^\sc,\tau_\vap^\sc)=0$ holds.
 
 With \eqref{eq:H2} it holds that $\pslash{F}{\tau_\vap} =   p''(\tau_\vap) (\tau_\vap - \tau_\liq)>0$.
 The monotonicity follows from $\ddfrac{F}{\tau}(\tau,g_s(\tau)) = \pslash{F}{\tau_\liq} + \pslash{F}{\tau_\vap}\,g_s'(\tau)=0$.
\end{proof}

\begin{figure}\centering

	\input{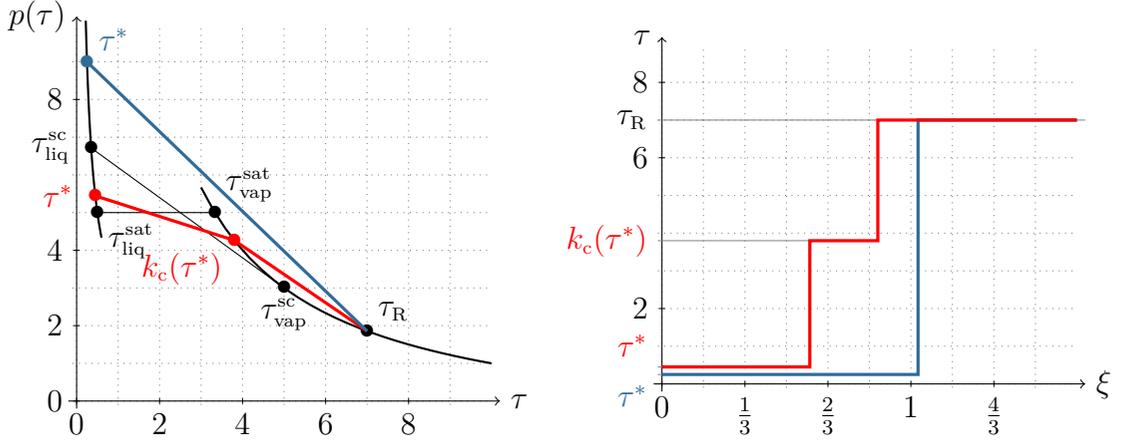}
	\newcommand{\taur}{7.0}
	
	\newcommand{\taured}{0.45}
	\newcommand{\kred}{3.8}
	
	\newcommand{\taublue}{.25}
	
	\newcommand{\scvap}{5}	
	\newcommand{\scliq}{0.35}
	
	\begin{tikzpicture}[xscale=0.55, yscale=0.5]

	\draw[thin,color=gray,dotted] (-0.1,-0.1) grid (9.9,9.9);

	\draw[->] (-0.2,0) -- (10.2,0) node[right] {$\tau$};
	\draw[->] (0,-0.2) -- (0,10.2) node[left] {$p(\tau)$};

	\foreach \x/\xtext in {0/0, 2/2, 4/4, 6/6, 8/8}
	    \draw[shift={(\x,0)}] (0pt,2pt) -- (0pt,-2pt) node[below] {$\xtext$};
	\foreach \y/\ytext in {0/0, 2/2, 4/4, 8/8}
	    \draw[shift={(0,\y)}] (2pt,0pt) -- (-2pt,0pt) node[left] {$\ytext$};

	\draw[thick] plot[domain=0.22:0.6] ({\x}, {pliq(\x)} ); 
	\draw[thick] plot[domain=3:10] (\x, {pvap(\x)} );
	
	\draw (\sl, {pliq(\sl)} ) node {\textbullet} node [below right] { $\tau_\liq^\sat$ } --
	      (\sv, {pvap(\sv)} ) node {\textbullet} node [above right] { $\tau_\vap^\sat$ };
	\draw (\scliq, {pliq(\scliq)} ) node {\textbullet} node [ left,xshift=-4pt] { $\tau_\liq^\sc$ } --
	      (\scvap, {pvap(\scvap)} ) node {\textbullet} node [below] { $\tau_\vap^\sc$ };
	\coordinate (R) at (\taur,{pvap(\taur)});
	\draw (R) node {\textbullet} node [above right] { $\tau_\R$ };

	\draw[very thick, color = red ] 
	  (R) -- (\kred,{pvap(\kred)} ) node {\textbullet} node [below left] {$k_\c(\tau^\ast)$} -- 
	  (\taured,{pliq(\taured)} )
	  node {\textbullet} 
	  node [ left,xshift=-4pt] {$\tau^\ast$};
	  
	\draw[very thick, color = blue] 
	  (R) -- 
	  (\taublue,{pliq(\taublue)} )
	  node {\textbullet} 
	  node [above right] {$\tau^\ast$};

	\end{tikzpicture}
	\newcommand{\stime}{6}
	\begin{tikzpicture}[xscale=0.55,yscale=0.5]

	\draw[thin,color=gray,dotted] (-0.1,-0.1) grid (9.9,8.9);

	\draw[->] (-0.2,0) -- (10.2,0) node[right] {$\xi$};
	\draw[->] (0,-0.2) -- (0,9.2) node[left] {$\tau$};

	\foreach \x/\xtext in {0/0, 2/\frac{1}{3}, 4/\frac{2}{3}, 6/1, 8/\frac{4}{3}}
	    \draw[shift={(\x,0)}] (0pt,2pt) -- (0pt,-2pt) node[below] {$\xtext$};
	\foreach \y/\ytext in { 2/2, 6/6, 8/8}
	    \draw[shift={(0,\y)}] (2pt,0pt) -- (-2pt,0pt) node[left] {$\ytext$};
	
	\draw [very thin, gray] (10.2,    \taur) -- (-0.1,    \taur)  node [black, left] { $\tau_\R$ };
	\draw [very thin, gray] ( 5.0,    \kred) -- (-0.1,    \kred)  node [red, left] { $k_\c(\tau^\ast)$};
	\draw [very thin, gray] ( 3.0,  \taured) -- (-0.1,  \taured)  node [black, above left, red ] { $\tau^\ast$ };
	\draw [very thin, gray] ( 6.0, \taublue) -- (-0.1, \taublue)  node [black, below left, blue] { $\tau^\ast$};

	\draw [very thick, color = blue  ]
	  (0,\taublue) -- 
	  ({ptspeed(\taublue,\taur)*\stime},\taublue) -- 
	  ({ptspeed(\taublue,\taur)*\stime} , \taur ) -- 
	  (10, \taur) ;

	\draw [very thick, color = red  ]
	  (0,\taured) -- 
	  ({ptspeed(\taured,\kred)*\stime},\taured) -- 
	  ({ptspeed(\taured,\kred)*\stime},\kred) -- 
	  ({vapspeed(\kred,\taur)*\stime} , \kred ) --
	  ({vapspeed(\kred,\taur)*\stime} , \taur ) --
	  (10, \taur) ;
	
	\end{tikzpicture}
	\caption{Pressure function (left) and specific volume distribution (right), like Figure~\ref{fig:KL1}.
	The red curve corresponds to wave type~\krwave{6} and the blue curve to wave type~\krwave{5}.}   \label{fig:KL2b}
\end{figure}

The Lax curves $v^\ast = v_\R + \mathcal{L}_2(\tau^\ast,\tau_\R)$ of the second family are given in Table~\ref{tab:KL2} and the main properties are summarized in the proposition below. Examples of wave type~\krwave{3} and wave type~\krwave{4} are shown in Figure~\ref{fig:KL2}, while Figure~\ref{fig:KL2b} shows an example of waves type~\krwave{5} and wave type~\krwave{6}. 
\begin{proposition}[Properties of the generalized Lax curve $\mathcal{L}_2$]\label{prop:KL2}
  Let a right state $(\tau_\R, v_\R)^\transp \in \admis_\vap\times \setR$ and 
  the map $\mathcal{L}_2 : (\tau_\liq^\tmin,\tau_\liq^\sat]\cup \admis_\vap\times \admis_\vap \to \setR$ of Table~\ref{tab:KL2} be given. 
  Then the following properties hold.
  \begin{enumerate}
    \item The map $\mathcal{L}_2$ is continuous.
    \item The map 
	\begin{align*}
	&(\tau_\liq^\tmin, \tau_\liq^\tmax)\cup \admis_\vap  \to \setR , &
	&\tau^\ast \mapsto  v^\ast = v_\R   + \mathcal{L}_2(\tau^\ast, \tau_\R)
	\end{align*}
	is differentiable and strictly monotone decreasing in $(\tau_\liq^\tmin, \tau_\liq^\tmax)$ and in $\admis_\vap$.
  
     \item It holds that $\mathcal{L}_2(\tau_\liq^\sat , \tau_\R) = \mathcal{L}_2(\tau_\vap^\sat , \tau_\R)$.
     
     \item All propagation speeds are positive. In wave \krwave{3}, \krwave{4} and \krwave{6}, the phase boundary propagates slower than the elementary wave in the vapor phase. 
     
     \item In wave \krwave{2} and wave \krwave{5} appear supersonic condensation waves. 
  
  \end{enumerate}  
\end{proposition}
\begin{proof}
 (\rmnum{1}) The map $\mathcal{L}_2$ is piecewise continuous and it is readily checked with Table~\ref{tab:KL2}, that also the transition from one domain of definition to another one is continuous.
 
 (\rmnum{2}) Note that $\mathcal{L}_2$ is piecewise smooth. The critical point in the transition of wave type~\krwave{2} to wave type~\krwave{3} is $\tau^\ast = \hat{\tau}$, 
 in the transition of wave type~\krwave{3} to wave type~\krwave{4} it is $\tau^\ast = \tau_\liq^\sc$ and from type~\krwave{5} to type~\krwave{6} it is $\tau^\ast=\check{\tau}$.
 For later use we derive 
 \begin{align*}
   &\ddfrac{\Pvel}{\tau}(\tau,g(\tau)) = \frac{(g'(\tau)-1) \, \s_\c(\tau,g(\tau))}{2} + \frac{c^2(g(\tau))\,g'(\tau) - c^2(\tau)}{2 \, \s_\c(\tau,g(\tau))},\\
   &\ddfrac{\Svel}{\tau}(g(\tau),\tau_\R) = -\frac{g'(\tau) \, \s_2(g(\tau),\tau_\R)}{2} - \frac{c^2(g(\tau))\,g'(\tau)}{2 \, \s_2(g(\tau), \tau_\R)}
 \end{align*}
 for some smooth function $g$ with $\tau<g(\tau)<\tau_\R$, the sound speed $c$ in \eqref{eq:H6}, the propagation speed $\s_\c$ in \eqref{eq:transitionspeed}. The bulk shock speed $\s_2$ is determined by 
 \begin{align*}
    \s_2(\tau_\l,\tau_\r )=+\sqrt{\frac{  -p(\tau_\r )+ p(\tau_\l  ) }{\tau_\r -\tau_\l  }}.  
 \end{align*}
 Furthermore, there holds $\ddfrac{R}{\tau}(g(\tau),\tau_\R) = -c(g(\tau))\,g'(\tau)$.
 
 We first check the limit $\tau^\ast \to \hat{\tau}$ and $\tau_\R\in(\tau_\vap^\tmin,\tau_\vap^\sc]$. 
 Note that $g_s(\hat{\tau})=\tau_\R$ and $\s_\c(\hat{\tau},\tau_\R)=c(\tau_\R)$ with Lemma~\ref{lm:gs}. 
 We use the above derivative with $g=g_s$ to find
 \begin{align*}
  &\lim_{\tau^\ast\to\hat{\tau}} \ddfrac{\Pvel}{\tau^\ast} (\tau^\ast, \tau_\R) = \frac{-1}{2} \left( c(\tau_\R)+\frac{c^2(\hat{\tau})}{c(\tau_\R)} \right), \,
  \lim_{\tau^\ast\to\hat{\tau}} \ddfrac{R}{\tau^\ast} (g_s(\tau^\ast), \tau_\R)   = - g_s'(\hat{\tau})\, c(\tau_\R),\\  
  &\lim_{\tau^\ast\to\hat{\tau}} \ddfrac{\Pvel}{\tau^\ast} (\tau^\ast, g_s(\tau^\ast)) = \frac{-1}{2} \left( c(\tau_\R)+\frac{c^2(\hat{\tau})}{c(\tau_\R)} \right) + g_s'(\hat{\tau})\, c(\tau_\R).
 \end{align*}
 Thus, the derivatives of a wave of type~\krwave{2} and a wave of type~\krwave{3} coincide in $\tau^\ast = \hat{\tau}$. 
 
 Now we check the limit $\tau^\ast \to\tau_\liq^\sc$ at $\tau_\R\in(\tau_\vap^\tmin,\tau_\vap^\sc]$.
 Here, it holds $k_\c(\tau_\liq^\sc) = g_s(\tau_\liq^\sc)=\tau_\vap^\sc$ and $\s_\c(\tau_\liq^\sc, \tau_\vap^\sc)= c(\tau_\vap^\sc)$ with Definition~\ref{def:kinfun}. 
 In wave type~\krwave{4}, we find
 \begin{align*}
  &\lim_{\tau^\ast\to\tau_\liq^\sc} \ddfrac{\Pvel}{\tau^\ast} (\tau^\ast, k_\c(\tau^\ast)) = 
     \frac{-1}{2} \left( c(\tau_\vap^\sc )+\frac{c^2(\tau_\liq^\sc)}{c(\tau_\vap^\sc)} \right) + c(\tau_\vap^\sc)\, k_\c'(\tau_\liq^\sc) ,  \\
  & \lim_{\tau^\ast\to\tau_\liq^\sc}\ddfrac{R}{\tau^\ast}(k_\c(\tau^\ast),\tau_\R) = 
    \lim_{\tau^\ast\to\tau_\liq^\sc}\ddfrac{\Svel}{\tau^\ast}(k_\c(\tau^\ast),\tau_\R)=  -c(\tau_\vap^\sc)\,k_\c'(\tau_\liq^\sc),   
 \end{align*}
 such that $\lim_{\tau^\ast \to\tau_\liq^\sc} \ddfrac{\Pvel}{\tau^\ast} \mathcal{L}_2(\tau^\ast,\tau_\R) = -1/2 \left( c(\tau_\vap^\sc )+c^2(\tau_\liq^\sc)/c(\tau_\vap^\sc) \right)$.
 The same holds for wave type~\krwave{3} replacing $k_\c$ by $g_s$. Thus, the derivatives coincide in $\tau^\ast = \tau_\liq^\sc$. 
 
 Finally, we have to check the limit $\tau^\ast \to \check{\tau}$ and $\tau_\R\in(\tau_\vap^\sc, \infty)$. 
 With Lemma~\ref{lm:hat}, it holds $\s_\c(\check{\tau}, k_\c(\check{\tau}) ) = \s_2(k_\c(\check{\tau}), \tau_\R) = \s_\c(\check{\tau}, \tau_\R )$. With above derivatives, we find that the limits from both sides (type~\krwave{5} and type~\krwave{6}) are 
 \begin{align*}
  &\lim_{\tau^\ast\to\check{\tau}} \ddfrac{\mathcal{L}_2}{\tau^\ast} (\tau^\ast, \tau_\R) = \frac{-1}{2} \left( \s_\c(\check{\tau},\tau_\R)+\frac{-c^2(\check{\tau})}{\s_\c(\check{\tau},\tau_\R)} \right).
 \end{align*}
 
 Monotonicity: the functions $\Evel$ and $\Pvel$ are strictly decreasing with respect to the first argument, thus for wave type~\krwave{1}, type~\krwave{2} and type~\krwave{5}, there is nothing to do. 
 
 Consider $\ddfrac{\mathcal{L}_2}{\tau^\ast} (\tau^\ast, \tau_\R)$ in case of wave type~\krwave{3}. All terms with $g_s'$ cancel out since $\s_\c(\tau^\ast,g_s(\tau^\ast))=c(g_s(\tau^\ast))$ holds. The remaining terms are negative such that $\mathcal{L}_2(\cdot, \tau_\R)$ is a strictly decreasing function. The same holds for wave type~\krwave{4} with $k_\c(\tau^\ast)>\tau_\R$. The wave is composed of a condensation wave and an attached 2-rarefaction wave, cf.\ wave type~\krwave{3}, and all terms with $k_\c'$ cancel out.
 
 In wave type~\krwave{4} with $k_\c(\tau^\ast)<\tau_\R$ and type~\krwave{6}, the function $k_\c$ is monotonously decreasing and the term $\s_\c+c^2(\tau^\ast)/\s_\c$ is positive. Thus, it remains to demonstrate that
 \begin{align*}
  \s_\c(\tau^\ast,k_\c(\tau^\ast)) + \frac{c^2(k_\c(\tau^\ast))}{\s_\c(\tau^\ast,k_\c(\tau^\ast))} - \s_2(k_\c(\tau^\ast),\tau^\ast_\R) - \frac{c^2(k_\c(\tau^\ast))}{\s_2(k_\c(\tau^\ast),\tau_\R)} \geq0.
 \end{align*}
 We skip the dependencies and rearrange the inequality: $(\s_2-\s_\c) \, \left( \frac{c^2}{\s_\c\,\s_2}-1 \right) \geq0$. This is true since the speeds in waves of type~\krwave{4} and type~\krwave{6} satisfy $c>\s_2\geq\s_\c$. Thus, $\mathcal{L}_2(\cdot, \tau_\R)$ is a strictly decreasing function.
 
 (\rmnum{3}) The condition holds due to $\Pvel(\tau_\liq^\sat,\tau_\vap^\sat)=0$.
 
 (\rmnum{4}), (\rmnum{5}) By definition, all waves of the second family have non-negative propagation speeds. The propagation speed of sonic and subsonic condensation waves is less than the sound speed in the vapor. The (supersonic) condensation wave in waves of type~\krwave{2} propagates faster than sound.

\end{proof}

The solution of the two-phase Riemann problem exists, if the two generalized Lax curves from Proposition~\ref{prop:KL1} and Proposition~\ref{prop:KL2} intersect each other. 

\begin{theorem}[Existence and uniqueness of two-phase Riemann solutions] \label{theo:kinrel}
    Let a pair of monotone decreasing kinetic functions $k_\c$, $k_\e$ as in Definition~\ref{def:kinfun} be given. 
    
    For any pair of states $(\tau_\L,v_\L)^\transp\in\admis_\liq\times \setR$ and $(\tau_\R,v_\L)^\transp\in\admis_\vap\times\setR$ the equation 
    \begin{align}  \label{eq:kinrelroot}
      v_\L + \mathcal{L}_1(\tau_\L,\tau^\ast) = v_\R + \mathcal{L}_2(\tau^\ast,\tau_\R)
    \end{align}   
    with $\mathcal{L}_1$ and $\mathcal{L}_2$ due to Table~\ref{tab:KL1} and Table~\ref{tab:KL2}, respectively,
    has a unique intersection point in 
    $(\tau_\liq^\tmin, \tau_\liq^\sat] \cup (\tau_\vap^\sat, \infty)$.
    
    The corresponding Riemann solution $\U =\left(\tau(\xi, t),v(\xi, t)\right)^\transp\in\admis\times\setR$ is a unique self-similar entropy solution, composed of rarefaction waves, shock waves and exactly one admissible phase boundary as in Definition~\ref{def:admissible}.
    The function $\U$ is composed of a wave connecting the left initial state with $(\tau^\ast,v^\ast)^\transp$ according to Table~\ref{tab:KL1} and a wave connecting $(\tau^\ast,v^\ast)^\transp$ to the right initial state according to Table~\ref{tab:KL2}, with $v^\ast = v_\L + \mathcal{L}_1(\tau_\L,\tau^\ast) = v_\R + \mathcal{L}_2(\tau^\ast,\tau_\R)$. 

\end{theorem}

Note that the solution contains exactly one phase boundary and subsonic phase boundaries are preferred, whenever this is possible. Both conditions are needed for uniqueness. 
Otherwise, Riemann solutions with, e.g., three phase boundaries are possible or a single supersonic evaporation wave instead of wave type~\klwave{3} would also be admissible.

Moreover, the two-phase Riemann solution depends continuously on the initial data. This has been proven in \cite{COLPRI2003} for initial data in stable phases.

\begin{proof}[Proof of Theorem~\ref{theo:kinrel}]
  First we see that $\tau^\ast \notin (\tau_\liq^\sat,\tau_\vap^\sat)$, such that we can exclude this interval from our consideration. 
  
  The Lax curves satisfy
  \begin{align*}
    \lim_{\tau\to\tau_\liq^\tmin} \mathcal{L}_1(\tau_\L, \tau) = - \infty, &&   \lim_{\tau\to\tau_\liq^\tmin} \mathcal{L}_2(\tau, \tau_\R) = + \infty,\\
    \lim_{\tau\to\infty         } \mathcal{L}_1(\tau_\L, \tau) = + \infty, &&   \lim_{\tau\to\infty         } \mathcal{L}_2(\tau, \tau_\R) = - \infty.
  \end{align*}
  Set $\Delta := \tau_\vap^\sat - \tau_\liq^\sat$. Proposition~\ref{prop:KL1} and Proposition~\ref{prop:KL2} ensure, that the function
  \begin{align*}
    f(\tau) = \left\{
	      \begin{matrix}
                   v_\R - v_\L & + \mathcal{L}_2(\tau, \tau_\R)        & - \mathcal{L}_1(\tau_\L, \tau)         & \text{for } \tau \leq \tau_\liq^\sat     \\
                   v_\R - v_\L & + \mathcal{L}_2(\tau-\Delta, \tau_\R) & - \mathcal{L}_1(\tau_\L, \tau-\Delta)  & \text{for } \tau  > \tau_\vap^\sat
              \end{matrix}
              \right.
  \end{align*}
  is continuous and strictly monotone decreasing from $+\infty$ to $-\infty$. Thus, $\tau^\ast\in (\tau_\liq^\tmin, \infty)$ exists such that $f(\tau^\ast)=0$. If $\tau \neq \tau_\liq^\sat$ then $\tau^\ast$ resp. $\tau^\ast+\Delta$ is the unique solution of \eqref{eq:kinrelroot}. If $\tau = \tau_\liq^\sat$, then also $\tau = \tau_\vap^\sat$ solves \eqref{eq:kinrelroot}.

  The existence of a unique Riemann solution follows from the existence of a unique intersection point of the Lax curves in Proposition~\ref{prop:KL1} and Proposition~\ref{prop:KL2}.
\end{proof}

\subsection{Algorithm and an illustrating example} \label{sub:solveralgo}

We are now able to define the two-phase Riemann solver for a properly defined pair of monotone decreasing kinetic functions $k_\c$, $k_\e$. The \defemph{two-phase Riemann solver} is a mapping of type
\begin{align}\label{eq:microm}
  \left\{   
  \begin{array}{rcl}
  \admis_\liq\times\setR \times \admis_\vap\times\setR \times \surfset &\to&
  \admis_\liq\times\setR \times \admis_\vap\times\setR \times\setR\\
    ( \tau_\L  , v_\L   , \tau_\R  , v_\R  , \surf)   &\mapsto& (\tau_\liq, v_\liq,\tau_\vap, v_\vap,\s),
  \end{array}
  \right.
\end{align}
which map the initial conditions \eqref{eq:ic} and the constant surface tension term $\surf$ ($:=(d-1)\,\surfcoeff\,\kappa$) to the end states and the speed
of the phase boundary. In this way, it is used in Section~\ref{sec:numerics}.

\begin{algo}[Two-phase Riemann solver]  \label{alg:kinrelsolver}
   Let the arguments $(\tau_\L, v_\L, \tau_\R, v_\R, \surf)$
   of mapping \eqref{eq:microm} be given.
  \begin{description}
    \item[Step~1] Determine the points in $\admis_\liq\cup\admis_\vap$ where the solution can alter its structure. 
    That are $\tau_{\liq/\vap}^\sat$ due to \eqref{eq:H3},
    $\tau_{\liq/\vap}^\se$, $\tau_{\liq/\vap}^\sc$ due to Definition~\ref{def:kinfun} and
    $\hat{\tau}$, $\check{\tau}$ due to Lemma~\ref{lm:hat}.
    \item[Step~2] Find $\tau^\ast \in (\tau_\liq^\tmin, \tau_\liq^\sat] \cup (\tau_\vap^\sat, \infty)$, that solves \eqref{eq:kinrelroot}.
    \item[Step~3] Return $(\tau_\liq, v_\liq,\tau_\vap, v_\vap,\s)$:
    \begin{itemize}
     \item In case of $\tau^\ast \in (\tau_\liq^\tmin, \tau_\liq^\sat]$, 
     the values $\tau_\liq$, $\tau_\vap$ are given in the last two columns of Table~\ref{tab:KL1}.
     The velocities are $v_\liq = v_\L + \Evel(\tau_\L  ,\tau_\liq)$ and $v_\vap = v_\liq + \Pvel(\tau_\liq, \tau_\vap)$
     and the speed is $\s=\s_\e(\tau_\liq, \tau_\vap)$.

     \item In case of $\tau^\ast \in (\tau_\vap^\sat, \infty)$, 
     the values $\tau_\liq$, $\tau_\vap$ are given in the last two columns of Table~\ref{tab:KL2}.
     The velocities are $v_\vap = v_\R - \Evel(\tau_\vap  ,\tau_\R)$ and $v_\liq = v_\vap - \Pvel(\tau_\liq, \tau_\vap)$
     and the speed is $\s=\s_\c(\tau_\liq, \tau_\vap)$.
     
    \end{itemize}
  \end{description}

\end{algo}

Note that Step~2 requires explicit knowledge of the kinetic functions.
We close the section with an illustrating example of rather simple kinetic functions.

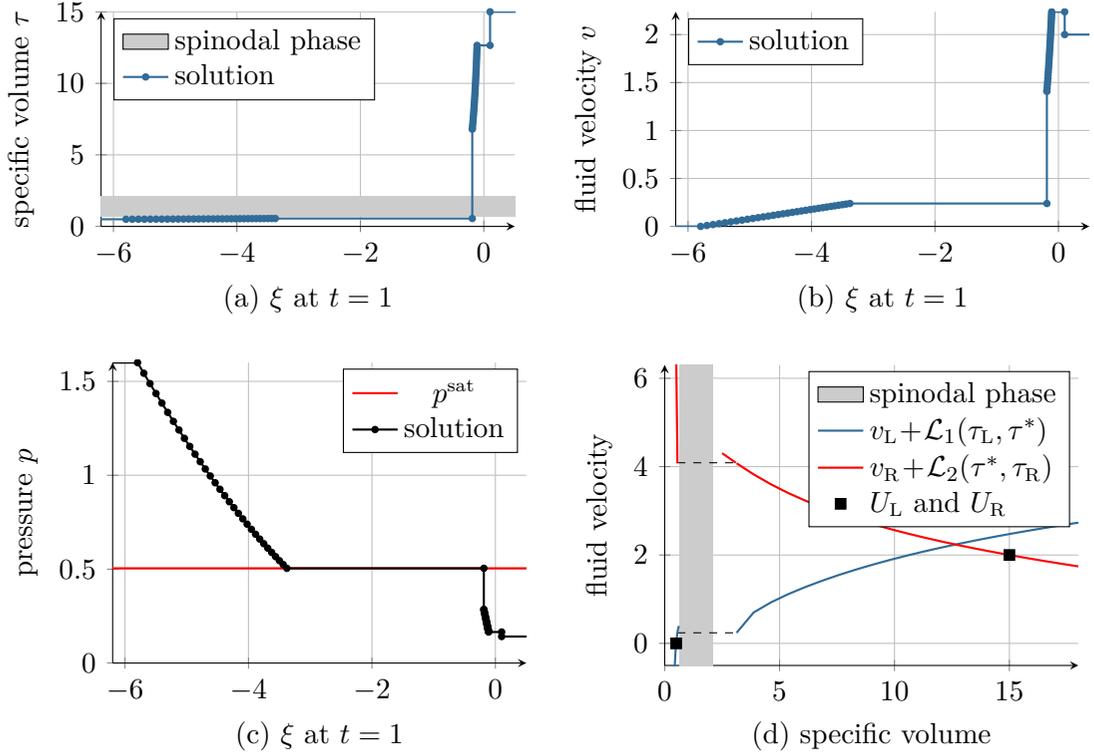
\begin{figure}\centering
 \newcommand{\tauliqsp}{0.65}
 \newcommand{\tauvapsp}{2.1}
 \newcommand{\psat}{0.5045}
  \begin{tikzpicture}
    \begin{axis}[xmin=-6.2, xmax=0.5,ymin=0, xlabel={(a) $\xi$ at $t = 1$},ylabel={specific volume $\tau$}, x post scale = 0.8, y post scale=0.5, legend pos=north west,legend cell align=left]
      \addplot[fill=lightgray,draw=none,area legend] coordinates { (-7,\tauliqsp) (1,\tauliqsp) (1,\tauvapsp) (-7,\tauvapsp) };
      \addlegendentry{spinodal phase}
      \addplot[thick,mark=*,mark size=1pt,blue] 	table[x expr={\thisrowno{0}}, y expr={\thisrowno{1}}, col sep=comma] {fanLiu.txt};
      \addlegendentry{solution}
    \end{axis}
  \end{tikzpicture}  
  \hspace{\hfloatsep}
  \begin{tikzpicture}
    \begin{axis}[xmin=-6.2, xmax=0.5,ymin=0, xlabel={(b) $\xi$ at $t = 1$},ylabel={fluid velocity $v$}, x post scale = 0.8, y post scale=0.5, legend pos=north west,legend cell align=left]
    \addplot[thick,mark=*,mark size=1pt,blue] 	table[x expr={\thisrowno{0}}, y expr={\thisrowno{2}}, col sep=comma] {fanLiu.txt};
    \addlegendentry{solution}
    \end{axis}
  \end{tikzpicture}    
  
  \vspace{\floatsep} 
  \begin{tikzpicture}
    \begin{axis}[xmin=-6.2, xmax=0.5,ymin=0, xlabel={(c) $\xi$ at $t = 1$},ylabel={pressure $p$}, x post scale = 0.8, y post scale=0.7]
      \addplot[color=red,thick] coordinates { (-7,\psat)  (1,\psat)  };
      \addlegendentry{$p^\sat$};    
      \addplot[thick,mark=*,mark size=1pt] 	table[x expr={\thisrowno{0}}, y expr={\thisrowno{3}}, col sep=comma] {fanLiu.txt};
      \addlegendentry{solution};
    \end{axis}
  \end{tikzpicture}  
  \hspace{\hfloatsep}
  \begin{tikzpicture}
    \begin{axis}[xmin=0, xmax=18,ymin=-0.5, ymax=6.3, xlabel={(d) specific volume},ylabel={fluid velocity}, x post scale = 0.8, y post scale=0.7,legend cell align=left]
      \addplot[fill=lightgray,draw=none,area legend] coordinates { (\tauliqsp,-1) (\tauliqsp,7) (\tauvapsp,7) (\tauvapsp,-1) };
      \addlegendentry{spinodal phase}   
    \addplot[dashed,forget plot] 	coordinates { (0.553, 0.239) (3.1, 0.239) };  
    \addplot[dashed,forget plot] 	coordinates { (0.553, 4.088) (3.1, 4.088) };  
    \addplot[color = blue,thick] 	table[x expr={\thisrowno{0}}, y expr={\thisrowno{1}}, col sep=comma] {curveKinLiu.txt};
    \addlegendentry{$v_\L \!+\! \mathcal{L}_1(\tau_\L  ,\tau^\ast)$}
    \addplot[color = red,thick] 	table[x expr={\thisrowno{0}}, y expr={\thisrowno{2}}, col sep=comma] {curveKinLiu.txt};
    \addlegendentry{$v_\R \!+\! \mathcal{L}_2(\tau^\ast,\tau_\R)$}
    \addplot[mark=square*,only marks]  coordinates { (0.5,0) (15,2) };
    \addlegendentry{$U_\L$ and $U_\R$}
    \end{axis}
  \end{tikzpicture}  
  \caption{Riemann solution of Example~\ref{exp:kin:1}. The figures show (a)~the specific volume, (b)~the velocity and (c)~the pressure over the Lagrangian space variable at time $t=1$. 
  The Lax curves of the first (blue) and the second (red) family are drawn in figure (d). The initial states are marked with a square.}
  \label{fig:kin:1}
\end{figure}

\begin{example}[Riemann solution and Lax curves]  \label{exp:kin:1}
 Consider the initial conditions $\vv U_\L=(0.5,0)^\transp$ and $\vv U_\R=(15,2)^\transp$, the van der Waals pressure of Example~\ref{exp:vdw}, $\surf=0$ and the following pair of monotone decreasing kinetic functions
 \begin{align*}
  k_\c(\tau_\liq) &= \tau_\vap^\sat, &  k_\e(\tau_\vap) &= \tau_\liq^\sat &   & \text{for all} & \tau_\liq\in[\tau_\liq^\sc , \tau_\liq^\sat], \tau_\vap\in[\tau_\vap^\sat, \tau_\vap^\se ].
 \end{align*}

 Figure~\ref{fig:kin:1} shows the solution, composed of a wave of type~\klwave{3} and type~\krwave{1} in Table~\ref{tab:KL1} and Table~\ref{tab:KL2}. 
 Waves of type~\klwave{3} consist of a rarefaction wave, followed by an evaporation wave and an attached rarefaction wave. A wave of type~\krwave{1} is solely a shock wave.  
 Figure (d) shows that the monotone increasing Lax curve of the first family intersects the monotone decreasing Lax curve of the second family in the point $(\tau^\ast, v^\ast)^\transp \approx (12.65, 2.24)^\transp$. 
 
\end{example}

Many kinetic functions are only implicitly available. This issue will be considered in the next section.

\section{Kinetic relations and kinetic functions for two-phase Riemann solvers} \label{sec:kinrel}

Pairs of monotone decreasing kinetic functions have been introduced in the last section in order to determine unique Riemann solutions. The more general form of an algebraic coupling condition to overcome the lack of well-posedness of the mixed hyperbolic-elliptic problem is a kinetic relation \cite{ABEKNO2006,TRU1993}. Kinetic relations provide an implicit condition to single out admissible phase boundaries. 
We will distinguish very clearly between kinetic relations, kinetic functions and, in particular, pairs of monotone decreasing kinetic functions, such that Theorem~\ref{theo:kinrel} applies.

Abeyaratne \& Knowles \cite{ABEKNO1990} and Hantke \& Dreyer \& Warnecke \cite{HANTKE2013} apply kinetic relations directly in order to construct Riemann solutions. However, their approaches require piecewise linear pressure functions and are not applicable to equations of state in the sense of Definition~\ref{def:thermo}.
The aim of this section is to derive criteria, which guarantee that a kinetic relation corresponds to a pair of (monotone decreasing) functions, see Theorem~\ref{theo:ex:kinfun} and Theorem~\ref{theo:ex:pair:kinfun}.


In the literature \defemph{kinetic relations} have been suggested (see  \cite{ABEKNO2006,TRU1993}), which control the entropy dissipation explicitly. In terms of a general form these are given by either
\begin{align} \label{eq:kinrel:f}
 K=K(f,\s) &:= f -g(\s) = 0 & 
 &\text{or}&
 K=K(f,\s) &:= h(f) -\s = 0 
\end{align} 
with continuous functions $g, h:\setR\to \setR$, the speed of the phase boundary $\s$, and a driving force $f$ in terms of the traces. 
Note that if $g$ is injective, then $h$ is just $g^{-1}$. Notably, there are examples with non-invertible $g$ or $h$, see for instance $K_1, K_5, K_8$ in Table~\ref{tab:kinrel} below.

Let the speed $\s$ given by the formulas \eqref{eq:transitionspeed}, and define the \defemph{driving force} $f:\admis_\liq \times \admis_\vap \mapsto \setR$ by 
\begin{align} \label{eq:f}
f(\tau_\liq,\tau_\vap) &= \jump{\psi(\tau) }+\jump{\tau} \mean{p(\tau)} + \surf\mean{\tau}.
\end{align}

The kinetic relation imposes a condition on the interfacial entropy production. 
The relation of \eqref{eq:kinrel:f} to entropy consistency can be seen as follows. 
Multiplying \eqref{eq:kinrel:f} by $\s$ or $f$ 
one obtains $-\s \, f = -g(\s)\,\s$ respectively  $-\s \, f = -h(f)\,f$.
This is related to the entropy jump condition \eqref{eq:entropyjump}, where the functions $g$, $h$ with  
\begin{align*} 
g(\s)\,\s \geq 0, && h(f)\,f \geq 0 
\end{align*}
determine the amount of entropy that is dissipated.

The connection between kinetic relations and kinetic functions is given by the following theorems.
Kinetic functions are applied only to subsonic phase boundaries, the same holds for kinetic relations.
The white area in Figure~\ref{fig:Ksets} illustrates admissible end states of subsonic phase boundaries, i.e.\ the set
\begin{align*}
\admis_\text{pb}:= \Set{(\tau_\liq,\tau_\vap)\in\admis_\liq\times\admis_\vap | 
                        p'(\tau_\liq), p'(\tau_\vap) \leq \frac{\jump{p}-\surf}{\jump{\tau}}, 
                       \jump{p}\geq \surf  }  . 
\end{align*}
We use to Lagrangian coordinates and equations of state as in Definition~\ref{def:thermo}. 

\begin{theorem}[Existence and uniqueness of kinetic functions]\label{theo:ex:kinfun} 
  Let $\s_\c:\admis_\text{pb}\to[0,\infty)$ and $\s_\e: \admis_\text{pb}\to(-\infty,0]$ be the propagation speed of condensation and evaporation waves, as defined in \eqref{eq:transitionspeed}, and let a Lipschitz continuous kinetic relation $K:\setR\times\setR \to \setR$ as in \eqref{eq:kinrel:f} be given.
  Assume that $K$ fulfills $K(0,0)=0$ and 
  \begin{align}\label{eq:K:ex}
    \pfrac{K}{f}\left(f(\tau_\liq,\tau_\vap),\s_{\c/\e}(\tau_\liq,\tau_\vap)\right) - \pfrac{K}{\s}\left(f(\tau_\liq,\tau_\vap),\s_{\c/\e}(\tau_\liq,\tau_\vap)\right)\frac{1}{\jump{\tau}^2}\sqrt{\abs{\frac{\jump{\tau} }{\surf- \jump{p}}}} > 0
  \end{align}
  for almost all $(\tau_\liq,\tau_\vap)\in\admis_\text{pb}$.
  
  Then there exist values $\tau_\liq^\sc \in [\tau_\liq^\tmin,\tau_\liq^\sat)$, $\tau_\vap^\se \in (\tau_\vap^\sat, \infty]$ and two continuous functions \linebreak
  $k_\c:  (\tau_\liq^\sc, \tau_\liq^\sat] \to \admis_\vap$ and 
  $k_\e:  [\tau_\vap^\sat, \tau_\vap^\se) \to \admis_\liq$ with 
  $K\left(f(\tau_\liq,k_\c(\tau_\liq),\s_\c(\tau_\liq,k_\c(\tau_\liq))\right)=0$ and 
  $K\left(f(k_\e(\tau_\vap),\tau_\vap),\s_\e(k_\e(\tau_\vap),\tau_\vap)\right)=0$.
  
\end{theorem}  
Note that either $\pslash{K}{f}=1$ and $\pslash{K}{\s}=-g'(\s)$ or $\pslash{K}{f}=h'(f)$ and $\pslash{K}{\s}=-1$.
Driving force and propagation speed are zero for the end states $(\tau_\liq^\sat, \tau_\vap^\sat)$. The condition $K(0,0)=0$ guarantees then, that the saturation states are a solution of the kinetic relation \eqref{eq:kinrel:f}. 
The Riemann solver of Section~\ref{sec:solver} requires kinetic functions \textbf{with} monotonic decay. 
The subsequent theorems state corresponding necessary conditions for the kinetic relations.

\begin{figure}\centering
 \begin{tikzpicture}[ xscale=0.4, yscale=0.4]
 
 \coordinate (S) at (8,7);
 
 \pgfmathdeclarefunction{f}{1}{\pgfmathparse{(  sqrt(#1-6.8)-0.5 )}}
 \coordinate (s0) at ( 7.05, {f( 7.05)});
 \coordinate (s1) at ( 8.0, {f( 8.0)});
 \coordinate (s2) at (10.0, {f(10.0)});
 \coordinate (s3) at (12.0, {f(12.0)});
 \coordinate (s4) at (15.0, {f(15.0)});
 \coordinate (s5) at (17.0, {f(17.0)});
 \coordinate (s6) at (20.0, {f(20.0)});
 \coordinate (s7) at (24.0, {f(24.0)});
 
 \fill [ pattern color=gray, pattern=north east lines]  (0,10.7) -- 
                  plot [smooth, tension=1.1] coordinates {(0,3) (S) (24,9)}  -- 
                  (24,10.7) -- (0,10.7)-- cycle;
 \draw [thick] plot [smooth, tension=1.1] coordinates {(0,3) (S) (24,9)} node [right,align=center] {$\s=0$};
 \draw [thick,green] plot [smooth, tension=1.1] coordinates {(0,1.8) (8,6) (24,8)} node [below,align=left] {$\s_\c=1$,\\$\s_\e=-1$};
 \draw [thick,blue] plot [smooth, tension=1.1] coordinates {(0,4.2) (8,8) (24,9.7)} node [above right,align=left,fill=white] {$\bar\s_\c=-1$,\\$\bar\s_\e=1$};
 
 \node[right,rectangle,fill=white] at (1,9) {$\s\notin\setR$};
 
 \filldraw [black] (S) circle (4pt) ;

 \fill [ fill=lightgray]  (24,0) --  
         plot [smooth] coordinates {(s0) (s1) (s2) (s3) (s4) (s5) (s6) (s7)} -- cycle;
 \node[right, align=center] at (8.8,3.5) {$\admis_\text{pb}$};
 \node[right, align=center,orange] at (16,6) {$f<0$};
 \node[right, align=center,orange] at (3,2) {$f>0$};
 \node[right] at (18,1) {supersonic};
 
 \draw [thick,orange] (S) .. controls (9,5) and ($(s4) - (3,0)$) .. (s4)  node [above,align=center] {$f=0$};
 \draw [thick,orange] (S) .. controls (7.5,8) and (7,9) .. (7,10.7) ;

 \draw [thick,red] plot [smooth] coordinates {(9.5,10.7) (S) (7.5,3) (s2)} node [below right,align=center] {$\Kr_\c=0$};
 \draw [thick,magenta] plot [smooth] coordinates {(4,10.7) (S) (17,4) (20,4.5) (24,4.5)} node [right,align=center] {$\Kr_\e=0$}; 
 
 \draw [dashed]  (8,-0.1) node [below] {$\tau_\vap^\sat$}  -- (S) -- (-0.1,7)  node [left] {$\tau_\liq^\sat$} ;
 \draw[->] (0,-.3) node [below] {$\tau_\vap^\tmin$} -- (0,11) node(yline)[above right] {$\tau_\liq$};
 \draw[->] (-.3,0) node [left] {$\tau_\liq^\tmin$} -- (24,0) node(xline)[right] {$\tau_\vap$};
 \draw    (-.1,10.7) node [left] {$\tau_\liq^\tmax$} -- (24,10.7) ;
 
\end{tikzpicture}
\caption{The Figure shows the set $\admis_\vap\times\admis_\liq$. The gray area corresponds to states, which lead to supersonic phase boundaries. The white area refers to the set $\admis_\text{pb}$. The shaded area corresponds to complex values of the functions $\s_\c$, $\s_\e$. The driving force $f$ is zero along the orange curve, positive on the left side of that curve and negative on the right side. 
The sound speed is zero along the black curve.
}
\label{fig:Ksets}
\end{figure}
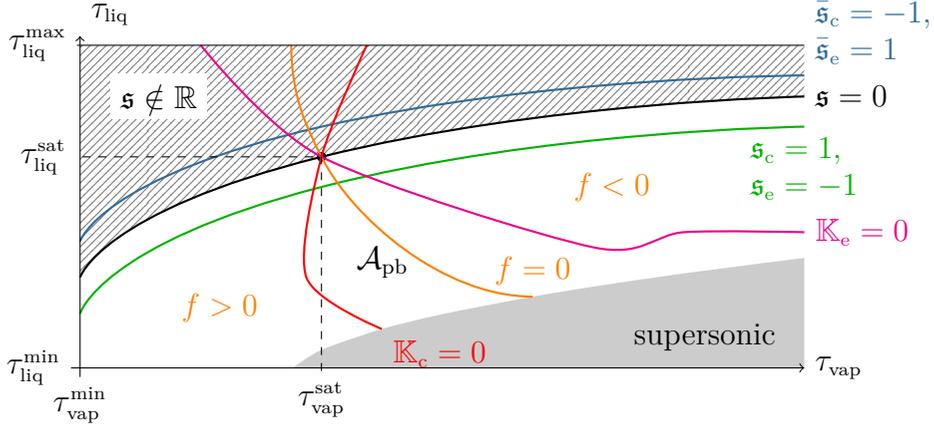

\begin{theorem}[Pairs of monotone decreasing kinetic functions]\label{theo:ex:pair:kinfun}   
  Let a kinetic relation $K:\setR\times\setR \to \setR$ be given, that fulfills the conditions of Theorem~\ref{theo:ex:kinfun}.
  If, in addition, $K$ is differentiable in $\setR\!\times\!\setR\setminus\set{(0,0)}$ and 
  \begin{align}\label{eq:K:mon}
    \pfrac{K}{f}\left(f(\tau_\liq,\tau_\vap),\s_{\c/\e}(\tau_\liq,\tau_\vap)\right) + \pfrac{K}{\s}\left(f(\tau_\liq,\tau_\vap),\s_{\c/\e}(\tau_\liq,\tau_\vap)\right) \frac{1}{\jump{\tau}^2}\sqrt{\abs{\frac{\jump{\tau} }{\surf- \jump{p}}}} \geq 0
  \end{align}  
  holds for all $(\tau_\liq,\tau_\vap)\in\mathring\admis_\text{pb}$, then a pair of monotone decreasing kinetic functions in the sense of Definition~\ref{def:kinfun} exists uniquely.
\end{theorem}  

The proof of Theorem~\ref{theo:ex:kinfun} requires a variant of the implicit function theorem, that does not require $\mathcal{C}^1$-smoothness.
\begin{theorem}[Implicit function theorem for continuous functions] \label{thm:implicit}
Suppose that $F : D \subset \setR \times \setR \to \setR$ is a continuous function with
\begin{align*}
  F(a_0, b_0) = 0. 
\end{align*}
Assume that there exist open neighborhoods $A \subset \setR$ and $B \subset \setR$ of $a_0$ and $b_0$, respectively, such that, for all 
$b \in B$, $F(\cdot, b) : A \subset \setR \to \setR$ is
injective. Then, for all $b \in B$, the equation
\begin{align*}
  F(a, b ) = 0 
\end{align*}
has a unique solution $a =H(b) \in A$, 
and the function $H : A \to \setR$ is continuous.
\end{theorem}
The theorem is proven in \cite{JITimplicit} for the more general case of functions  $F : D \subset \setR^n \times \setR^m \to \setR^n$ with $n,m\in\setN$.

\begin{proof}[Proof of Theorem~\ref{theo:ex:kinfun}]
Let us first extend the functions $\s_\c$ and $\s_\e$ to the domain 
\begin{align*}
\admis_\text{ext}:= \Set{(\tau_\liq,\tau_\vap)\in\admis_\liq\times\admis_\vap | 
                        p'(\tau_\liq), p'(\tau_\vap) \leq \abs{ \frac{\jump{p}-\surf}{\jump{\tau}} } }. 
\end{align*}
Define $\bar\s_\c(\tau_\liq,\tau_\vap)= \sign(\surf - \jump{p})\sqrt{\abs{ (\surf - \jump{p}) / \jump{\tau} } }$ and $\bar\s_\e(\tau_\liq,\tau_\vap)=-\bar\s_\c(\tau_\liq,\tau_\vap)$ for $(\tau_\liq,\tau_\vap)\in\admis_\text{ext}$.
Note that the pair of saturation states $(\tau_\liq^\sat, \tau_\vap^\sat)$ is an inner point of the set $\admis_\text{ext}$.
In Figure~\ref{fig:Ksets}, the set $\admis_\text{ext}$ is the union of the white area with the shaded area.

The following derivatives and monotonicity properties are readily checked
\begin{align*}
 \pfrac{f}{\tau_\vap} (\tau_\liq,\tau_\vap) &= \frac{1}{2} \left(p'(\tau_\vap) \jump{\tau} +\surf - \jump{p} \right) <0 \text{ in } \mathring\admis_\text{ext}, \\
 \pfrac{f}{\tau_\liq} (\tau_\liq,\tau_\vap) &= \frac{1}{2} \left(p'(\tau_\liq) \jump{\tau} +\surf - \jump{p} \right) <0 \text{ in } \admis_\text{ext}, \\[1ex]
 -\pfrac{\s_\c}{\tau_\vap} (\tau_\liq,\tau_\vap)  =  \pfrac{\s_\e}{\tau_\vap} (\tau_\liq,\tau_\vap)
   &= \frac{1}{2\jump{\tau}^2}\sqrt{\abs{\frac{\jump{\tau} }{\surf-\jump{p}}}}\, \left(p'(\tau_\vap) \jump{\tau} +\surf - \jump{p} \right) <0 \text{ in } \mathring\admis_\text{ext},\\
 \pfrac{\s_\c}{\tau_\liq} (\tau_\liq,\tau_\vap)  =  -\pfrac{\s_\e}{\tau_\liq} (\tau_\liq,\tau_\vap)
   &= \frac{1}{2\jump{\tau}^2}\sqrt{\abs{\frac{\jump{\tau} }{\surf-\jump{p}}}}\, \left(p'(\tau_\liq) \jump{\tau} +\surf - \jump{p} \right) <0 \text{ in } \admis_\text{ext}.
\end{align*}
The derivatives with respect to $\tau_\vap$ are zero in the sonic case, such that strict monotonicity holds only in the interior of the set $\admis_\text{ext}$.
The derivatives with respect to $\tau_\liq$ are negative in the sonic point.
The saturation states \eqref{eq:H3} satisfy $f(\tau_\liq^\sat,\tau_\vap^\sat)=0$ and $\s_{\e/\c}(\tau_\liq^\sat,\tau_\vap^\sat)=0$. 
Furthermore, the condition $K(0,0)=0$ ensures, that one solution is given by $(\tau_\liq^\sat,\tau_\vap^\sat)\in \admis_\text{ext}$.

We start with the condensation case and define a kinetic relation in terms of specific volume values via $\Kr_\c(\tau_\liq,\tau_\vap) := K(f(\tau_\liq,\tau_\vap), \bar\s_\c(\tau_\liq,\tau_\vap))$ for $(\tau_\liq,\tau_\vap)\in\admis_\text{ext}$. Figure~\ref{fig:Ksets} illustrates the set $\set{(\tau_\liq,\tau_\vap)\in\admis_\text{ext}| \Kr_\c(\tau_\liq,\tau_\vap)=0}$.
With \eqref{eq:K:ex} it holds
\begin{align*}
 \ddfrac{\Kr_\c}{\tau_\vap}(\tau_\liq,\tau_\vap) &= \pfrac{K}{f}\,\pfrac{f}{\tau_\vap}  + \pfrac{K}{\s} \pfrac{\s_\c}{\tau_\vap}  \\
 &= \frac{1}{2} \left(p'(\tau_\vap) \jump{\tau} +\surf - \jump{p} \right) \left( \pfrac{K}{f} - \pfrac{K}{\s}\frac{1}{\jump{\tau}^2}\sqrt{\abs{\frac{\jump{\tau} }{\surf- \jump{p}}}} \right) <0 
\end{align*}
for almost all $(\tau_\liq,\tau_\vap)\in\mathring\admis_\text{ext}$. There exists an open neighborhood $B_\liq\subset\admis_\liq$ of $\tau_\liq^\sat$ and an open neighborhood $B_\vap\subset\admis_\vap$ of $\tau_\vap^\sat$,  such that the function $\Kr_\c(\tau_\liq,\cdot): B_\vap \to \setR$ is strictly decreasing and injective for any $\tau_\liq\in B_\liq$. 
With Theorem~\ref{thm:implicit}, there exist a unique continuous function $k_\c: B_\liq \to B_\vap$, such that $\Kr_\c(\tau_\liq,k_\c(\tau_\liq))=0$ holds.
For values $\tau_\liq<\tau_\liq^\sat$ we can proceed with the same arguments, as long as, \eqref{eq:K:ex} holds. 
Finally, we restrict the domain of definition to values less or equal than $\tau_\liq^\sat$.

The evaporation case is very similar. For $\Kr_\e(\tau_\liq,\tau_\vap) := K(f(\tau_\liq,\tau_\vap), \bar\s_\e(\tau_\liq,\tau_\vap))$ it holds with \eqref{eq:K:ex}
\begin{align*}
 \ddfrac{\Kr_\e}{\tau_\liq}(\tau_\liq,\tau_\vap) &= \pfrac{K}{f}\,\pfrac{f}{\tau_\liq}  + \pfrac{K}{\s} \pfrac{\s_\e}{\tau_\liq}  \\
 &= \frac{1}{2} \left(p'(\tau_\liq) \jump{\tau} +\surf - \jump{p} \right) \left( \pfrac{K}{f} - \pfrac{K}{\s}\frac{1}{\jump{\tau}^2}\sqrt{\abs{\frac{\jump{\tau} }{\surf- \jump{p}}}} \right)<0
\end{align*}
for almost all $(\tau_\liq,\tau_\vap)\in\admis_\text{ext}$. 
The function $\Kr_\e(\cdot,\tau_\vap)$ is strictly decreasing and injective in an open neighborhood of the saturation states and we can apply the same arguments as above.

\end{proof}

\begin{proof}[Proof of Theorem~\ref{theo:ex:pair:kinfun}]
Due to Theorem~\ref{theo:ex:kinfun}, there are continuous kinetic functions $k_\c:  (\tau_\liq^\sc, \tau_\liq^\sat] \to \admis_\vap$ and $k_\e:  [\tau_\vap^\sat, \tau_\vap^\se) \to \admis_\liq$. 
The extra regularity assumption of differentiability is inherited to the kinetic functions.

We show, that $k_\c$ is a monotone decreasing function and use the functions defined in the proof of Theorem~\ref{theo:ex:kinfun}. From condition \eqref{eq:K:mon} it follows that
\begin{align*}
 \ddfrac{\Kr_\c}{\tau_\liq}(\tau_\liq,\tau_\vap) &= 
  \pfrac{K}{f}\,\pfrac{f}{\tau_\liq}  + \pfrac{K}{\s} \pfrac{\s_\c}{\tau_\liq}  \\
 &= \frac{1}{2} \left(p'(\tau_\liq) \jump{\tau} +\surf - \jump{p} \right) \left( \pfrac{K}{f} + \pfrac{K}{\s}\frac{1}{\jump{\tau}^2}\sqrt{\frac{\jump{\tau} }{\surf- \jump{p}}} \right)\leq0
\end{align*}
for $(\tau_\liq,\tau_\vap)\in\admis_\text{pb}$. 
We consider $\Kr_\c( \tau,k_\c(\tau))=0$ and derive 
\begin{align*}
  0 = \ddfrac{\Kr_\c}{\tau}\left( \tau,k_\c(\tau)\right) = 
  \pfrac{\Kr_\c}{\tau_\liq}\left( \tau,k_\c(\tau)\right) + 
  \pfrac{\Kr_\c}{\tau_\vap}\left( \tau,k_\c(\tau)\right) \,  k_\c'(\tau).
\end{align*}
The derivatives of the kinetic relation are both not positive and $\pfrac{\Kr_\c}{\tau_\vap}$ is negative except of the sonic point $\tau_\vap^\sc$, thus $k_\c'(\tau_\vap)\leq0$ for all $\tau_\vap\in[\tau_\vap^\sat,\tau_\vap^\sc)$.

We proceed with the evaporation case and show that the function $k_\e$ is monotone decreasing. 
From \eqref{eq:K:mon} it follows that
\begin{align*}
 \ddfrac{\Kr_\e}{\tau_\vap}(\tau_\liq,\tau_\vap) &= 
  \pfrac{K}{f}\,\pfrac{f}{\tau_\vap}  + \pfrac{K}{\s} \pfrac{\s_\e}{\tau_\vap}  \\
 &= \frac{1}{2} \left(p'(\tau_\vap) \jump{\tau} +\surf - \jump{p} \right) \left( \pfrac{K}{f} + \pfrac{K}{\s}\frac{1}{\jump{\tau}^2}\sqrt{\frac{\jump{\tau} }{\surf- \jump{p}}} \right) \leq0
\end{align*}
for $(\tau_\liq,\tau_\vap)\in\admis_\text{pb}$. 
Consider $\Kr_\e(k_\e(\tau), \tau)=0$ and derive 
\begin{align}\label{eq:ke:mon}
  0 = \ddfrac{\Kr_\e}{\tau}\left( k_\e(\tau),\tau\right) = 
  \pfrac{\Kr_\e}{\tau_\liq}\left( k_\e(\tau),\tau\right) \,  k_\e'(\tau) + 
  \pfrac{\Kr_\e}{\tau_\vap}\left( k_\e(\tau),\tau\right).
\end{align}
The term $\pfrac{\Kr_\e}{\tau_\liq}$ is negative and the term $\pfrac{\Kr_\e}{\tau_\vap}$ is not positive, thus $k_\e'\leq0$ in $(\tau_\liq^\sc,\tau_\liq^\sat]$.

It remains to show that the domain of definition can be extended up to the sonic points, with
$ \s_\c(\tau_\liq^\sc, k_\c(\tau_\liq^\sc)) = c(\tau_\vap^\sc)$, 
$-\s_\e(k_\e(\tau_\vap^\se), \tau_\vap^\se) = c(\tau_\vap^\se)$
and the condition $k_\e'(\tau_\vap^\se) =0$.
Note that the points $(\tau_\liq^\sc, k_\c(\tau_\liq^\sc))$, $(k_\e(\tau_\vap^\se), \tau_\vap^\se)$ are the intersection points of the kinetic functions with the boundary segment $p'(\tau_\vap) = \frac{\jump{p}-\surf}{\jump{\tau}}$, c.f.\ Figure~\ref{fig:Ksets}. The kinetic functions are monotone decreasing in $\admis_\text{pb}$. 
Thus, $k_\c$ and $k_\e$ intersect the boundary segment $p'(\tau_\vap) = \frac{\jump{p}-\surf}{\jump{\tau}}$ between the points 
$(\tau_\liq^\ast,\tau_\vap^\sat)$ and 
$(\tau_\liq^\sat,\tau_\vap^\ast)$. 
The points are the intersection points with a horizontal line $\set{(\tau_\liq,\tau_\vap)\in\admis_\text{pb}|\tau_\liq=\tau_\liq^\sat}$ and a vertical line $\set{(\tau_\liq,\tau_\vap)\in\admis_\text{pb}|\tau_\vap=\tau_\vap^\sat}$ through the saturation states.
The first intersection point exists due to Lemma~\ref{lm:hat} with $\tau_\R = \tau_\vap^\sat$ and $\tau_\liq^\ast = \hat\tau$.
The second intersection point exists due to Lemma~\ref{lm:gs} with  $\tau_\vap^\ast = g_s(\tau_\liq^\sat)$. 
Thus, also intersection points $(\tau_\liq^\sc, \tau_\vap^\sc)$ and $(\tau_\liq^\se,\tau_\vap^\se)$ exist, with $\tau_\vap^\sc :=k_\c(\tau_\liq^\sc)$, $\tau_\liq^\se:=k_\e(\tau_\vap^\se)$. 

Finally, we find that $\ddfrac{\Kr_\e}{\tau_\vap}(\tau_\liq^\se,\tau_\vap^\se)=0$ and $\ddfrac{\Kr_\e}{\tau_\liq}(\tau_\liq^\se,\tau_\vap^\se)>0$ in \eqref{eq:ke:mon}, that means $k_\e'(\tau_\vap^\se)=0$. 
Thus, $(k_\c,k_\e)$ is a pair of monotone decreasing kinetic functions.
\end{proof}

The applicability of kinetic relations that correspond to pairs of monotone decreasing functions is in fact limited. 
This is underlined by the following result (see also Subsection~\ref{sub:simoera}).

\begin{corollary}[Metastable phase boundaries]\label{col:metastable:waves}
  Consider a phase boundary \eqref{eq:discontinuous}, that obeys a kinetic relation as required in Theorem~\ref{theo:ex:pair:kinfun}.
  
  Then, the end states $\tau_\liq$, $\tau_\vap$ of the phase boundary belong to stable phases, i.e.\ $\tau_\liq\leq\tau_\liq^\sat$ and $\tau_\vap\geq\tau_\vap^\sat$. Thus, metastable end states are excluded.
\end{corollary}
\begin{proof}
  Due to Theorem~\ref{theo:ex:pair:kinfun} there is a pair of monotone decreasing kinetic function, such that the end states $\tau_\liq\in\admis_\liq$, $\tau_\vap\in\admis_\vap$ satisfy $k_\c(\tau_\liq)=\tau_\vap$ for a condensation wave and $k_\e(\tau_\liq)=\tau_\vap$ for an evaporation wave. One pair of end states is given by the saturation states, i.e.\  $k_\c(\tau_\liq^\sat)=\tau_\vap^\sat$ and $k_\e(\tau_\liq^\sat)=\tau_\vap^\sat$.
  Because of the monotonicity of $k_\c$ and $k_\e$ it holds that $\tau_\vap^\sat \leq \tau_\vap$ for $\tau_\liq \in [\tau_\liq^\sc, \tau_\liq^\sat]$ in the condensation case and that $\tau_\liq \leq \tau_\liq^\sat$ for $\tau_\vap \in [\tau_\vap^\sat, \tau_\vap^\se]$ in the evaporation case.
\end{proof}

\section{Examples of kinetic relations and two-phase Riemann solver for examples of kinetic relations}\label{sec:kinrelex}

We apply the theorems of the previous section to examples of kinetic relations, as they have been suggested in the literature. 
Furthermore, two-phase Riemann solutions are determined and studied with respect to different kinetic relations, but also with respect surface tension.
A comparison with experimental measurements from \cite{Simdodecane1999} is presented.

\subsection{Examples of kinetic relations}

\begin{table}
 {
  \renewcommand{\arraystretch}{1.3}
  \begin{tabularx}{\textwidth}{X}
    \toprule
  \end{tabularx} 
 }
 {
  \belowrulesep=0pt
  \aboverulesep=0pt
  \begin{tabularx}{\textwidth}{p{6.5cm}R{3.6cm}R{2.6cm}L}
      \textbf{kinetic relation}
      & corresponds to a pair
      & \multicolumn{2}{c}{the phase boundary}\\        \cmidrule(l){3-4}    
      & of monotone decreasing kinetic functions 
      & dissipates entropy (\small$f\,\s\neq0$) & is static ($\s=0$)\\ 
  \end{tabularx}
  }
\renewcommand{\arraystretch}{1.3}  
  \begin{tabularx}{\textwidth}{p{6.5cm}R{3.6cm}R{2.6cm}L}
      \cmidrule(r){1-1} \cmidrule(lr){2-2} \cmidrule(lr){3-3}\cmidrule(l){4-4}
      $K_1(f,\s) := f $                             & \cmark & \xmark & \xmark \\
      $K_2(f,\s) := f - k^\ast\,\s$                 & \xmark & \cmark & \xmark \\
      $K_3(f,\s) := f - k^\ast\, \sign(\s)\,\s^2$   & \cmark \quad for small $k^\ast>0$ & \cmark & \xmark \\
      $K_4(f,\s) := f - k^\ast\,  \s^3$             & \cmark \quad for small $k^\ast>0$ & \cmark & \xmark \\
      $K_5(f,\s) :=  \left\{
		\begin{array}{lll}
		  f+a & -k^\ast\,\s &: f<-a  \\
		      & -k^\ast\,\s &: \abs{f}\leq a\\
		  f-a & -k^\ast\,\s &: f>a              
	      \end{array}\right.$                     & \xmark & \cmark & \xmark \\
      $K_6(f,\s,\tau_\liq,\tau_\vap) := f - \sign(\s)\,\s^2\,\jump{\tau}^2$	   & \xmark & \cmark & \xmark \\      
      $K_7$ such that $ \begin{cases}
               k_\c(\tau_\liq) = \tau_\vap^\sat &: \s\geq0 \\
               k_\e(\tau_\vap) = \tau_\liq^\sat &: \s<0 
              \end{cases}$
						    & \cmark & \cmark & \xmark \\
      $K_8(f,\s) := -\s$			    & \xmark & \xmark & \cmark \\
    \bottomrule
  \end{tabularx}
  \caption{Different kinetic relations and properties, in particular the existence of a corresponding pair of monotone decreasing kinetic functions due to Theorem~\ref{theo:ex:pair:kinfun}. The parameters satisfy $k^\ast>0$, $a >0$.} \label{tab:kinrel}
\end{table}

Table~\ref{tab:kinrel} provides a list with examples for kinetic relations as they can be found in the literature \cite{ABEKNO2006,TRU1993,Kabil2016,HANTKE2013,BEN1998,BENFRE2004,BEN1999}.
Figure~\ref{fig:Kcontour} shows the zero contour lines of the kinetic relations and an equation of state as in Definition~\ref{def:thermo}. 
Figure~\ref{fig:Kcontour:dodecane} illustrates the same as Figure~\ref{fig:Kcontour} but in terms of the pressure and for an equation of state of n-dodecane at $T=\unit[230]{\degree C}$, computed by the library CoolProp \cite{CoolProp}. To be precise the set $\set{ (p(\tau_\vap),p(\tau_\liq))\in\setR_+^2 | K(\tau_\liq, \tau_\vap)=0}$ is shown.

We proceed with a description of the kinetic relations of Table~\ref{tab:kinrel}.

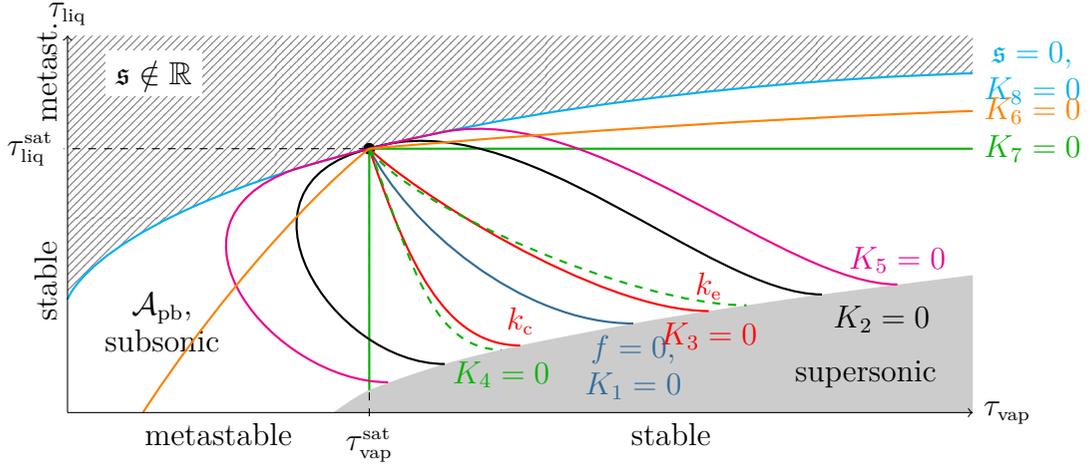
\begin{figure}\centering
 \begin{tikzpicture}[ xscale=0.5, yscale=0.5]
 
 \coordinate (S) at (8,7);
 
 \pgfmathdeclarefunction{f}{1}{\pgfmathparse{(  sqrt(#1-6.8)-0.5 )}}
 \coordinate (s0) at ( 7.05, {f( 7.05)});
 \coordinate (s1) at ( 8.0, {f( 8.0)});
 \coordinate (s2) at (10.0, {f(10.0)});
 \coordinate (s22)at ( 8.5, {f( 8.5)});
 \coordinate (s3) at (12.0, {f(12.0)});
 \coordinate (s32)at (11.5, {f(11.5)});
 \coordinate (s4) at (15.0, {f(15.0)});
 \coordinate (s5) at (17.0, {f(17.0)});
 \coordinate (s52)at (18.0, {f(18.0)});
 \coordinate (s6) at (20.0, {f(20.0)});
 \coordinate (s62)at (22.0, {f(22.0)});
 \coordinate (s7) at (24.0, {f(24.0)});
 
 \fill [ pattern color=gray, pattern=north east lines]  (0,10) -- 
                  plot [smooth, tension=1.1] coordinates {(0,3) (S) (24,9)}  -- 
                  (24,10) -- (0,10)-- cycle;
 \draw [thick, cyan] plot [smooth, tension=1.1] coordinates {(0,3) (S) (24,9)} node [right,align=center] {$\s=0$,\\$K_8=0$};
 \node[right,rectangle,fill=white] at (1,9) {$\s\notin\setR$};
 
 \filldraw [black] (S) circle (4pt) ;

 \fill [ fill=lightgray]  (24,0) --  
         plot [smooth] coordinates {(s0) (s1) (s2) (s3) (s4) (s5) (s6) (s7)} -- cycle;
 \node[right, align=center] at (0.7,2.4) {$\admis_\text{pb}$,\\subsonic};
 \node[right] at (19,1) {supersonic};
 
 \draw [thick, blue] (S) .. controls (9,5) and ($(s4) - (3,0)$) .. (s4)  node [below,align=center] {$f=0$,\\$K_1=0$};

 \draw [dashed]  (8,-0.1) node [below] {$\tau_\vap^\sat$}  -- (S) -- (-0.1,7)  node [left] {$\tau_\liq^\sat$} ;
 \draw [thick,green] (s1) -- (S) -- (24,7)  node [right] {$K_7=0$} ;
 \draw [thick] (S) .. controls (12,8.3) and ($(s6) - (3,0)$) .. (s6) node [below right] {$K_2=0$} ;
 \draw [thick] (s2) .. controls ($(s2) - (3,0)$) and (4,6) .. (S) ;

 \draw [thick,magenta] (10,7.45) .. controls (14,8.2) and ($(s62) - (3,0)$) .. (s62) node [above] {$K_5=0$} ;
 \draw [thick,magenta] (s22) .. controls ($(s22) - (3,0)$) and (2,5) .. (6,6.45) ; 
 \draw [thick,magenta] plot [smooth, tension=1.1] coordinates {(10,7.45) (S) (6,6.45)};
 
 \draw [thick, red] (S) .. controls (9,6) and ($(s5) - (3,0)$) .. (s5)   node [above] {$k_\e$} node [below] {$K_3=0$}; 
 \draw [thick, red]  (S) .. controls (9,4) and ($(s3) - (2,0)$) .. (s3)  node [above] {$k_\c$};

 \draw [thick,dashed,green] (S) .. controls (9,5.5) and ($(s52) - (3,0)$) .. (s52)  ;
 \draw [thick,dashed,green]  (S) .. controls (9,4.5) and ($(s32) - (2,0)$) .. (s32) node [below] {$K_4=0$};
 
 \draw [thick,orange] plot [smooth, tension=1.1] coordinates {(S) (16,7.6)  (24,8)} node[right] {$K_6=0$};
 \draw [thick,orange] plot [smooth, tension=1.1] coordinates {(S) (5,4)  (2,0)};
 
 \draw[->] (0,0) -- (0,10) node(yline)[above] {$\tau_\liq$};
 \draw[->] (0,0) -- (24,0) node(xline)[right] {$\tau_\vap$};
 
 \draw ( 4,0) node [below] {metastable};
 \draw (16,0) node [below] {stable};
 \draw (0,3.5) node [rotate=90,above] {stable};
 \draw (0,7) node [rotate=90,above right] {\, metast.};
\end{tikzpicture}
\caption{Zero contour lines of kinetic relations, i.e.\ $\set{(\tau_\liq,\tau_\vap)\in\admis_\text{pb}|K(\tau_\liq,\tau_\vap)=0}$, and a pair of monotone decreasing kinetic functions, i.e.\ $(\tau_\liq,k_\c(\tau_\liq))\subset \admis_\text{pb}$ and $(k_\e(\tau_\vap),\tau_\vap)\subset \admis_\text{pb}$ resulting from $K_3$. The shaded area corresponds to complex speeds $\s_\c$, $\s_\e$ and the gray area to supersonic phase boundaries. The white area corresponds to the set $\admis_\text{pb}$ of subsonic phase boundaries.}
\label{fig:Kcontour}
\end{figure}

\begin{figure} \centering  
\newcommand{\datafile}{KinrelDodecane.txt}
  \begin{tikzpicture}
    \begin{axis}[x post scale = 1.8, y post scale=1.4,
                 xmin=0.37, xmax=1.8,   
                 ymin=0.48, ymax=2.02,
                 xlabel={vapor pressure $p(\tau_\vap)$ in \unit{bar}}, ylabel={liquid pressure $p(\tau_\liq)$ in \unit{bar}},
                  legend pos=south east
                  ]
    \addplot[pattern color=gray, pattern=north east lines, area legend]  table[x expr={\thisrowno{15}}, y expr={\thisrowno{14}}, col sep=comma] {\datafile} \closedcycle;                  
    \addlegendentry{$\s\notin\setR$};
    \addplot[thick, blue]  table[x expr={\thisrowno{1}}, y expr={\thisrowno{0}}, col sep=comma] {\datafile};
    \addlegendentry{$f=0$, $K_1=0$};

    \addplot[thick,forget plot]  table[x expr={\thisrowno{3}}, y expr={\thisrowno{2}}, col sep=comma] {\datafile};
    \addplot[thick]  table[x expr={\thisrowno{5}}, y expr={\thisrowno{4}}, col sep=comma] {\datafile};
    \addlegendentry{$K_2=0$};
    
    \addplot[thick,dashed,forget plot]  table[x expr={\thisrowno{7}}, y expr={\thisrowno{6}}, col sep=comma] {\datafile};
    \addplot[thick,dashed]  table[x expr={\thisrowno{9}}, y expr={\thisrowno{8}}, col sep=comma] {\datafile};
    \addlegendentry{$K_\text{dft}=0$};

    \addplot[thick, red]  table[x expr={\thisrowno{11}}, y expr={\thisrowno{10}}, col sep=comma] {\datafile};
    \addlegendentry{$k_\c: K_3=0$};
    \addplot[thick, magenta]  table[x expr={\thisrowno{13}}, y expr={\thisrowno{12}}, col sep=comma] {\datafile};
    \addlegendentry{$k_\e: K_3=0$};
    
    \addplot[thick,orange,forget plot]      table[x expr={\thisrowno{17}}, y expr={\thisrowno{16}}, col sep=comma] {\datafile};
    \addplot[thick,orange]  table[x expr={\thisrowno{19}}, y expr={\thisrowno{18}}, col sep=comma] {\datafile};
    \addlegendentry{$K_6=0$};
    
    \addplot[thick,green] plot coordinates { (0.7123,1.39) (1.39,1.39) (1.39,2.2) }; 
    \addlegendentry{$K_7=0$};

    \addplot[thick, cyan]  table[x expr={\thisrowno{15}}, y expr={\thisrowno{14}}, col sep=comma] {\datafile} ;
    \addlegendentry{$\s=0$, $K_8=0$};
    
    \addplot[thick, gray, name path=sonic,forget plot]  table[x expr={\thisrowno{21}}, y expr={\thisrowno{20}}, col sep=comma] {\datafile};
    \path[name path=top] (axis cs:0,2.5) -- (axis cs:2,2.5);
    
    \addplot [
        thick,
        color=lightgray,
        fill=lightgray
    ]
    fill between[
        of=sonic and top,
    ];
    \addlegendentry{supersonic};
    
    \addplot[color=black,mark=*, only marks]  table[x expr={\thisrowno{3}}, y expr={\thisrowno{2}}, col sep=comma] {IsoEvaporationDodecaneRS.txt};
    \addlegendentry{measurements};

    \end{axis}
  \end{tikzpicture}    
  \caption{Kinetic relations for n-dodecane with respect to the pressure. Relation $K_2$ for $k^\ast=\unitfrac[5]{m^4}{kg\ s}$ and $K_\text{dft}$ as in Example~\ref{exp:dft}. Kinetic functions $k_\c$ and $k_\e$ result from $K_3$ with $k^\ast=\unitfrac[0.005]{m^6}{kg^2}$. The gray line marks sonic phase boundaries. The black dots mark measured values from the experiment in Subsection~\ref{sub:simoera}.} \label{fig:Kcontour:dodecane}
\end{figure}
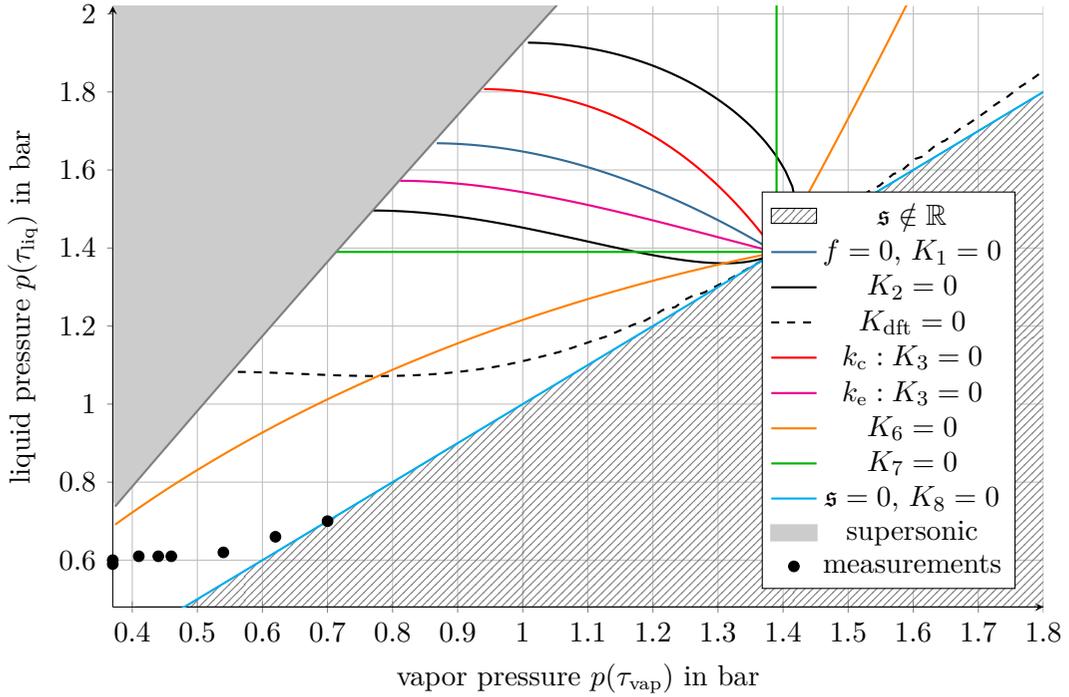

\subsubsection{$K_1$: kinetic relation with zero entropy dissipation}
The kinetic relation $K_1$ has been analyzed in \cite{BENFRE2004,BEN1998}.
We find $\pslash{K_1}{f}=1$, $\pslash{K_1}{\s}=0$ such that the conditions \eqref{eq:K:ex} and \eqref{eq:K:mon} hold. Due to Theorem~\ref{theo:ex:kinfun}, there is a pair of monotone decreasing kinetic functions $(k_\c,\, k_\e)$.

A phase boundary, that satisfies kinetic relation $K_1$, conserves the entropy since $\s\,f=0$ (cf.\ \eqref{eq:entropyjump}) and $k_\e$ is the inverse function of $k_\c$. From a thermodynamic point of view this can be interpreted as a reversible process.

\subsubsection{$K_2$, $K_3$ and $K_4$: kinetic relations with polynomial growth}
The kinetic relation $K_2$ has been suggested in \cite{TRU1993} and has been analyzed in \cite{BEN1999,Kabil2016}.
We find $\pslash{K_2}{f}=1$, $\pslash{K_2}{\s}=-k^\ast$ such that condition \eqref{eq:K:ex} holds for any $k^\ast > 0$ but \eqref{eq:K:mon} is not satisfied. The term $\sqrt{{\jump{\tau} }/{(\surf- \jump{p}})}=1/\abs{\s}$ is infinite in the saturation state ($\s=0$) and monotone decreasing in $\abs{\s}$.
Due to Theorem~\ref{theo:ex:kinfun}, kinetic functions exist but are not monotone decreasing for $k^\ast>0$.  
A pair of monotone decreasing kinetic functions exists only for $k^\ast = 0$ (kinetic relation $K_1$).
Otherwise the kinetic functions are not monotone, see also Figure~\ref{fig:Kcontour:dodecane}.

A specific choice of $k^\ast>0$ in $K_2$, that will lead to consistent results with physical experiments in Subsection~\ref{sub:simoera}, is given by the following example.
\begin{example}[Density functional theory and kinetic relation $K_\text{dft}$]\label{exp:dft}
  Density functional theory is used in \cite{GROSS1008} to compute resistivities for heat transfer and for mass transfer at vapor liquid interfaces.
  The authors assume a correlation between the interfacial mass flux and differences in the chemical potential that is similar to kinetic relation $K_2$. For isothermal one-component fluids the correlation reduces to $\jump{\mu} = -T\,R\,j$, where $R\geq0$ is called interfacial resistivity. That gives for \eqref{eq:iso:freeenergy}, \eqref{eq:f}, \eqref{eq:j-s} and $\surf=0$ the relation 
  \begin{align*}
    \jump{f} + \mean{\tau}\jump{p} = T\,R\,\s.
  \end{align*}
  Note that this is $K_2$ with $k^\ast=T\,R$ up to the term $\mean{\tau}\jump{p}$. The term vanishes in the equilibrium case \eqref{eq:equilibrium} and is small for slow phase boundaries, since $\abs{\s}=\sqrt{\jump{p}/\jump{\tau}}$.
  
  One finds values of $R$ for n-octane in \cite{klink2015analysis}. We assume that the fluids n-octane and n-dodecane behave similar, since both are alkanes. 
  The resistivity values are now used to estimate $k^\ast$ in kinetic relation $K_2$.
  For n-dodecane at \unit[230]{\degree C}, this results in the definition
  \begin{align}\label{eq:dft_dodecane}
   K_\text{dft}(f,\s) := f-\s\,k^\ast_\text{dft} &&\text{with}&& k^\ast_\text{dft}=\unitfrac[28]{m^4}{kg\ s}.
  \end{align}
  As a particular choice of $K_2$, the kinetic functions for $K_\text{dft}$ exist, but they are not monotone decreasing. Thus, Theorem~\ref{theo:kinrel} is not applicable.
  However, Subsection~\ref{sub:simoera} below shows that Riemann solutions can nevertheless be computed by Algorithm~\ref{alg:kinrelsolver}.
\end{example}

Kinetic relations $K_3$ and $K_4$ in Table~\ref{tab:kinrel} were chosen as further examples to satisfy the conditions of Theorem~\ref{theo:ex:pair:kinfun}. They lead to pairs of monotone decreasing kinetic functions, for sufficiently small $k^\ast>0$, and can be used for Algorithm~\ref{alg:kinrelsolver}. Kinetic relations $K_3$ and $K_4$ behave very similar and we consider only $K_3$ in the following.
Note that the parameter $k^\ast>0$ in the kinetic relations $K_2$, $K_3$ and $K_4$ has different physical units.

If one splits the contour lines at the saturation point into two branches, one finds the corresponding kinetic functions for evaporation $k_\e=k_\e(\tau_\vap)$ and condensation waves $k_\c=k_\c(\tau_\liq)$ respectively. In Figure~\ref{fig:Kcontour}, this is shown for $K_3$. Generally, kinetic functions for evaporation waves are located to the right of the curve $K_1=0$ and kinetic functions for condensation waves are located to the left of this curve. This is a consequence of the entropy inequality \eqref{eq:entropyjump}, since $f\,\s\geq0$ holds.

\subsubsection{$K_5$: a non-smooth kinetic relation with multiple static solutions} \label{sub:K5}
Kinetic relations like $K_5$ in Table~\ref{tab:kinrel} are often considered for phase boundaries in solid mechanics (see \cite[Section~4.4]{ABEKNO2006}). 
There, no transition takes place until the driving force $f$ passes a certain threshold $a>0$. If the driving force $f$ is sufficiently small, the phase boundary does not propagate. Note that this involves static phase boundaries, whose end states are not the saturation states.
    
The conditions of Theorem~\ref{theo:ex:kinfun} are satisfied, but for $\abs{f}<a$ condition \eqref{eq:K:mon} is violated. Thus, Theorem~\ref{theo:ex:pair:kinfun} does not apply.
In \cite{ABEKNO1990}, unique Riemann solutions are singled out assuming non-monotone pressure functions that are piecewise linear.

\subsubsection{$K_6$: limit case of a kinetic relation with maximal entropy dissipation}
Because there is no entropy dissipation for $K_1$ and $K_8$ (see Table~\ref{tab:kinrel}), since either $f=0$ or $\s=0$, the kinetic relation with the highest entropy release has to be searched somewhere in between.
The interfacial entropy production is given by the product $\s\,f$, see \eqref{eq:entropyjump}. 
We may derive a kinetic relation with the highest entropy release at constant $\tau_\liq$ or at constant $\tau_\vap$ related to the extreme value of $f(\tau_\liq, \tau_\vap)\,\s_{\c/\e}(\tau_\liq, \tau_\vap)$.
The conditions $\ddfrac{}{\tau_\vap} f\,\s_\c=0$ and $\ddfrac{}{\tau_\liq} f\,\s_\e=0$ lead to the relations
\begin{align*}
  f(\tau_\liq,\tau_\vap) + \s_\c(\tau_\liq,\tau_\vap)^2 \, \jump{\tau}^2 &= 0, &
  f(\tau_\liq,\tau_\vap) - \s_\e(\tau_\liq,\tau_\vap)^2 \, \jump{\tau}^2 &=0.
\end{align*}
Kinetic relation $K_6(f,\s,\tau_\liq,\tau_\vap) = f - \sign(\s) \, \s^2 \, \jump{\tau}^2$ takes both cases into account. Note that $K_6$ needs more arguments.
Figure~\ref{fig:Kcontour:dodecane} shows, that the corresponding kinetic functions for $K_6$ are monotone increasing, thus Theorem~\ref{theo:ex:pair:kinfun} is not applicable.

Note that this kinetic relation does not correspond to the energy rate admissibility criterion in \cite{DAFentropyrate1973,HAT1986}. There, entropy is minimized over a set of admissible Riemann solutions, while here it is minimized over a set of phase boundaries with one fixed end state.

\subsubsection{$K_7$: limit case of a kinetic relation that corresponds to Liu's entropy criterion}\label{sub:Liu}
    Godlewski \& Seguin solved in \cite{GODSEG2006} the one-dimensional two-phase Riemann problem for homogenized pressure laws applying the Maxwell equal area rule. For uniqueness they apply the entropy criterion of Liu \cite{LIU1974}. This was extended to the surface tension dependent case in \cite{JAEROHZER2012}. 
    
    In terms of Definition~\ref{def:thermo} the homogenized pressure law is given by
    \begin{align}\label{eq:pLiu}
     p^\surf &:(\tau_\liq^\tmin,\infty) \to \setR &   p^\surf(\tau) = 
       \begin{cases}
         p(\tau)+\surf & \text{if } \tau \in (\tau_\liq^\tmin, \tau_\liq^\sat ],\\
         p(\tau_\vap^\sat) & \text{if } \tau \in (\tau_\liq^\sat, \tau_\vap^\sat), \\
         p(\tau) & \text{if } \tau \in [\tau_\vap^\sat, \infty).
       \end{cases}
    \end{align}
    Note that $p^\surf$ depends on the surface tension term $\surf$ and that $p(\tau_\liq^\sat)+\surf = p(\tau_\liq^\sat)$, see \eqref{eq:H3}. 
    The two-phase Riemann problem with that pressure law and the entropy criterion of Liu implies a kinetic relation implicitly.
    All subsonic phase boundaries connect to one of the saturation states. This determines the kinetic functions
    \begin{align}\label{eq:kinfun:liu}
      k_\c: \left\{
      \begin{array}{cl}
         [\tau_\liq^\sc,  \tau_\liq^\sat] &\to \admis_\vap, \\
                   \tau_\liq \quad        &\mapsto \tau_\vap^\sat, 
      \end{array}
      \right.
      &&
      k_\e: \left\{ 
      \begin{array}{cl}
         [\tau_\vap^\sat, \tau_\vap^\sc ]&\to \admis_\liq, \\
         \tau_\vap \quad                 &\mapsto \tau_\liq^\sat.
      \end{array}
      \right.
    \end{align}    
    The corresponding kinetic relation is named $K_7$ in Table~\ref{tab:kinrel} and Figure~\ref{fig:Kcontour}.  
    For $K_7$, it is simpler to state the kinetic functions directly.
    The relation was already applied in Example~\ref{exp:kin:1}. 
    
    The kinetic functions are constant and can be seen as the limit case of monotone decreasing functions, since $k_\c'=0$ and $k_\e'=0$.
    They fulfill the conditions of Definition~\ref{def:kinfun} and Theorem~\ref{theo:kinrel} applies.
   
    Note that the difference between the Riemann solution with $K_7$ and the \defemph{Liu Riemann solution} from \cite{GODSEG2006,JAEROHZER2012} is the different underlying pressure function. The Liu Riemann solver uses the homogenized pressure \eqref{eq:pLiu} and not the pressure of Definition~\ref{def:thermo}, which is defined only for bulk phases.
    However, because of $p'(\tau) = p^{\surf\prime}(\tau)$ for $\tau\in(\tau_\liq^\tmin,\tau_\liq^\sat] \cup [\tau_\vap^\sat,\infty)$, both solutions are identical for initial states in stable phases. A proof of that statement can be found in \cite{Z2015}.

\subsubsection{$K_8$: limit case of a kinetic relation for static phase boundaries / zero mass flux}
    The limit case $K_2$ with $k^\ast \to \infty$ leads to the kinetic relation $K_8(f,\s) = -\s$, what means that no entropy is dissipated, since $\s\,f=0$ (cf.\ \eqref{eq:entropyjump}).
    Recall, that the case $k^\ast \to 0$ leads to $K_1$. 
    Theorem~\ref{theo:ex:kinfun} can be applied for $K_8$ but the corresponding kinetic functions are monotone increasing. 
    Phase boundaries, that obey $K_8$, satisfy $\s=0$, $v_\liq=v_\vap$, $p(\tau_\liq)=p(\tau_\vap)$. 
    The kinetic functions are given by 
    \begin{align}\label{eq:kinfunK8}
     k_\c(\tau_\liq) &= p_\vap^{-1}(p(\tau_\liq))& &\text{and}&
     k_\e(\tau_\vap) &= p_\liq^{-1}(p(\tau_\vap)),
    \end{align}
    where $p_\liq^{-1}$ is the inverse function of $p:\admis_\liq\to\setR$ and $p_\vap^{-1}$ is the inverse of $p:\admis_\vap\to\setR$.

    In Eulerian coordinates $\s=j=0$ (cf.\ \eqref{eq:j-s}) means, that there is no mass transfer between the phases. Such a phase boundary may represent material boundaries of different immiscible substances. 
    Riemann solvers for impermeable material boundaries can be found, e.g.\ in \cite{FJS2013}.

\begin{remark}[Entropy dissipation rate for evaporation waves and condensation waves]
  Kinetic relation $K_2,\ldots,K_5$ depend on the parameter $k^\ast$, that controls the amount of entropy dissipation. There is no physical reason why evaporation waves and condensation waves share the same value for $k^\ast$. The parameter could also depend on the sign of $\s$, but this is not considered here.  
\end{remark}

\subsection{Riemann solvers for non-decreasing kinetic functions}\label{sub:increasingkinfun}
  We presented several examples of kinetic relations, which lead to non-decreasing kinetic function, such that Theorem~\ref{theo:kinrel} is not applicable. However, it is remarkable that unique Riemann solutions may still exist, see e.g.\ $K_\text{dft}$ in Subsection~\ref{sub:simoera}. Further examples are $K_7$ and $K_8$:
  
  The arguments for $K_8$ are rather simple because the related generalized Lax curves are strictly monotone.
  The kinetic functions \eqref{eq:kinfunK8} are monotone increasing and such that the pressure is equal in both end states. 
  That means the value of the Lax curve is the same as in the metastable phase.
  More precisely 
  \begin{alignat*}{3}
   \mathcal{L}_1(\tau_\L,\tau_\vap) &= \mathcal{L}_1(\tau_\L,k_\e(\tau_\vap)) 
   &&\quad\text{ for }  \tau_\vap\geq \tau_\vap^\sat, &\text{ since } p(k_\e(\tau_\vap))&=p(\tau_\vap) \text{ and} \\
   \mathcal{L}_2(\tau_\liq,\tau_\R) &= \mathcal{L}_2(k_\c(\tau_\liq),\tau_\R) 
   &&\quad\text{ for }  \tau_\liq\leq \tau_\liq^\sat, &\text{ since } p(k_\c(\tau_\liq))&=p(\tau_\liq).
  \end{alignat*}
  The domain of definition for such Lax curves is restricted since we cannot expect that the pressure function provides for any pressure value in a stable phase a corresponding metastable volume value with the same pressure. Furthermore, attached waves are excluded due to the zero propagation speed of the phase boundary. However, as long as the Lax curves exist, they are monotone.
  
  For $K_7$, the corresponding kinetic functions \eqref{eq:kinfun:liu} are constant, what is related to the extreme case of a monotone function. But even in this case, the Lax curves are by far not constant (see Figure~\ref{fig:kin:1} (right)), that would be the crucial limit for monotonicity. We believe therefore, that considering monotone decreasing kinetic functions is too restrictive and not necessary for unique two-phase Riemann solutions.

\subsection{Comparative study of Riemann solutions obeying different kinetic relations} \label{sub:kin:examples}
We apply different kinetic relations, or the related pairs of kinetic functions, to the Riemann solver of Section~\ref{sec:solver}.
In order to distinguish two-phase Riemann solutions, we write \defemph{$K_n$-Riemann solution} if the contained phase boundary satisfies one of the kinetic relations $K_n$ in Table~\ref{tab:kinrel}. We will consider the kinetic relations $K_1$, $K_3$ and $K_7$. For them, Theorem~\ref{theo:kinrel} guarantees unique solvability.

\begin{example}[Influence of different kinetic relations]
This example illustrates the effect of different kinetic relations.
We use the van der Waals pressure of Example~\ref{exp:vdw} and initial conditions
\begin{align*}
  \U(\xi, 0) = \begin{cases}
               (0.57,0)^\transp & \text{for } \xi\leq0, \\
               (50,0)^\transp & \text{for } \xi >0,
             \end{cases}
\end{align*}
such that the liquid state is in the metastable phase. 
The solid lines in Figure~\ref{fig:exact:entropy} show Riemann solutions for $\surf=0$ and different kinetic relations. All solutions are composed of a shock wave followed by an evaporation wave with attached rarefaction wave and a shock wave. In terms of the notation in Table~\ref{tab:KL1} and Table~\ref{tab:KL2} the solution is composed of wave type~\klwave{3} and type~\krwave{6}. 
We see that the pressure in the liquid phase is higher for phase boundaries that dissipate more entropy, while the propagation speed becomes slower.

Furthermore, the example illustrates the difference to the Liu Riemann solution, which uses the homogenized pressure law \eqref{eq:pLiu}, see Subsection~\ref{sub:Liu}.
The Liu Riemann solution is plotted with a dashed line in Figure~\ref{fig:exact:entropy} and differs from the $K_7$-Riemann solution, since the liquid initial states is in the metastable phase. 

\end{example}

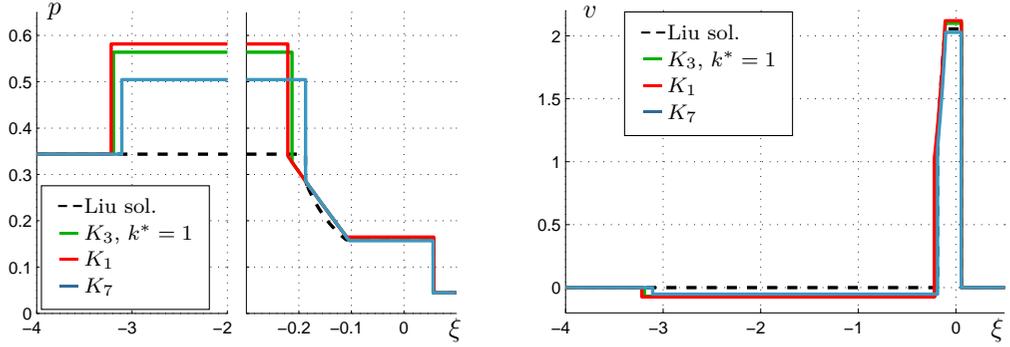
\begin{figure}\centering
\setlength{\includewidth}{7cm}
\begin{tikzpicture}[scale=0.9]
	\node [anchor=south west] (O) at (0,0) {\includegraphics[width=0.9\includewidth]{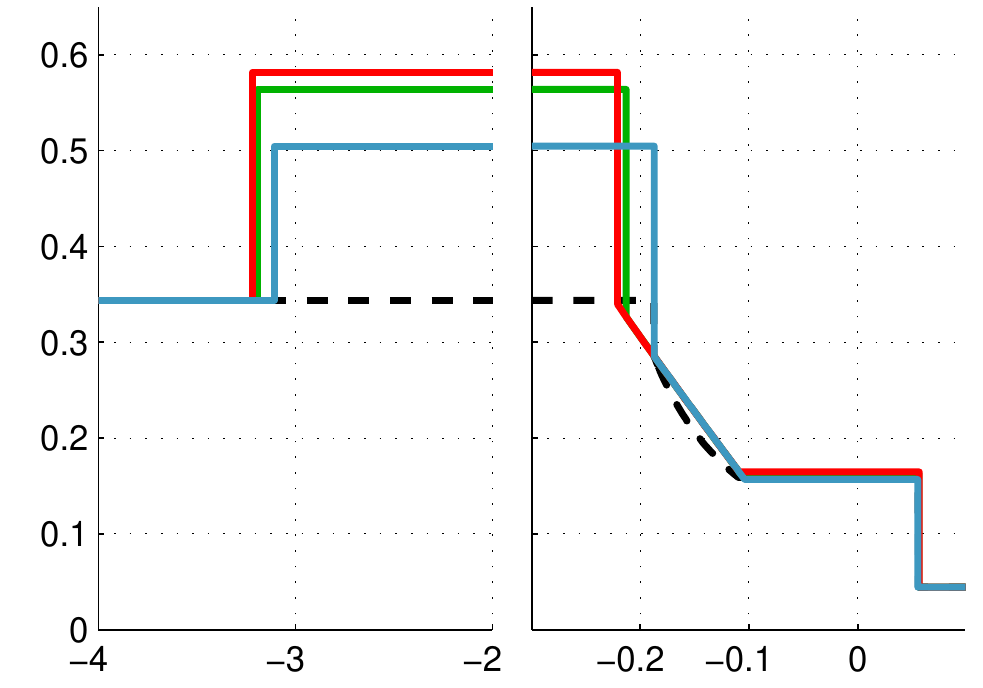}};
	\node at (7,0.3) {$\xi$};
	\node at (1.1,5){$p$};
	\node [above right, draw=black, fill=white] at (0.9,0.6){%
	  \begin{footnotesize}
	  \begin{tabular}{@{}c@{}l@{}}
	     \textcolor{black}{\markdashed} & Liu sol. \\	  
	     \textcolor{green}{\markline}   &  $K_3$, $k^\ast=1$ \\
	     \textcolor{red}{\markline}     &  $K_1$  \\
	     \textcolor{blue}{\markline}    &  $K_7$  
	   \end{tabular}
	  \end{footnotesize} 
	  };
\end{tikzpicture}
\hspace{\hfloatsep}
\begin{tikzpicture}[scale=0.9]
	\node [anchor=south west] (O) at (0,0) {\includegraphics[width=0.9\includewidth]{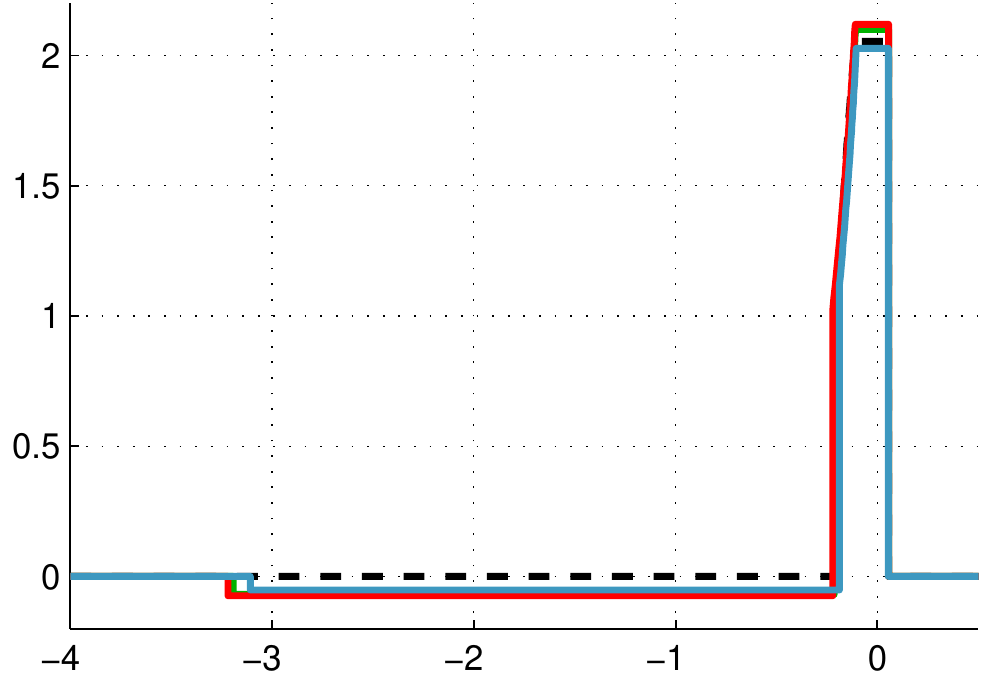}};
	\node at (7,0.3) {$\xi$};
	\node at (1,5){$v$};
	\node [below right, draw=black, fill=white] at (1.5,5){%
	  \begin{footnotesize}
	  \begin{tabular}{@{}c@{}l@{}}
	     \textcolor{black}{\markdashed} & Liu sol. \\	  
	     \textcolor{green}{\markline}   &  $K_3$, $k^\ast=1$ \\
	     \textcolor{red}{\markline}     &  $K_1$  \\
	     \textcolor{blue}{\markline}    &  $K_7$  
	   \end{tabular}
	  \end{footnotesize} 
	  };
\end{tikzpicture}
\caption{Liu Riemann solution (dashed line) and Riemann solution (solid lines) with different kinetic relations. The left figure shows the pressure and the right one the velocity as function of the Lagrangian space variable at time $t=1$. }            \label{fig:exact:entropy}
\end{figure}

\begin{example}[Static solutions and influence of the surface tension term $\surf$]\label{exp:exact:surface}
This example intends to check the basic property, that thermodynamic equilibrium solutions are preserved.
The saturation states $\tau_\liq^\sat\approx 0.55336$, $\tau_\vap^\sat\approx3.1276$ for the van der Waals pressure of Example~\ref{exp:vdw} with $\surf=0$ are used as initial states 
\begin{align*}
  \U_0(\xi) = \begin{cases}
               (\tau_\liq^\sat,0)^\transp & \text{for } \xi\leq0, \\
               (\tau_\vap^\sat,0)^\transp & \text{for } \xi >0
             \end{cases}
\end{align*}
and we apply the kinetic relation $K_7$.
The red line in Figure~\ref{fig:exact:surface} shows that the $K_7$-Riemann solution and initial condition are identical. 
Note that this holds for all kinetic relations in Table~\ref{tab:kinrel}, since $K(0,0)=0$ and $f(\tau_\liq^\sat,\tau_\vap^\sat)=0$, $\s_{\c/\e}(\tau_\liq^\sat,\tau_\vap^\sat)=0$.

That changes for $\surf\neq0$. Figure~\ref{fig:exact:surface} shows also the $K_7$-Riemann solution for $\surf = \pm 0.01$. 
The $K_7$-Riemann solution for $\surf =-0.01$ is a composition of a shock wave followed by an evaporating wave with speed $\s\approx -0.004$ and another shock wave, respectively a composition of wave type~\klwave{2} and \krwave{1}. 
For $\surf =0.01$ we find a rarefaction wave followed by a condensation wave with speed $\s\approx 0.004$ and another rarefaction wave, respectively a composition of wave type~\klwave{1} and \krwave{4}. 

One can interpret the examples with $\surf\neq0$ as considering a spherical bubble or droplet of the same radius with the same pressure and Gibbs free energy inside and outside. 
In both cases, the radius decreases in order to compensate the pressure difference due to the Young-Laplace law. Note that due to that law, the pressure inside a static bubble or droplet is higher than outside.
\end{example}

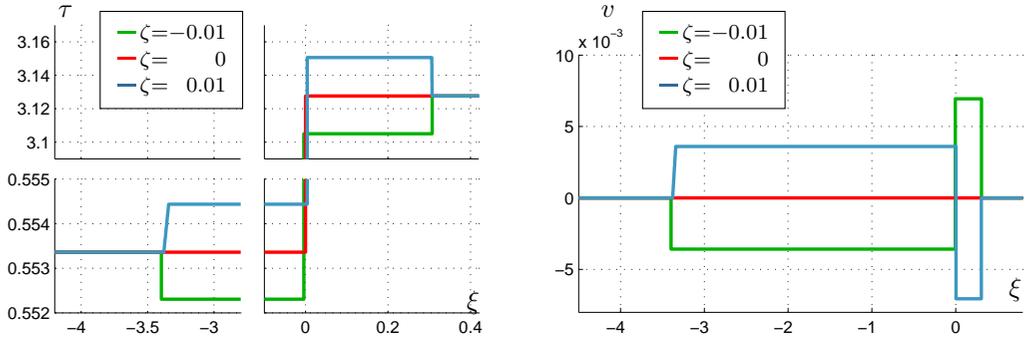
\begin{figure}\centering
\setlength{\includewidth}{7cm}
\begin{tikzpicture}[scale=0.9]
	\node [anchor=south west] (O) at (0,0) {\includegraphics[width=0.9\includewidth]{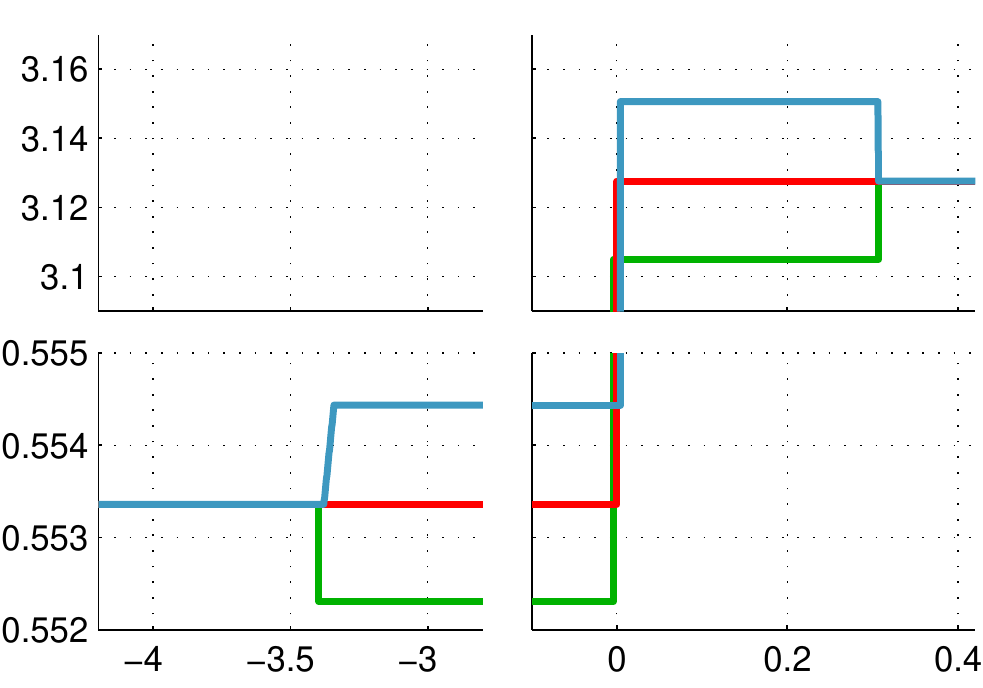}};
	\node at (7,0.7) {$\xi$};
	\node at (1,5){$\tau$};
	\node [below right, draw=black, fill=white] at (1.5,5){%
	  \begin{footnotesize}
	  $\begin{array}{@{}c@{}c@{}c@{}r@{}}
	     \textcolor{green}{\markline} &  \surf &=& -0.01 \\
	     \textcolor{red}{\markline} &  \surf &=& 0  \\
	     \textcolor{blue}{\markline} &  \surf &=& 0.01  
	   \end{array}$
	  \end{footnotesize} 
	  };
\end{tikzpicture}
\hspace{\hfloatsep}
\begin{tikzpicture}[scale=0.9]
	\node [anchor=south west] (O) at (0,0) {\includegraphics[width=0.9\includewidth]{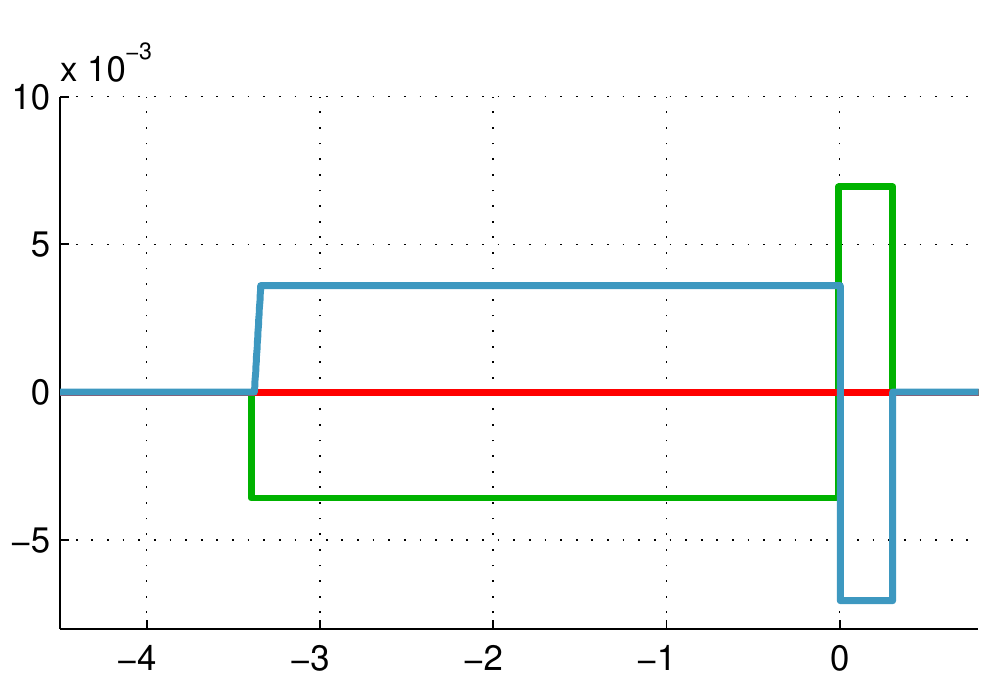}};
	\node at (7,0.9) {$\xi$};
	\node at (1,5){$v$};
	\node [below right, draw=black, fill=white] at (1.5,5){%
	  \begin{footnotesize}
	  $\begin{array}{@{}c@{}c@{}c@{}r@{}}
	     \textcolor{green}{\markline} &  \surf &=& -0.01 \\
	     \textcolor{red}{\markline} &  \surf &=& 0  \\
	     \textcolor{blue}{\markline} &  \surf &=& 0.01 
	   \end{array}$
	  \end{footnotesize} 
	  };
\end{tikzpicture}
\caption{$K_7$-Riemann solutions for different surface tension terms. The left figure shows the specific volume and the right one the velocity as function of the Lagrangian space variable at time $t=1$. }            \label{fig:exact:surface}
\end{figure}

\subsection{Validation with shock tube experiments}\label{sub:simoera}
We compare the Riemann solvers against the shock tube experiments of Simoes-Moreira \& Shepherd in \cite{Simdodecane1999}.
In their experiments liquid n-dodecane was relaxed into a low pressure reservoir. Initially, the liquid was at saturation pressure and the vapor pressure varies between almost vacuum and the saturation pressure. They observed stable evaporation fronts of high velocity.

We consider here only the series of experiments at constant temperature $T=\unit[230]{\degree C}$ and we compare the measured (planar) evaporation front speed of the experiment with data from Riemann solutions. We assume that the dissipation rate $k^\ast$ in kinetic relation $K_2$ or $K_3$ involves temperature\footnote{Density functional theory, cf.\ Example~\ref{exp:dft}, predicts temperature dependent resistivities.}. Thus, the isothermal series allows us to use the same value of $k^\ast$ for all test cases.

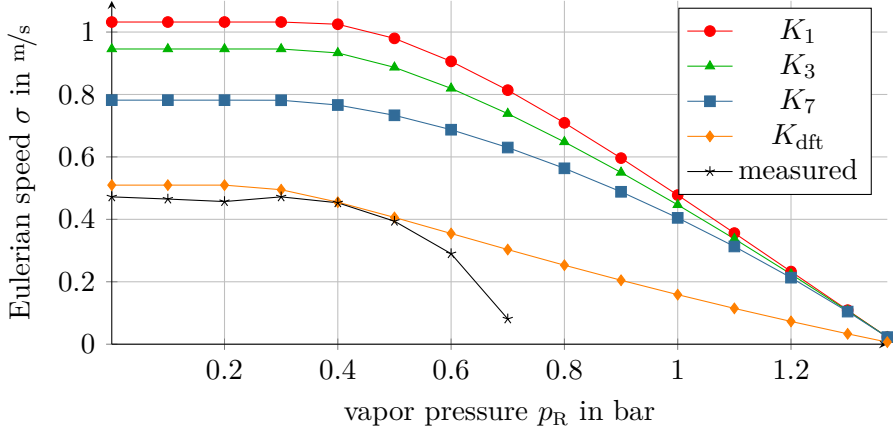
\begin{figure} \centering  
  \begin{tikzpicture}
    \begin{axis}[ x post scale = 1.5, y post scale=0.8,
                  ymin=0, ymax=1.1,
                  xlabel={vapor pressure $p_\R$ in \unit{bar}}, ylabel={Eulerian speed $\sigma$ in \unitfrac{m}{s}},
                  ytick={0,0.2,0.4,0.6,0.8,1},
                  ]
    \addplot[color=red,mark=*]  table[x expr={\thisrowno{1}}, y expr={\thisrowno{7}}, col sep=comma] {IsoEvaporationDodecaneRS.txt};
    \addlegendentry{$K_1$}   
    
    \addplot[color=green,mark=triangle*]  table[x expr={\thisrowno{1}}, y expr={\thisrowno{8}}, col sep=comma] {IsoEvaporationDodecaneRS.txt};
    \addlegendentry{$K_3$}       
    
    \addplot[color=blue,mark=square*]  table[x expr={\thisrowno{1}}, y expr={\thisrowno{9}}, col sep=comma] {IsoEvaporationDodecaneRS.txt};
    \addlegendentry{$K_7$}       
    
    \addplot[color=orange,mark=diamond*]  table[x expr={\thisrowno{1}}, y expr={\thisrowno{10}}, col sep=comma] {IsoEvaporationDodecaneRS.txt};
    \addlegendentry{$K_\text{dft}$}        
    
    \addplot[mark=star]  table[x expr={\thisrowno{1}}, y expr={\thisrowno{4}}, col sep=comma] {IsoEvaporationDodecaneRS.txt};
    \addlegendentry{measured}     
    
    \end{axis}
  \end{tikzpicture}    
  \caption{Comparison of evaporation front speeds for different initial vapor pressure values $p_\R$. In black, the measured values from \cite{Simdodecane1999}. The colored lines refer to interface speeds of two-phase Riemann solutions.} \label{fig:isop:dodecane:RP}
\end{figure}

The experiment shows stable evaporation fronts until a vapor pressure of $p_\R=\unit[0.7]{bar}$. Figure~\ref{fig:isop:dodecane:RP} shows the measured front speed for different values of $p_\R$. At higher pressure values, there was either no evaporation process starting or they observed a train of bubbles and unstable waves. The first case corresponds to zero transition speed. In the second case no evaporation front could be determined. Our special interest lies on the test cases which led to stable evaporation fronts, i.e.\ the range $\unit[0]{bar} \leq p_\R \leq \unit[0.7]{bar}$, in order to compare front speeds.

The initial conditions for the Riemann problems are $\tau_\L=\tau_\liq^\sat$ ($p_\L=p(\tau_\liq^\sat)\approx\unit[1.39]{bar}$) and different values for $\tau_\R$, such that the vapor pressure varies from $p_\R=\unit[1.37]{bar}$ to almost vacuum. The initial velocity is zero on both sides. 
The thermodynamic properties of n-dodecane are calculated with the library CoolProp \cite{CoolProp}.

Figure~\ref{fig:isop:dodecane:RP} shows the propagation speed in Eulerian coordinates of the evaporation wave for the kinetic relations $K_1$, $K_3$, $K_7$ and $K_\text{dft}$. The constant for $K_3$ is $k^\ast=\unitfrac[0.005]{m^6}{kg^2}$ and the corresponding kinetic functions are monotone decreasing. 
For $K_\text{dft}$, Theorem~\ref{theo:kinrel} is not applicable, however, we checked numerically that the corresponding Lax curves are monotone such that $K_\text{dft}$-Riemann solutions exist uniquely. The kinetic relations under consideration are shown in Figure~\ref{fig:Kcontour:dodecane}. 

We compare the solutions with the shock tube experiments.   
For vapor pressure values from almost vacuum to $\unit[0.4]{bar}$, the measured front speed values, as well as, the speed predicted by the two-phase Riemann solver are constant. For lower pressure values the front speeds are decreasing.
  
The measured front speed is close to zero around $\unit[0.7]{bar}$. 
The propagation speeds, computed via the two-phase Riemann solvers, are decreasing much slower. They reach the value $\sigma=0$ for $p_\R=p^\sat$. That reflects the fact that here only thermodynamic equilibrium solutions are static. A behavior, as in the experiment, would require a kinetic relation, in which the mass flux is zero until a certain threshold is passed. Such a kinetic relation is described in Subsection~\ref{sub:K5}. 
Recall that the authors observed unstable waves and bubbly flows for $p_\R>\unit[0.7]{bar}$. Such flows are not comparable with the solutions of Riemann problems.
  
Let us concentrate again on the range $\unit[0]{bar} \leq p_\R \leq \unit[0.7]{bar}$, where Simoes-Moreira and Shepherd observed stable evaporation fronts.
It is remarkable that the propagation speed values of $K_\text{dft}$-Riemann solutions match the measured vales. Note that there is no parameter that could be tuned. 
The propagation speeds for the kinetic relations $K_1$, $K_3$ and $K_7$ are faster than those of the experiment. The difference reduces, with rising entropy dissipation.
The comparison demonstrates, that for this experiment non-decreasing kinetic functions, e.g.\ $K_\text{dft}$, are necessary to predict the correct propagation speed.

The authors measured also the pressure near the evaporation front. This is used for a second study.
Assume for a moment that the measured values are comparable to the end states at the phase boundary. 
The measured pressure values ($P_\text{bottom}$ and $P_\text{exit}$ in \cite{Simdodecane1999}) are plotted into Figure~\ref{fig:Kcontour:dodecane} with black dots. The dots are far from what we can reach with monotone decreasing kinetic functions. 
A kinetic function, that is fitted to the measured values and the saturation state, would be a non-decreasing function.
Note that the liquid pressure values correspond to the liquid metastable phase and phase boundaries with such end states are generally excluded by monotone decreasing kinetic functions, see Corollary~\ref{col:metastable:waves}.


\section{Application of the two-phase Riemann solvers in interface tracking schemes and verification}\label{sec:numerics}
As mentioned in the introduction, one of the applications of two-phase Riemann solvers are numerical schemes of tracking type. 
Such \defemph{interface tracking schemes} involve a tracking of the phase boundary and the computation of fluxes from the liquid phase to the vapor phase and vice versa. Like in Godunov type schemes, Riemann solvers, i.e.\ mapping \eqref{eq:microm}, are applied at edges which are identified with the phase boundary, in order to compute the interfacial flux.
A \defemph{bulk solver}, e.g.\ a finite volume or discontinuous Galerkin method, is then used to solve the Euler system in the bulk.
We analyze this approach with the scheme described in \cite{ROHZER2014} for one-dimensional and radially symmetric solutions of \eqref{eq:euler}-\eqref{eq:Euler:kinrel}. In the radially symmetric framework it is possible to take into account curvature effects without requiring a complex computation of the curvature. Furthermore, the scheme in \cite{ROHZER2014} is conservative. It bases on a first order finite volume method with local grid adaption at the interface and serves as a test environment for two-phase Riemann solvers.

Section~\ref{sec:solver} provides a constructive algorithm to determine two-phase Riemann solutions for kinetic functions and surface tension. In one space dimension (without surface tension), this is also the exact solution. We are now able to verify the interface tracking approach. This was kept open in \cite{ROHZER2014}, since no exact solution was available. Furthermore, two previously developed (approximate) Riemann solvers will be analyzed: 
the \defemph{Liu (Riemann) solver} from \cite{GODSEG2006,JAEROHZER2012}, see Subsection~\ref{sub:Liu}, and 
an approximate Riemann solver for general kinetic relations \eqref{eq:Euler:kinrel} based on relaxation techniques \cite{ROHZER2014}. We called the latter one \defemph{relaxation $K_n$-(Riemann) solver} if the considered relation is $K_n$. All Riemann solvers are mappings of type \eqref{eq:microm}. In order to distinguish the different two-phase solvers, we call Algorithm~\ref{alg:kinrelsolver} \defemph{(exact) $K_n$-Riemann solver}.

The Riemann solver of Subsection~\ref{sub:solveralgo} is implemented for $K_1$, $K_3$ and $K_7$. 
The relaxation Riemann solver \cite{ROHZER2014} applies kinetic relations directly and is less restrictive. Implementations for $K_1$, $K_2$ and $K_3$ are available.
The Liu solution is considered as an approximate solution of the two-phase Riemann problem, since it applies the modified (homogenized) equation of state \eqref{eq:pLiu}.
Thus, we treat the Liu solver as an approximate solver for kinetic relation $K_7$, cf.\ Subsection~\ref{sub:Liu}.

We refer to the space in Eulerian coordinates and transform the output of the Riemann solver mapping \eqref{eq:microm} to that coordinates.
For the numerical flux computation in the bulk phases, we use the local Lax-Friedrichs flux \cite{LEV2007}. Unless otherwise specified, we apply a CFL-like time step restriction with the $\cfl$ number $0.9$, details are described in \cite{ROHZER2014}. The examples apply either the dimensionless van der Waals pressure of Example~\ref{exp:vdw} or equations of state that are provided by the thermodynamic library CoolProp \cite{CoolProp}.

\subsection{Experimental order of convergence}\label{sub:eoc:exact}

\qquad
We consider radially symmetric solutions $\W=\big(\rho(r,t), m(r,t)\big)^\transp$, $r=\abs{\vv x}$, of the Euler system \eqref{eq:euler} in the domain $\Omega=\set{\vv x\in\setR^d| R_\tmin< \abs{\vv x} < R_\tmax}$ and the initial data
\begin{align}\label{eq:rp:radial}
 \W(r,0)  = \begin{cases}
              \W_\L &:  r \in [ R_{\tmin}, \gamma^0 ),\\
              \W_\R &:  r \in [ \gamma^0, R_{\tmax} ].
            \end{cases}
\end{align}
The states $\W_\L\in\admisr_{\liq/\vap}\times\setR$ and $\W_\R\in\admisr_{\vap/\liq}\times\setR$ are constant and in different phases, where $\admisr_{\liq}$, $\admisr_{\vap}$ are the admissible sets for the density corresponding to $\admis_{\liq}$, $\admis_{\vap}$ in Definition~\ref{def:thermo}. Thus, the phase boundary is initially located at $\gamma^0$.
Note that in one spatial dimension \eqref{eq:rp:radial} defines a Riemann problem. The set $[R_{\tmin}, R_{\tmax}]$ is just an interval for any $R_{\tmin}\in\setR$. 
The domain in the multidimensional case is a disc or a ball with a hole in the center, since $R_{\tmin}>0$. The hole is due to a singularity of the radially symmetric system in $r=0$, see \cite{ROHZER2014}. 
Domain $\Omega\subset\setR^d$ and time interval $[0, \theta]$ are chosen such that the waves originating in $\gamma^0$ do not reach the boundary. 
Furthermore, we use the boundary condition $\W(R_{\tmin},t)=\W_\L$, $\W(R_{\tmax},t)=\W_\R$ for $t\in[0, \theta]$.

With respect to a  \defemph{reference solution} $\hat \W = (\hat \rho, \hat m)^\transp$, we compute the \defemph{relative error}
\begin{align*}
  e_I = \int_0^\theta \int_{R_{\tmin}}^{R_{\tmax}} A_d(r)  
  \left( \frac{\abs{\rho_I-\hat\rho }}{1+\abs{\hat\rho }} + \frac{\abs{m_I-\hat m }}{1+\abs{\hat m }} \right) \dd r \dd t,
\end{align*}
where $(\rho_I,m_I)^\transp$ is the numerical solution on a grid with $I\in\setN$ cells and $A_d(r)$ is the volume of a $d$-dimensional sphere with radius $r > 0$.

For a sequence of grids with $I_l\in\setN$ cells and corresponding relative errors $e_{I_l}$ 
we compute the \defemph{experimental order of convergence}
 $\text{eoc}_l := {  \ln\left(e_{I_{l+1}} \big/ e_{I_l}\right) }/{ \ln\left(I_{l}  \big/ I_{l+1}\right)}$.
The number $I$ is also the degree of freedom for the bulk solver.
The optimal order that can be expected in view of the first-order scheme and solutions, that contain discontinuities is between $0.5$ and $1$, cf.\ \cite{LEV2007}.

\subsubsection{Verification of the interface tracking approach in 1D} \label{sub:eoc:1d:exact}

In one space dimension, the solution of the Riemann problem \eqref{eq:rp:radial} is given by the exact $K_n$-Riemann solver after the transformation to Eulerian coordinates. Thus, it is considered as reference solution. This framework allows us to examine convergence towards the exact solution. Note that this was not possible in \cite{ROHZER2014}, since no exact solutions was available. 
We consider kinetic relation $K_3$ with $k^\ast=\unitfrac[0.005]{m^6}{kg^2}$ and an equation of state for the fluid n-dodecane at $T=\unit[230]{\degree C}$, provided by the library CoolProp \cite{CoolProp}. 

Table~\ref{tab:eoc:exact} shows the experimental order of convergence for the conditions (A) and (B) in Table~\ref{tab:eoc:dodecane:ic}. 
The order is in the expected optimal range in view of a first order scheme.
Here, the initial densities $\rho_\L\in\admisr_\liq$, $\rho_\R\in\admisr_\vap$ are computed, such that the pressure values in column $p_\L$ and column $p_\R$ hold initially. 
Note that such conditions were already used in Subsection~\ref{sub:simoera}.

Figure~\ref{fig:testg} displays the pressure distribution of test case (A) at time $t=\unit[0.8]{ms}$.
It shows the numerical solution for a sequence of refined grids and the exact $K_3$-Riemann solution.
The solution is a composition of a 1-shock wave, an evaporation wave, followed by a 2-shock wave. The phase boundary is tracked sharply and the bulk shock waves are approximated very well. 
Note that this example is more challenging that tests cases for the van der Waals fluid, since the pressure in the liquid phase is much stiffer that in the vapor phase. For instance, one finds for the initial states $p'(\tau_\L)\approx \unitfrac[-10^6]{bar\ kg}{m^3}$ in the liquid phase and $p'(\tau_\R)\approx \unitfrac[-0.65]{bar\,kg}{m^3}$ in the vapor phase. 
Furthermore set $\admis_\liq$ is extreme small compared to the spinodal phase.

\begin{table}
\setlength{\tabcolsep}{1.2\tabcolsep}
\newcommand{\midrules}{\cmidrule(lr){2-6}\cmidrule(l){7-11}}
\begin{tabularx}{\textwidth}{*{11}{l}L}
 \toprule
 	& $p_\L$		& $p_\R$		& $v_\L$& $v_\R$ & $K$ 		& $\theta$ 		& $\gamma^0$ 	& $R_\tmin$		& $R_\tmax$ 	& d &\\ \midrules	
 (A)	&\unit[1.39]{bar}	&\unit[0.4]{bar}	&   0   &   0    & $K_3$	& \unit[0.8]{ms} 	& \unit[0.7]{m}	&\unit[0.0]{m}		&\unit[1.0]{m}	& 1 &\\
 (B)	&\unit[1.39]{bar}	&\unit[1.0]{bar}	&   0   &   0    & $K_3$	& \unit[0.8]{ms} 	& \unit[0.7]{m}	&\unit[0.0]{m}		&\unit[1.0]{m}	& 1 &\\ 
 (C)	&\unit[0.098]{bar}	&\unit[0.13]{bar}	&   0   &   0    & $K_3$	& \unit[0.03]{ms} 	& \unit[0.05]{m}&\unit[0.033]{m}	&\unit[0.07]{m}	& 2 &\\
 (D)	&\unit[0.098]{bar}	&\unit[0.13]{bar}	&   0   &   0    & $K_3$	& \unit[0.03]{ms} 	& \unit[0.1]{m}	&\unit[0.066]{m}	&\unit[0.14]{m}	& 3 &\\
 \bottomrule
\end{tabularx}
\caption{Series of initial conditions for n-dodecane in $d$ spatial dimensions. The parameter for $K_3$ is $k^\ast=\unitfrac[0.005]{m^6}{kg^2}$. }\label{tab:eoc:dodecane:ic}
\end{table}

\begin{table}
\newcommand{\midrules}{\cmidrule(r){1-1}\cmidrule(lr){3-4}\cmidrule(lr){6-7}\cmidrule(lr){9-9}\cmidrule(lr){11-12}\cmidrule(l){14-15}}
\begin{tabularx}{\textwidth}{*{7}lX*{7}l}
 \toprule
&& \multicolumn{2}{l}{Test (A)}	&& \multicolumn{2}{l}{Test (B)} && && \multicolumn{2}{l}{Test (C)}	&& \multicolumn{2}{l}{Test (D)}\\
  $I$ 	&& $e_I$\hspace{3em} & eoc && $e_I$\hspace{3em} & eoc 	&& $I$ 	&& $e_I$\hspace{3em} & eoc && $e_I$\hspace{3em} & eoc \\ \midrules
 500 && 5.9e-04 &      && 8.6e-04 &       &&   200  && 4.0e-07 &      && 2.6e-07 &       \\    
1000 && 3.7e-04 & 0.68 && 4.8e-04 & 0.84  &&   400  && 2.0e-07 & 1.00 && 1.3e-07 & 0.97  \\
2000 && 2.2e-04 & 0.76 && 2.6e-04 & 0.88  &&   800  && 8.5e-08 & 1.23 && 5.8e-08 & 1.21  \\
4000 && 1.2e-04 & 0.88 && 1.4e-04 & 0.91  &&   1600 && 2.7e-08 & 1.65 && 1.8e-08 & 1.67  \\
8000 && 6.2e-05 & 0.91 && 7.9e-05 & 0.81  &&   3200 && 4.0e-09 & 2.77 && 2.3e-09 & 3.01  \\  
 \bottomrule
\end{tabularx}
\caption{Error analysis for the front tracking scheme with the exact $K_3$-Riemann solver, see Subsection~\ref{sub:eoc:1d:exact} and Subsection~\ref{sub:eoc:exact:multidim}. }\label{tab:eoc:exact}
\end{table}

\begin{figure} \centering  
  \begin{tikzpicture}[new spy style]
    \begin{axis}[ x post scale = 1.8, y post scale=0.8,
                  xmin=0, xmax=1,
                  ymin=0.39, ymax=1.6,
                  xlabel={one-dimensional spatial variable $r$ in \unit{m} at $t=\unit[0.05]{s}$}, ylabel={pressure $p$ in \unit{bar}},
                  legend pos=south west,
                  clip marker paths=true,
                  mark size = 0.4,
                  cycle list name=exotic,
                  legend cell align=left,
                  ]
    \addplot  table[x expr={\thisrowno{0}}, y expr={1e-5*\thisrowno{1}}, col sep=comma] {EOCdsolseq_K3_p4000.txt};
    \addlegendentry{$I=50$};                  
    \addplot  table[x expr={\thisrowno{2}}, y expr={1e-5*\thisrowno{3}}, col sep=comma] {EOCdsolseq_K3_p4000.txt};
    \addlegendentry{$I=100$};  
    \addplot  table[x expr={\thisrowno{4}}, y expr={1e-5*\thisrowno{5}}, col sep=comma] {EOCdsolseq_K3_p4000.txt};
    \addlegendentry{$I=200$};  
    \addplot[dashed, no marks, green]  table[x expr={\thisrowno{6}}, y expr={1e-5*\thisrowno{7}}, col sep=comma] {EOCdsolseq_K3_p4000.txt};
    \addlegendentry{$I=500$};   
    \addplot[dashdotted, no marks, red]  table[x expr={\thisrowno{8}}, y expr={1e-5*\thisrowno{9}}, col sep=comma] {EOCdsolseq_K3_p4000.txt};
    \addlegendentry{$I=2000$};   
    \addplot[color=black]  table[x expr={\thisrowno{10}}, y expr={1e-5*\thisrowno{11}}, col sep=comma] {EOCdsolseq_K3_p4000.txt};
    \addlegendentry{exact};    
    
    \coordinate (spypoint)  at (axis cs:0.73,0.9); 
    \coordinate (spyviewer) at (axis cs:0.45,1.0);      
    \end{axis}
    
    \spy[ width=5cm,height=2cm,magnification=5] on (spypoint) in node at (spyviewer);  

  \end{tikzpicture}    
  \caption{Pressure distribution for test case (A) with n-dodecane fluid. In color the numerical solution with $K_3$-Riemann solver and in black the (exact) $K_3$-Riemann solution. } \label{fig:testg}
\end{figure}
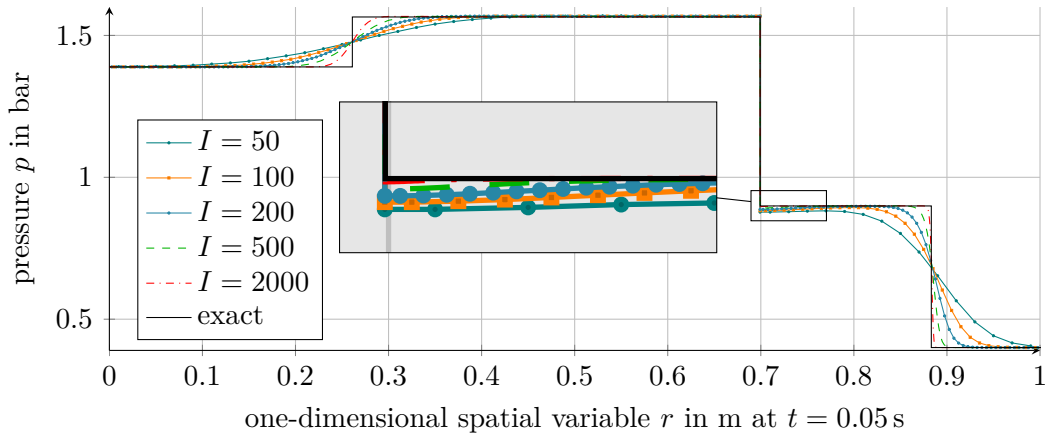

\subsubsection{Verification of the interface tracking approach for radially symmetric solutions}\label{sub:eoc:exact:multidim}

Exact radially symmetric solutions are not available. For simplicity and in order to visualize the wave structure let us use the same initial data \eqref{eq:rp:radial}. 
But here the reference solution is the approximation itself on a fine grid, here $I=6400$ cells. Thus, we are merely able to examine grid convergence.

We consider an n-dodecane bubble in liquid n-dodecane with the initial states (C) and (D) in Table~\ref{tab:eoc:dodecane:ic}.
Test case (C) is considered in $\setR^2$ and (D) is considered in $\setR^3$.
The resulting time step can get very low for small values of $R_\tmin$ due to the CFL condition, see \cite{ROHZER2014}. This limits the size of the computational domain and thus also the diameter of bubbles or droplets. For that reason, we consider quite big bubbles.
The time step for $I=6400$ and $\cfl=0.9$ was in the order of $\unit[10^{-10}]{s}$.

The surface tension at $T=\unit[230]{\degree C}$ is $\surfcoeff=\unitfrac[0.0089]{N}{m}$ (computed with \cite{CoolProp}).
Due to the chosen bubble radii, surface tension does not affect the dynamics in these examples.
Note that the initial pressure values are far from the saturation pressure, here $p^\sat\approx \unit[1.39]{bar}$, and the liquid state is metastable.

Table~\ref{tab:eoc:exact} shows the error and the experimental order of convergence for the kinetic relation $K_3$ in the cases (C), (D).
The computed order varies between $1$ and $3$. 
Figure~\ref{fig:testi} displays the pressure distribution on those grids, which were used for the error analysis. 
This demonstrates that the numerical solution converges with increasing grid resolution towards the finest solution.
Note that plateau values do not form due to the intrinsic geometry change in $r$.
Any fluid movement towards the center accumulates mass, while for flows in direction of the outer boundary mass is distributed over increasing volume units. Thus, the pressure between $r=0.05$ and $r=0.065$ is not constant.

We have already seen in Subsection~\ref{sub:eoc:1d:exact} that for $d=1$, the method converges to the exact solution. Here, we observed grid convergence for a real fluid equations of state. Hence, we expect that the method converges, also in the multidimensional case, towards the exact solution.

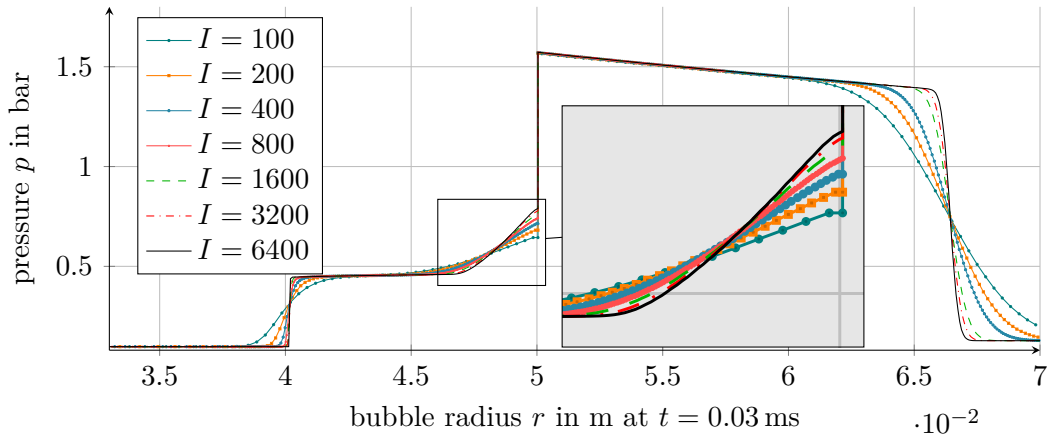
\begin{figure} \centering  
  \begin{tikzpicture}[new spy style]
    \begin{axis}[ x post scale = 1.8, y post scale=0.8,
                  xmin=0.033, xmax=0.07,
                  ymin=0.08, ymax=1.8,
                  xlabel={bubble radius $r$ in \unit{m} at $t=\unit[0.03]{ms}$}, ylabel={pressure $p$ in \unit{bar}},
                  legend pos=north west,
                  clip marker paths=true,
                  mark size = 0.4,
                  cycle list name=exotic,
                  legend cell align=left,
                  ]
    \addplot  table[x expr={\thisrowno{0}}, y expr={1e-5*\thisrowno{1}}, col sep=comma] {EOCdsolseq_radial_K3_dodecane.txt};
    \addlegendentry{$I=100$}                  
    \addplot  table[x expr={\thisrowno{2}}, y expr={1e-5*\thisrowno{3}}, col sep=comma] {EOCdsolseq_radial_K3_dodecane.txt};
    \addlegendentry{$I=200$}  
    \addplot  table[x expr={\thisrowno{4}}, y expr={1e-5*\thisrowno{5}}, col sep=comma] {EOCdsolseq_radial_K3_dodecane.txt};
    \addlegendentry{$I=400$}  
    \addplot  table[x expr={\thisrowno{6}}, y expr={1e-5*\thisrowno{7}}, col sep=comma] {EOCdsolseq_radial_K3_dodecane.txt};
    \addlegendentry{$I=800$}   
    \addplot [dashed, no marks, green] table[x expr={\thisrowno{8}}, y expr={1e-5*\thisrowno{9}}, col sep=comma] {EOCdsolseq_radial_K3_dodecane.txt};
    \addlegendentry{$I=1600$}  
    \addplot [dashdotted, no marks, red] table[x expr={\thisrowno{10}}, y expr={1e-5*\thisrowno{11}}, col sep=comma] {EOCdsolseq_radial_K3_dodecane.txt};
    \addlegendentry{$I=3200$}
    \addplot [color=black] table[x expr={\thisrowno{12}}, y expr={1e-5*\thisrowno{13}}, col sep=comma] {EOCdsolseq_radial_K3_dodecane.txt};
    \addlegendentry{$I=6400$}  
   
    \coordinate (spypoint)  at (axis cs:0.0482,0.62); 
    \coordinate (spyviewer) at (axis cs:0.057,0.7); 

    \end{axis}
    \spy [ width=4cm,height=3.2cm,magnification=2.8] on (spypoint)  in node at (spyviewer);  

  \end{tikzpicture}
 
  \caption{Radial symmetric two-dimensional solution. Pressure distribution of test case (C) and n-dodecane fluid. In color the numerical solution with $K_3$-Riemann solver. 
  The numerical solution for $I=6400$ cells is used as reference solution for the error computation in Table~\ref{tab:eoc:exact}.} \label{fig:testi}
\end{figure}

\subsection{Experimental order of convergence with approximate Riemann solvers}

\subsubsection{Application of the Liu Riemann solver} \label{sub:eoc:liu}

We verify the Riemann solver \cite{JAEROHZER2012} in the framework of the one-dimensional interface tracking scheme. 
The solver is implemented for the van der Waals pressure. 
We compare the numerical solution for kinetic relation $K_7$. 

Table~\ref{tab:eoc:liu} shows the error values and the experimental orders of convergence for increasing grid resolution and the test cases (E)--(G) in Table~\ref{tab:eoc:vdw:ic}.
The initial values of the cases (E) and (F) are in the stable phases. Here, the scheme converges with the expected order. However, for metastable initial values (case (G)) the algorithm converges to a different solution. The error values in that case remain almost constant for decreasing grid sizes.
The reason is the modification of the equation of state between the saturation states, see Subsection~\ref{sub:Liu}.

\begin{table}
{
\setlength{\tabcolsep}{1.5\tabcolsep}
\newcommand{\midrules}{\cmidrule(lr){2-6}\cmidrule(l){7-11}}
\begin{tabularx}{\textwidth}{*{11}l}
\toprule
	& $\tau_\L$	& $\tau_\R$	& $v_\L$	& $v_\R$ 	& $K$	& $\theta$	& $\gamma^0$ 	& $R_\tmin$	& $R_\tmax$ 	& d	\\ \midrules
 (E)	& $0.553$   	& $5.5$		& $1.0$		& $0.0$		& $K_7$	& $0.20$	& $0.5$		& $0$          	& $1$		& 1	\\
 (F)	& $0.500$	& $5.0$		& $0.0$		& $5.0$		& $K_7$	& $0.05$	& $0.5$		& $0$          	& $1$		& 1	\\
 (G)	& $0.557$	& $3.0$		& $0.0$		& $0.0$		& $K_7$	& $0.10$	& $0.5$		& $0$          	& $1$		& 1	\\ \midrules
 (H)	& $0.553$   	& $5.5$		& $1.0$		& $0.0$		& $K_3$	& $0.20$	& $0.5$		& $0$          	& $1$		& 1	\\
 (I)	& $0.530$	& $3.0$		& $0.1$		& $5.0$		& $K_3$	& $0.10$	& $0.5$		& $0$          	& $1$		& 1	\\
 (J)	& $0.557$	& $3.0$		& $0.0$		& $0.0$		& $K_3$	& $0.10$	& $0.5$		& $0$          	& $1$		& 1	\\ 
\bottomrule
\end{tabularx}
}
\caption{Series of initial conditions for a van der Waals fluid in one spatial dimension. The parameter for $K_3$ is $k^\ast=0.2$. }\label{tab:eoc:vdw:ic}
\end{table}

\begin{table}
\newcommand{\midrules}{\cmidrule(r){1-1}\cmidrule(lr){3-4}\cmidrule(lr){6-7}\cmidrule(l){9-10}}
\begin{tabularx}{\textwidth}{*{10}l}
 \toprule
&& \multicolumn{2}{l}{Test (E)}	&& \multicolumn{2}{l}{Test (F)} && \multicolumn{2}{l}{Test (G)} \\
  $I$ 	&& $e_I$\hspace{3em} & eoc && $e_I$\hspace{3em} & eoc && $e_I$\hspace{3em} & eoc\\ \midrules
500  && 3.8e-04 &      && 2.9e-04 &      && 1.0e-04 &      \\
1000 && 2.2e-04 & 0.81 && 1.9e-04 & 0.64 && 9.7e-05 & 0.10 \\
2000 && 1.2e-04 & 0.86 && 1.2e-04 & 0.68 && 9.2e-05 & 0.07 \\
4000 && 6.6e-05 & 0.88 && 7.2e-05 & 0.71 && 8.9e-05 & 0.05 \\
8000 && 3.6e-05 & 0.86 && 4.3e-05 & 0.74 && 8.7e-05 & 0.03 \\  
 \bottomrule
\end{tabularx}
\caption{Error analysis for the method with the Liu Riemann solver, see Subsection~\ref{sub:eoc:liu}. } \label{tab:eoc:liu}
\end{table}

\subsubsection{Application of the relaxation Riemann solver}\label{sub:eoc:relax}

The relaxation solver \cite{ROHZER2014} is implemented for van der Waals fluids and also for external thermodynamic libraries. 
We compare towards the exact $K_3$-Riemann solution.

\begin{example}[Error analysis for van der Waals equations of state]\label{exp:eoc:relax:vdw}
The bulk solver combined with the $K_3$-relaxation Riemann solver and is applied to the test cases (H)--(J) in Table~\ref{tab:eoc:vdw:ic}.
We could not observe decreasing error norms for the time step restriction with $\cfl=0.9$: 
the numerical solution in case (H) seemed to converge towards a different solution, initial conditions of case (I) led to negative values of specific volume and pressure. The numerical solution in case (J) was oscillatory.

Table~\ref{tab:eoc:relax} shows the result for $\cfl=0.1$. The relaxation solver needs apparently more iteration steps to converge. This was already reported in \cite{CHACOQENGROH2012}.
However, the convergence orders are low and decreasing. In particular for case (J) the numerical solution does not converge to the exact solution.

\end{example}

\begin{table}
\newcommand{\midrules}{\cmidrule(r){1-1}\cmidrule(lr){3-4}\cmidrule(lr){6-7}\cmidrule(lr){9-10}\cmidrule(l){12-13}}
\begin{tabularx}{\textwidth}{*{13}l}
 \toprule
&& \multicolumn{2}{l}{Test (H)}	&& \multicolumn{2}{l}{Test (I)} 	&& \multicolumn{2}{l}{Test (J)} &&\multicolumn{2}{l}{Test (B)}\\
  $I$ 	&& $e_I$\hspace{3em} & eoc && $e_I$\hspace{3em} & eoc && $e_I$\hspace{3em} & eoc && $e_I$\hspace{3em} & eoc 	\\ \midrules
500  && 6.3e-04 &      && 3.9e-04 &      && 8.3e-05 &        && 8.8e-04 &      \\
1000 && 4.2e-04 & 0.60 && 2.8e-04 & 0.49 && 7.4e-05 & 0.17   && 6.1e-04 & 0.53 \\
2000 && 2.9e-04 & 0.52 && 2.0e-04 & 0.47 && 6.6e-05 & 0.16   && 4.3e-04 & 0.50 \\
4000 && 2.2e-04 & 0.38 && 1.5e-04 & 0.45 && 6.1e-05 & 0.10   && 3.3e-04 & 0.40 \\  
8000 && 1.9e-04 & 0.23 && 1.1e-04 & 0.41 && 5.8e-05 & 0.07   && 2.8e-04 & 0.24 \\
 \bottomrule
\end{tabularx}
\caption{Error analysis for the method with the relaxation $K_3$-Riemann solver, see Subsection~\ref{sub:eoc:relax}. } \label{tab:eoc:relax}
\end{table}

\begin{example}[Error analysis for n-dodecane equations of state] \label{exp:eoc:relax:dodecane}
For the second example,  we use the test cases of Table~\ref{tab:eoc:dodecane:ic}. The fluid under consideration is n-dodecane. 
We tried several combinations of parameters and CFL numbers but only test case (B) led to a stable result.
Any proper choice of the parameters for the first few iterates, failed at a later time step. 
The problem are negative specific volume values or values in the spinodal phase.

The initial conditions of test case (B) are near the equilibrium solution, here elementary waves are almost negligible and the solution manly consists of a single traveling wave.
Note that this is a simple test case for the relaxation Riemann solver, since the solver was conceived in order to preserve isolated phase boundaries.

The error for test case (B) can be found in Table~\ref{tab:eoc:relax}.
Figure~\ref{fig:testf:relax} displays the solution on different grids and the exact $K_3$-Riemann solution. 
One clearly can see that the numerical solution converges, but to a different solution. 
\end{example}

The examples demonstrate, that the relaxation solver combined with the interface tracking scheme does not converge to the exact solution. We observe grid convergence towards some other solution. 
In previous contributions \cite{CHACOQENGROH2012,ROHZER2014} the relaxation solver was applied only to very specific examples, in particular much simpler equations of state and linear kinetic functions. More complex problems can now be solved with the exact $K_n$-Riemann solvers.

\begin{figure} \centering  
  \begin{tikzpicture}[new spy style]
    \begin{axis}[ x post scale = 1.8, y post scale=0.8,
                  xmin=0, xmax=1,
                  ymin=0.99, ymax=1.45,
                  xlabel={one-dimensional spatial variable $r$ in \unit{m} at $t=\unit[0.05]{s}$}, ylabel={pressure $p$ in \unit{bar}},
                  legend pos=south west,
                  clip marker paths=true,
                  mark size = 0.4,
                  cycle list name=exotic,
                  legend cell align=left,
                  ]
    \addplot  table[x expr={\thisrowno{0}}, y expr={1e-5*\thisrowno{1}}, col sep=comma] {EOCdsolseq_K3_p10000_Relax.txt};
    \addlegendentry{$I=50$}                  
    \addplot  table[x expr={\thisrowno{2}}, y expr={1e-5*\thisrowno{3}}, col sep=comma] {EOCdsolseq_K3_p10000_Relax.txt};
    \addlegendentry{$I=100$}  
    \addplot  table[x expr={\thisrowno{4}}, y expr={1e-5*\thisrowno{5}}, col sep=comma] {EOCdsolseq_K3_p10000_Relax.txt};
    \addlegendentry{$I=200$}  
    \addplot[dashed, no marks, green]  table[x expr={\thisrowno{6}}, y expr={1e-5*\thisrowno{7}}, col sep=comma] {EOCdsolseq_K3_p10000_Relax.txt};
    \addlegendentry{$I=500$}   
    \addplot[dashdotted, no marks, red]  table[x expr={\thisrowno{8}}, y expr={1e-5*\thisrowno{9}}, col sep=comma] {EOCdsolseq_K3_p10000_Relax.txt};
    \addlegendentry{$I=2000$}   
    \addplot[color=black]  table[x expr={\thisrowno{10}}, y expr={1e-5*\thisrowno{11}}, col sep=comma] {EOCdsolseq_K3_p10000_Relax.txt};
    \addlegendentry{exact}    
    \end{axis}
    
    \spy[ width=4cm,height=3.5cm,magnification=2.9] on (8.6,3.75) in node at (5.6,2);  

  \end{tikzpicture}    
  \caption{Pressure distribution for test case (B) with n-dodecane fluid. In color the numerical solution with relaxation $K_3$-Riemann solver and in black the (exact) $K_3$-Riemann solution.} \label{fig:testf:relax}
\end{figure}
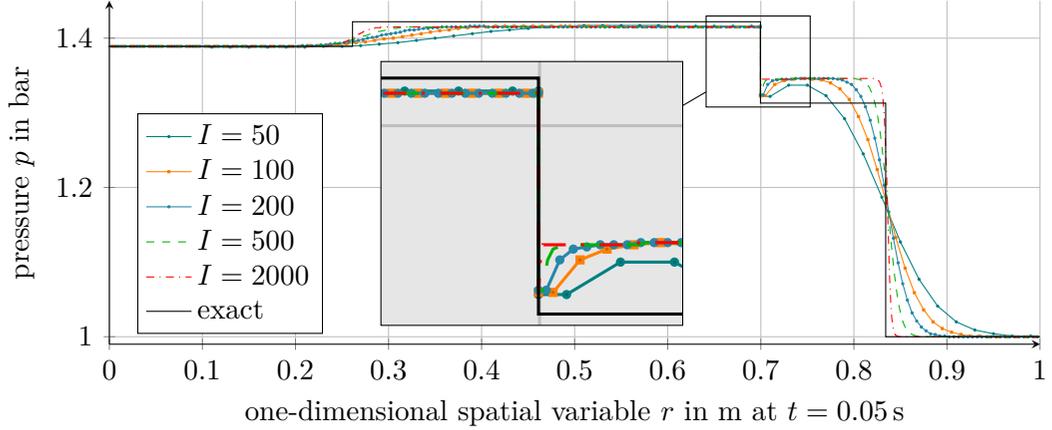


\subsection{Global entropy release and steady state solutions}\label{sub:globalentropy}
A transient solution should reach its steady state $\W(\vv x,t) \to \W^\infty(\vv x) \in (\admisr_\liq\cup\admisr_\vap)\times\setR$ for $t \to\infty$ and
at the same time $\Gamma(t)\to\Gamma^\infty\subset\setR^d$ and $\Omega_{\liq/\vap}(t)\to\Omega_{\liq/\vap}^\infty\subset\setR^d$.
Furthermore, the steady state should be the minimizer of the associated mathematical entropy.
For reflecting boundary conditions, the mathematical entropy at time $t$ is given by
\begin{align*}
\mathcal{E}(\rho(\cdot,t),\vv{m}(\cdot,t))= 
\intdomain{ {  \rho(\vv x,t)\,\psi\left( \frac{1}{\rho(\vv x,t)}\right) +   \frac{ \abs{\vv{m}(\vv x,t)}^2}{2\,\rho(\vv x,t)} }}
  + \surfcoeff \,\abs{\Gamma(t)}.
\end{align*}
Gurtin has demonstrated in \cite{GUR1985} that the minimum $\mathcal{E}^\infty := \min\{\mathcal{E}(\rho^\infty,\vv{m}^\infty) | \int_{\Omega} \rho^\infty \dd \vv x = \int_{\Omega} \rho_0 \dd \vv x  \}$ is determined by the global thermodynamic equilibrium. 
Moreover, the minimizer corresponds to a single spherical droplet or bubble, cf.\ \cite{GUR1985}. 
Thus, we expect that $\Gamma^\infty$ is a sphere with some radius $\gamma^\infty >0$ and
\begin{align*}
 \rho^\infty(\vv x) = \begin{cases}
                       1/\tau_\liq^\sat &\text{for } \vv x \in \Omega_\liq^\infty,\\
                       1/\tau_\vap^\sat &\text{for } \vv x \in \Omega_\vap^\infty,
                      \end{cases}
 &&                   
 \vv m^\infty(\vv x) = \vv 0.
\end{align*}
Note that saturation states $\tau_{\liq/\vap}^\sat=\tau_{\liq/\vap}^\sat(\surf^\infty)$ exist uniquely, since for spherical bubbles $\surf^\infty:=(d-1)\,\surfcoeff/\gamma^\infty$ is constant. The same holds for spherical droplets, with $\surf^\infty:=-(d-1)\,\surfcoeff/\gamma^\infty$.

We consider a van der Waals fluid with $\surfcoeff=0.01$ and radially symmetric solutions in $\Omega = \set{\vv{x}\in\setR^2 | 0.005<\abs{\vv{x}}<2}$.
The phase boundary is initially located at $\Gamma(0)=\setS$.
The saturation states of a droplet with radius $1$ are 
$\tau_\liq^\sat(0.01) \approx 0.55444$, $\tau_\vap^\sat(0.01) \approx 3.15$.
Initial condition
\begin{align*}
  \begin{pmatrix}
   \rho \\ \vv v
  \end{pmatrix} (\vv x, 0) &=
  \begin{cases}
    (1/\tau_\liq^\sat ,\hphantom{-}0.05)^\transp &\text{for } \abs{\vv x} \in [0.005,1],\\
    (1/\tau_\vap^\sat ,           -0.05)^\transp &\text{for } \abs{\vv x} \in (1,2]
  \end{cases} 
\end{align*}
and boundary condition $  \vv v \cdot \vv n = 0$ at $\partial \Omega$
are such that, right from the beginning, waves are emitted and reflected from the boundary. 
The initial condition satisfies $\rho(\vv x,0) = \rho^\infty(\vv x)$, such that potential energy and surface energy are initially at the global minimum, while the total kinetic energy is positive. As time passes, waves slop ahead and back within some density range around the saturation solution and with decreasing amplitudes.

We compare the exact and approximate Riemann solvers. We will find, that only the newly developed exact $K_n$-Riemann solvers lead to monotone energy decay. 

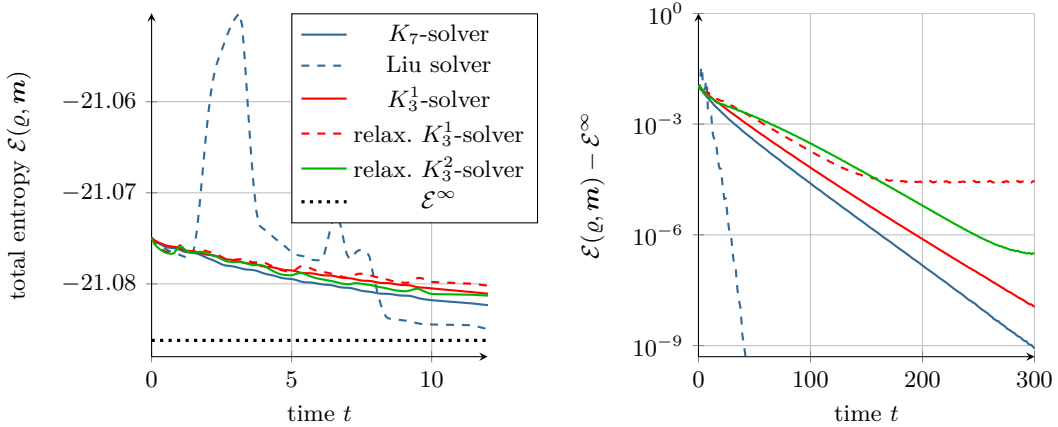
\begin{figure}[tb] \centering  
  \begin{tikzpicture}[new spy style]
    \begin{axis}[ x post scale = 0.7, y post scale=0.8,
                  xmin=0, xmax=12,
                  ymin=-21.088,
                  xlabel={time $t$}, ylabel={total entropy $\mathcal{E}(\rho,\vv{m})$},
                  legend style={at={(1.15,1)}},
                  ylabel style = {yshift=15pt},
                  ]
    \addplot [thick, solid, blue] table[x expr={\thisrowno{0}}, y expr={\thisrowno{1}}, col sep=comma] {TransientSolution.txt};
    \addlegendentry{$K_7$-solver}; 

    \addplot [thick, dashed, blue] table[x expr={\thisrowno{0}}, y expr={\thisrowno{2}}, col sep=comma] {TransientSolution.txt};
    \addlegendentry{Liu solver}; 
    
    \addplot [thick, solid, red] table[x expr={\thisrowno{0}}, y expr={\thisrowno{3}}, col sep=comma] {TransientSolution.txt};
    \addlegendentry{$K_3^1$-solver}; 

    \addplot [thick, dashed, red] table[x expr={\thisrowno{0}}, y expr={\thisrowno{4}}, col sep=comma] {TransientSolution.txt};
    \addlegendentry{relax.~$K_3^1$-solver}; 

    \addplot [thick, green] table[x expr={\thisrowno{0}}, y expr={\thisrowno{5}}, col sep=comma] {TransientSolution.txt};
    \addlegendentry{relax.~$K_3^2$-solver}; 
    
    \addplot [very thick, dotted, black] coordinates { (0, -21.086209) (12, -21.086209) };
    \addlegendentry{$\mathcal{E}^\infty$} 
    \end{axis}
  \end{tikzpicture} 
   \hspace{0.5\hfloatsep}
  \begin{tikzpicture}[new spy style]
    \begin{semilogyaxis}[log basis y=10,
                  x post scale = 0.7, y post scale=0.8,
                  ymin=5e-10, ymax=1,
                  xlabel={time $t$}, ylabel={$\mathcal{E}(\rho,\vv{m})-\mathcal{E}^\infty$},
                  legend pos=north east,
                  ylabel style = {yshift=5pt},
                  ]
    \addplot [thick, solid, blue] table[x expr={\thisrowno{0}}, y expr={\thisrowno{6}}, col sep=comma] {TransientSolution.txt};

    \addplot [thick, dashed, blue] table[x expr={\thisrowno{0}}, y expr={\thisrowno{7}}, col sep=comma] {TransientSolution.txt};
    
    \addplot [thick, solid, red] table[x expr={\thisrowno{0}}, y expr={\thisrowno{8}}, col sep=comma] {TransientSolution.txt};

    \addplot [thick, dashed, red] table[x expr={\thisrowno{0}}, y expr={\thisrowno{9}}, col sep=comma] {TransientSolution.txt};

    \addplot [thick, green] table[x expr={\thisrowno{0}}, y expr={\thisrowno{10}}, col sep=comma] {TransientSolution.txt};
    \end{semilogyaxis}
  \end{tikzpicture}     
  \caption{Evolution of the total mathematical entropy in time. Both figures correspond to the same legend.} \label{fig:transient}
\end{figure}

\begin{example}[Entropy release applying the Liu Riemann solver]\label{exp:globalentropy:Liu}
The numerical results in Figure~\ref{fig:transient} are performed for the bulk solver on a grid with $I=100$ cells, combined with the Riemann solvers. 
  Figure~\ref{fig:transient} shows the evolution of the total entropy $t\to\mathcal{E}(\rho,\vv{m})$ (left) and the shifted total entropy $t\mapsto \mathcal{E}(\rho,\vv{m})-\mathcal{E}^\infty$ (right) in order to use a logarithmic scale.
  The steady state solution is given by above saturation states. We find $\mathcal{E}^\infty \approx -21.08621$, where the contribution of the surface energy is $\surfcoeff\,\abs{\Gamma^\infty}=0.02\,\pi$.

  Observe that the Liu solver leads to an increase in the total entropy at the beginning of the simulation time. For $t>8$, the entropy decays very fast compared to the result obtained with the $K_7$-Riemann solver. This strange behavior is due to the fact that the Liu solver applies a different pressure function as the bulk solver.
  Note that the initial states were chosen, such that the bulk solution varies around the saturation states. Thus, initial states for the Riemann solvers are very often in the metastable phases, where the pressure functions actually are different.
\end{example}

\begin{example}[Entropy release applying the relaxation Riemann solver]\label{exp:globalentropy:relax}

  For the relaxation solver with $K_3^1$ and $k^\ast=0.2$, one observes in Figure~\ref{fig:transient} that the method converge to the stationary solution up to a difference of $10^{-5}$. For $t>150$, the numerical solution behaves unstable and $\mathcal{E}$ remains on a constant level.
  The entropy decay is not completely monotone, furthermore a CFL number of $0.01$ was necessary. For $\cfl=0.5$ and $\cfl=0.1$, the final difference to the stationary solution was around $10^{-2}$. 
  Decreasing the $\cfl$ number once more (not shown in the figure) or using a higher dissipation rate, i.e.\ $K_3^2$ with $k^\ast=2$, pushes the final difference below $10^{-6}$.

\end{example}

\begin{example}[Entropy release applying the exact Riemann solver]\label{exp:globalentropy:exact}

  The numerical results for the interface tracking scheme combined with the exact Riemann solvers are convincing.
  Figure~\ref{fig:transient} shows strictly monotone decreasing values of total mathematical entropy towards the expected limit $\mathcal{E}^\infty$.
  Although surface tension is entirely handled on the Riemann solver level, the method is capable to predict the global contribution of the surface energy.
  The decay rate for $K_7$ is higher than for $K_3^1$. Note that kinetic relation $K_7$ dissipates more entropy than $K_3^1$ with $k^\ast=0.2$.
  This indicates that increasing the interfacial entropy dissipation has a damping effect. 

\end{example}

\subsection{Condensation of bubbles}\label{sub:condensbubble}
We consider spherical bubbles in the domain $\Omega = \{\vv{x}\in\setR^2 | \unit[0.5]{mm}<\abs{\vv{x}}<\unit[20]{mm}\}$ with initial and boundary conditions such that the bubbles vanish.
More precisely, we compare the evolution of the phase boundary until it approaches the inner boundary. The test is performed for equations of state of the fluids n-dodecane at \unit[230]{\degree C}, butane at \unit[20]{\degree C}, acetone at \unit[20]{\degree C}, water at \unit[80]{\degree C} and different kinetic relations. The fluid n-dodecane was already used in former test cases, the other fluids are just randomly selected. Note that Algorithm~\ref{alg:kinrelsolver} does not rely on a specific equation of state and enables to compare diverse fluids and kinetic relations.

The setting is as follows. We compute the saturation pressure $p^\sat$ (with $\surf=0$) for each fluid and apply initial density values such that the vapor pressure is $0.4\, p^\sat$ and the liquid pressure is $4\,p^\sat$. The initial fluid velocity is zero and the bubble radius is $\gamma^0=\unit[10]{mm}$.
Waves at the inner boundary are reflected. At the outer boundary, we apply a Dirichlet condition for the density to keep the pressure constant. 
The fixed pressure at the outer boundary guarantees that the bubble vanishes.

We use the interface tracking scheme with the exact two-phase Riemann solver of Algorithm~\ref{alg:kinrelsolver} for $I=100$ cells and $\cfl=0.9$. 
The evolution of the bubble radii, see Figure~\ref{fig:evap:bubble}, depends on the selected fluid and the kinetic relation. We do not want to classify that correlation. But, as expected, all bubbles vanish for the selected boundary condition. For higher entropy dissipation (kinetic relation $K_7$) the vapor liquefies faster. The difference is low for butane and n-dodecane but still visible. Once more, we see that increasing the interfacial entropy dissipation has a damping effect. 

The radius is not always monotone decreasing, see the example of acetone with $K_1$. At $t=\unit[1.2]{ms}$ the radius is increasing. This is an effect of the bulk dynamics, but we were wondering if it is influenced by curvature effects or the volume change towards the center. The same setting with $\surfcoeff=0$ (circles in Figure~\ref{fig:evap:bubble}) shows that surface tension is too low to affect the evolution. The behavior in the one-dimensional setting (denoted by triangles) is different. The radius decreases monotone but slower.

Let us remark, that nucleation of bubbles is not taken into account. However, we observe waves of high amplitudes and negative pressure values in the liquid shortly after the bubbles collapsed. Negative pressure values may indicate the nucleation of a new vapor phase. 
The effect of surface tension was not visible in the examples, since the curvature is too low. The simulation of smaller bubbles require a different bulk solver. The time step in this experiment was between $\unit[10^{-10}]{s}$ and $\unit[10^{-9}]{s}$, independent of the fluid. However, the simulation of the water test cases took much longer, the evaluation of the associated equations of state is apparently more expensive.

\begin{figure}[tb] \centering  
  \begin{tikzpicture}[new spy style]
    \begin{axis}[ x post scale = 1.8, y post scale=0.8,
                  xmin=0, xmax=2.5,
                  xlabel={time $t$ in \unit{ms}}, ylabel={radius $\gamma$ in \unit{mm}},
                  legend pos=south east,
                  legend cell align=left,
                  ]
    \addplot [thick, red] table[x expr={1e3*\thisrowno{10}}, y expr={1e3*\thisrowno{11}}, col sep=comma] {CondensBubble.txt};
    \addlegendentry{n-dodecane, $K_1$}
    \addplot [thick, dashed, red] table[x expr={1e3*\thisrowno{8}}, y expr={1e3*\thisrowno{9}}, col sep=comma] {CondensBubble.txt};
    \addlegendentry{n-dodecane, $K_7$} 
                  
    \addplot [thick, blue] table[x expr={1e3*\thisrowno{2}}, y expr={1e3*\thisrowno{3}}, col sep=comma] {CondensBubble.txt};
    \addlegendentry{water, $K_1$}
    \addplot [thick, dashed, blue] table[x expr={1e3*\thisrowno{0}}, y expr={1e3*\thisrowno{1}}, col sep=comma] {CondensBubble.txt};
    \addlegendentry{water, $K_7$}                  

    \addplot [thick, green] table[x expr={1e3*\thisrowno{6}}, y expr={1e3*\thisrowno{7}}, col sep=comma] {CondensBubble.txt};
    \addlegendentry{butane, $K_1$}                     
    \addplot [thick, dashed, green] table[x expr={1e3*\thisrowno{4}}, y expr={1e3*\thisrowno{5}}, col sep=comma] {CondensBubble.txt};
    \addlegendentry{butane, $K_7$}    
    
    \addplot [thick] table[x expr={1e3*\thisrowno{14}}, y expr={1e3*\thisrowno{15}}, col sep=comma] {CondensBubble.txt};
    \addlegendentry{acetone, $K_1$}
    \addplot [thick, dashed] table[x expr={1e3*\thisrowno{12}}, y expr={1e3*\thisrowno{13}}, col sep=comma] {CondensBubble.txt};
    \addlegendentry{acetone, $K_7$}     
    \addplot [mark=o, only marks, each nth point=7] table[x expr={1e3*\thisrowno{16}}, y expr={1e3*\thisrowno{17}}, col sep=comma] {CondensBubble.txt};
    \addlegendentry{acetone, $K_1$, $\surfcoeff\!=\!0$}     
    \addplot [mark=triangle, only marks, each nth point=7] table[x expr={1e3*\thisrowno{18}}, y expr={1e3*\thisrowno{19}}, col sep=comma] {CondensBubble.txt};
    \addlegendentry{acetone, $K_1$, $d=1$}       
    \end{axis}

  \end{tikzpicture}    
  \caption{Time evolution of the radii $\gamma(t)$ of vapor bubbles in different fluids and for different kinetic relations.} \label{fig:evap:bubble}
\end{figure}
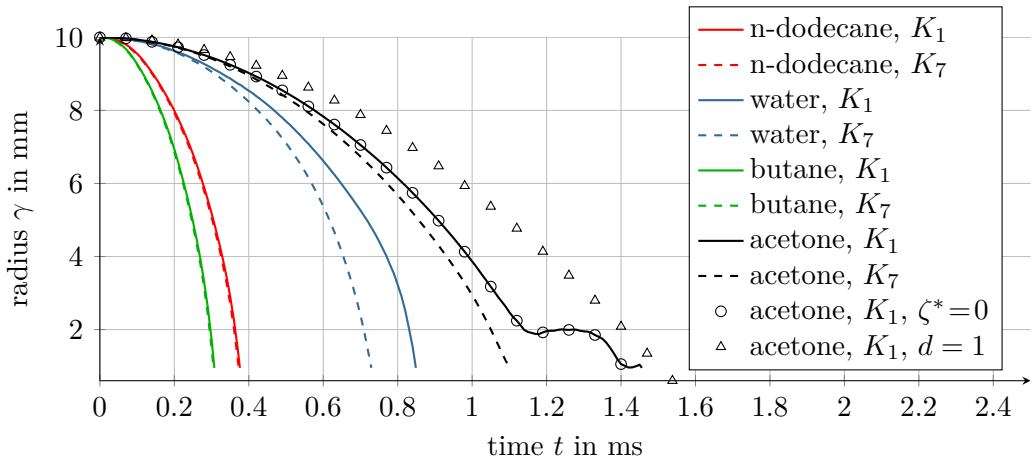

\end{document}